\documentclass[a4paper,cleveref,autoref,USenglish,authorcolumns]{classes/lipics-v2021}
\nolinenumbers
\hideLIPIcs

\author{Sam Coy}{University of Warwick, UK}{S.Coy@warwick.ac.uk}{https://orcid.org/0000-0001-8500-8690}{Research supported in part by the Centre for Discrete Mathematics and its Applications (DIMAP) and an EPSRC studentship.}

    \author{Artur Czumaj}{University of Warwick, UK}{A.Czumaj@warwick.ac.uk}{https://orcid.org/0000-0002-7743-438X}{Research supported in part by the Centre for Discrete Mathematics and its Applications (DIMAP), EPSRC award EP/V01305X/1, and the Simons Foundation Award No. 663281 granted to the Institute of Mathematics of the Polish Academy of Sciences for the years 2021--2023.}
    \author{Christian Scheideler}{University of Paderborn, Germany}{scheideler@upb.de}{https://orcid.org/0000-0002-5278-528X}{This work was partially supported by the German Research Foundation (DFG) within the Collaborative Research Centre \emph{On-The-Fly Computing (GZ: SFB 901/3)} under the project number 160364472}
    \author{Philipp Schneider}{University of Freiburg, Germany}{philipp.schneider@cs.uni-freiburg.de}{https://orcid.org/0000-0001-9660-1270}{}
    \author{Julian Werthmann}{University of Paderborn, Germany}{jwerth@mail.upb.de}{https://orcid.org/0000-0002-5110-5625}{This work was partially supported by the German Research Foundation (DFG) within the Collaborative Research Centre \emph{On-The-Fly Computing (GZ: SFB 901/3)} under the project number 160364472}
    \author{\ }{\ }{}{}{}
    \author{\ }{\ }{}{}{}

\authorrunning{S. Coy \and A. Czumaj \and C. Scheideler \and P. Schneider \and J. Werthmann}
\Copyright{Sam Coy and Artur Czumaj and Christian Scheideler and Philipp Schneider and Julian Werthmann}

\acknowledgements{We would like to thank Martijn Struijs for valuable discussions concerning the geometric insights of the paper. Any errors remain our own.}



\usepackage[utf8]{inputenc}
\usepackage[T1]{fontenc}
\usepackage{lmodern}
\usepackage{amsmath, amssymb, thmtools, amsfonts}
\usepackage{bbm}
\usepackage{appendix}
\usepackage{graphicx}
\usepackage{color}
\usepackage{algorithm}
\usepackage[noend]{algpseudocode}

\usepackage{xspace}
\usepackage{mathtools}
\usepackage{booktabs}
\usepackage{enumerate}
\usepackage{enumitem}
\usepackage{ulem}

\usepackage{tikz}
\usepackage{pgfplots}
\pgfplotsset{compat=1.17}
\usetikzlibrary{math,calc,patterns,intersections,pgfplots.fillbetween,hobby}
\usepackage{caption}
\usepackage{subcaption}
\usepackage{tabularx}
\usepackage{tablefootnote}
\usepackage[parfill]{parskip}
\usepackage{float}

\newcommand{\jwi}[1]{\todo[inline, color=yellow!50]{Julian: #1}}
\newcommand{\jw}[1]{\todo[color=yellow!50]{\tiny Julian: #1}}

\newcommand{\ps}[1]{\todo[color=blue!40]{\tiny Phil: #1}}

\newcommand{\junk}[1]{}

\newcommand{\LOCAL}{\ensuremath{\mathsf{LOCAL}}\xspace}
\newcommand{\CONGEST}{\ensuremath{\mathsf{CONGEST}}\xspace}
\newcommand{\NCC}{\ensuremath{\mathsf{NCC}}\xspace}
\newcommand{\NCCzero}{\ensuremath{\mathsf{NCC_0}}\xspace}
\newcommand{\HYBRID}{\ensuremath{\mathsf{HYBRID}}\xspace}

\newcommand{\bigO}{\smash{\ensuremath{O}}}
\newcommand{\tilO}{\smash{\ensuremath{\widetilde{O}}}}
\newcommand{\tilOm}{\smash{\ensuremath{\widetilde{\Omega}}}}

\newcommand{\eps}{\varepsilon}

\newcommand{\calH}{\mathcal{H}}
\newcommand{\calR}{\mathcal{R}}
\newcommand{\calP}{\mathcal{P}}

\newcommand{\p}{\!+\!}
\newcommand{\m}{\!-\!}

\newcommand{\holes}{\mathcal{H}}
\newcommand{\regions}{\mathcal{R}}
\newcommand{\landmarks}{\ensuremath{V_\Lambda}}

\DeclareMathOperator{\polylog}{polylog}

\DeclareMathOperator*{\argmin}{arg\,min}
\DeclareMathOperator{\dist}{\,dist}
\DeclareMathOperator{\hop}{\,hop}

\newcommand{\convex}{path-convex\xspace}
\newcommand{\code}[1]{\texttt{#1}}

\definecolor{starship}{rgb}{1.0, 0.75, 0.0}
\definecolor{deepcerulean}{RGB}{  0, 114, 178}
\definecolor{tenn}        {RGB}{255,  0,   0}
\definecolor{seagreen}{rgb}{0.18, 0.55, 0.34}

\title{Routing Schemes for Hybrid Communication Networks}

\usepackage{amsthm}
\usepackage[disable]{todonotes}
\newtheorem{fact}[theorem]{Fact}

\begin{document}

\maketitle

\vspace{-2.5em}

\begin{abstract}
    Hybrid communication networks provide multiple modes of communication with varying characteristics. The $\mathsf{HYBRID}$ model was introduced to allow theoretical study of such networks. It combines a \emph{local} mode, which restricts communication among nodes to a given graph and a \emph{global} mode where any two nodes may communicate in principle, but only very little such communication can take place per unit of time. For an example consider an ad-hoc wifi network among mobile devices supplemented with comparatively limited communication via the cellular network.\\
    We consider the problem of the computation of \emph{routing schemes}, where nodes have to compute small labels and routing tables that allow for efficient routing of messages in the local network, which typically offers the majority of the throughput. Recent work has shown that using the $\mathsf{HYBRID}$ model admits a significant speed-up compared to what would be possible if either communication mode were used in isolation. Nonetheless, if general graphs are used as local graph the computation of routing schemes still takes polynomial rounds in the $\mathsf{HYBRID}$ model.\\    
    We bypass this lower bound by restricting the local graph to unit-disc-graphs with few ``radio holes'' and solve the problem deterministically with running time, label size, and size of routing tables all in $O(|\mathcal H|^2 \!+\! \log n)$ where $|\mathcal H|$ is the number of such radio holes. Our work builds on that of [Coy et al., OPODIS'21], which obtains this result if no radio holes are present thus making the graph much simpler (topologically similar to a tree). We develop new techniques to achieve this, a decomposition of the local graph into path-convex regions, where each region contains a shortest path for any pair of nodes in it. We continue to investigate properties of shortest paths algorithms on grid graphs, which can be utilized as simpler representations of unit-disc-graphs.
\end{abstract}

\jwi{We need to remember to add an acknowledgement for Martijn}
\section{Introduction}\label{sec:introduction}


The \HYBRID model was introduced in \cite{Augustine2020} as a means to study distributed systems which leverage multiple communication modes of different characteristics. Of particular interest are networks that combine a \emph{local} communication mode, which has a large bandwidth but is restricted to edges of a \emph{graph} on the nodes of the network, with a \emph{global} communication mode where any two nodes may communicate in principle, but very little such communication can take place per unit of time. This concept captures various real distributed systems, notably networks of cellphones that combine high bandwidth but locally restricted wireless communication on a \emph{unit-disc-graph} with data transmission via the cellular network.

\emph{Routing Schemes} are one of the most fundamental distributed data structures, most prominently employed in the Internet, and are used to forward packets among connected nodes in a network in order to facilitate data exchange between any pairs of nodes. In the distributed variant of the problem the nodes initially only know their incident neighbors in the network and need to communicate as efficiently as possible using their available means of communication such that subsequently each node knows its label and a routing table with the following properties. Given a packet with the label of the receiver node in the header, any node must be able to forward this packet in the network using the label and its routing table such that the packet eventually reaches the intended receiver. Algorithms for routing schemes in hybrid networks are of increasing importance, as contemporary communication standards support such settings, one prominent example being the 5G standard \cite{Asadi2016}. Formally we define routing schemes as follows.

\begin{definition}[Routing Schemes]\label{def:routing_schemes}
    A routing scheme on a connected graph $G = (V,E)$ consists of labels $\lambda(v)$
    and routing functions (aka routing table) $\rho_v$ for each $v \in V$. $\rho_v$ maps labels to neighbors of $v$ in $G$, such that the following holds. Let $s,t \in V$. Let $v_0 = s$ and $v_{i+1} = \rho_{v_{i}}(\lambda(t))$ for $i \geq 1$. Then there is an $\ell \in \mathbb N$, such that $v_\ell = t$. 
    A routing scheme is an approximation with \emph{stretch} $\alpha$ if $\ell_{st} \leq \alpha \cdot \hop(s,t)$ for all $s,t \in V$, where $\hop(s,t)$ is smallest number of edges of any $st$-path, and $\ell_{st}$ is the length of the induced routing path from $s$ to $t$.\footnote{Note that minimizing hop-distance in a unit-disc-graph essentially minimizes the Euclidean distance that the path covers, thus graph weights are not required. This is due to the fact that there is always a shortest path $\Pi$ that covers distance at least 1 with every two hops. Thus, if $d_G(\Pi)$ is the distance in the UDG measured by the euclidean length of its edges then $d_G(\Pi) \leq |\Pi| \leq 2d_G(\Pi)$.}
\end{definition}

Since typically large amounts of packets are exchanged between senders and receivers as part of simultaneously ongoing sessions we concentrate on routing schemes for the {local network graph}, which offers much larger throughput that than what is possible on the global network, since the latter either involves higher costs or is more restricted as infrastructure is shared, like communicating via the cellular network (however, we need very little local communication to actually compute the routing scheme).

Our first goal is to optimize the round complexity of computing such a routing scheme, which is important since frequent changes in the topology of a local network among mobile devices necessitates its fast re-computation. The second goal is to minimize the size of the labels and local routing tables 
as these must be shared in advance (e.g.~via the global network) to initiate a session between two nodes. The third goal is to minimize the \textit{stretch} of the routing path between sender and receiver, minimizing latency and alleviating congestion.

We consider the above problem in the \HYBRID model of computing that has received increasing attention during the last few years \cite{Augustine2020,Kuhn2020,Feldmann2020,CensorHillel2021,CensorHillel2021a,Anagnostides2021,Kuhn2022,Coy2022}. 
Formally, the \HYBRID model builds on the classic principle of synchronous message passing:

\begin{definition}[Synchronous message passing, cf.\ \cite{Lynch1996}]
     We have $n$ computational nodes with some initial state and unique identifiers (IDs) in $[n]:=\{1, \dots, n\}$. Time is slotted into discrete rounds. In each round, nodes receive messages from the previous round; they perform (unlimited) computation based on their internal states and the messages they received so far; and finally, based on those computations, they send messages to other nodes in the network. 
\end{definition}

Note that the synchronous message passing model focuses on the analysis of round complexity of a distributed problem (the number of rounds required to solve it).
The \HYBRID model restricts which nodes may communicate in a given round and to what extent.
    
\begin{definition}[\HYBRID model, cf.\ \cite{Augustine2020}]
     The \emph{local} communication mode is modeled as a connected graph, in which each node is initially aware of its neighbors and is allowed to send a message of size $\lambda$ bits to each neighbor in each round. In the \emph{global} communication mode, each round each node may send or receive $\gamma$ bits to/from every other node that can be addressed with its ID in $[n]$ in case it is known. If any restrictions are violated in a given round, we assume a strong adversary selects the messages that are dropped.
\end{definition}

In this paper, we consider a weak form of the \HYBRID model, which sets $\lambda \in \bigO(\log n)$ and $\gamma \in \bigO(\log^2 n)$, which corresponds to the combination of the classic distributed models \CONGEST\footnote{Some previous papers that consider hybrid models use $\lambda = \infty$, i.e., the \LOCAL model as local mode.} as local mode, and $\mathsf{NODE \,\, CAPACITATED \,\, CLIQUE}$ 
(\NCC)\footnote{Our approach works for the stricter \NCCzero model where only incident nodes in the local network and those that have been introduced can communicate globally.} 
as global mode.

The distributed problem of computing routing schemes {on the local communication graph} is an excellent fit for the \HYBRID model, since the problem is known to require $\tilOm(n)$ rounds of communication (where $\tilO, \tilOm$ hides $\polylog n$ factors) if only \emph{either} communication via the local mode \emph{or} the global mode is permitted, see \cite{Kuhn2022}. The lower bound for the local communication mode holds even if the input graph is a path and even for unbounded local communication. It is natural to wonder if adding a modest amount of global communication on top of a local network significantly improves the required number of communication rounds to establish a routing scheme in the network. This question was recently answered positively by \cite{Kuhn2022}, where it was shown that routing schemes with small labels can be computed in $\tilO(n^{1/3})$ rounds for \emph{arbitrary} local graphs. However, \cite{Kuhn2022} also shows that a polynomial number of rounds is required to solve the problem even approximately (in particular, an exact solution with labels up to size $\bigO(n^{2/3})$ requires $\tilOm(n^{1/3})$ rounds).

To mitigate this lower bound, \cite{Coy2022} considers local communication networks that are restricted to certain interesting classes of graphs for which they can compute routing schemes in just $\bigO(\log n)$ rounds.
In this article we will continue this line of work and consider local communication graphs that are unit-disk graphs (UDGs). Such a UDG $G=(V,E)$ satisfies the property that nodes are embedded in the Euclidean plane and are connected iff they are at distance at most 1 (see Def.~\ref{def:udg}).
Note that UDGs have been extensively studied as a model capturing how multiple devices using wireless ad-hoc connections communicate (see, e.g., \cite{BoseMSU01,CastenowKS20,Coy2022,KuhnWZZ03,KuhnWZ02,KuhnWZ03}, all of which handle routing schemes in such UDGs). 

In \cite{Coy2022} it was shown that the nodes of a UDG can together simulate a much simpler \emph{grid graph} (see Def.\ \ref{def:grid_graph}) structure with constant overhead in round complexity, such that the connectivity, the hole-freeness, and the hop distance up to a constant factor are preserved. 
Furthermore, \cite{Coy2022} shows that a routing scheme on a grid graph can efficiently be transformed into a routing scheme for the underlying UDG, which introduces only a constant overhead on label size and local routing information and takes only a constant number of additional rounds.

This allows them to consider the much simpler grid graphs for the computation of routing schemes to generate good routing schemes for UDGs.
However, their actual algorithm for computing a routing scheme for a grid graph comes with a caveat: it works only for grid graphs \emph{without holes}, which, loosely speaking, are points in the grid without nodes on them that are enclosed by a cycle in the grid graph (more formally in Definition \ref{def:hole}), which implies routing schemes only for UDGs without ``radio holes'', which roughly correspond to areas enclosed by the UDG that are not covered by nodes (see \cite{Coy2022} for the formal definition of radio holes in UDGs).

In this work we extend the solution to UDGs with such radio holes. Note that the transformations given in \cite{Coy2022} from UDGs to grid graphs work even if there are holes, which essentially allows us to focus on the computation routing schemes for grid graphs with holes, which gives the same for arbitrary UDGs. Our algorithm is asymptotically as efficient (in computation time, and label and table sizes of the resulting routing scheme) as the algorithm of \cite{Coy2022} if the number of holes $|\holes|$ is small. 

We stress that it is \emph{significantly} more challenging to compute routing schemes on grid graphs with such holes than on grid graphs without holes. In the paper of Coy~et~al.~(\cite{Coy2022}) the authors heavily exploit the property that any $st$-path can be deformed into any other $st$-path: this is \emph{not true} in our setting. In particular, it is easy to come up with examples of grid graphs with $|\holes|$ holes for which there are at least $2^{|\holes|}$ ``reasonable''
\footnote{By a reasonable path we mean a path which does not completely encircle a hole or spiral around one: such a path can clearly be shortened.}
classes of $st$-paths (intuitively: paths which do not completely encircle or spiral around holes) which cannot be deformed into each other. Worse still, we can make all-but-one of these classes of paths almost arbitrarily long, and so it is not the case that we can consider just one arbitrary class of $st$-paths and obtain an approximate shortest path. Therefore we must determine the class in which an exact shortest $st$-path lies, and this seems to require a sparsifying structure which scales in complexity with the number of holes.


\subsection{Contributions}

Our main contribution is the following result.

\begin{theorem}[Main Result for Grid Graphs]\label{thm:grid_graph_routing_scheme}
    Given a grid graph $\Gamma$ with a set $\holes$ of holes (see Def.~\ref{def:hole}), we can compute an exact routing scheme for $\Gamma$ in $O(|\holes|^2 + \log n)$ rounds in the \HYBRID model. The labels are of size $O(\log n)$; nodes need to locally store $O(|\holes|^2 \log n)$ bits.
\end{theorem}

This result acts as an extension of the result presented in Coy et al.,~\cite{Coy2022}, which assumes that no holes are present making the graph structure much simpler. Note that \cite{Coy2022} showed that grid graphs approximate unit-disk graphs well (we briefly summarize this in Section \ref{sec:udgs_and_grids}), which leads to the following corollary:

\begin{corollary}[Main Result for UDGs]\label{cor:udg_routing_scheme}
    The above routing scheme can be transformed into a routing scheme which yields constant-stretch shortest paths for unit-disk graphs. Round complexity, label size, and local storage are asymptotically the same as above.
\end{corollary}

We believe that several of our technical contributions are of independent interest.
Our main technical contribution is a decomposition of a grid graph into \emph{simple, path-convex} regions which have useful properties for routing. 
We also provide a small skeleton structure of the UDG called a landmark graph with the property that shortest paths in the landmark graph are topologically the same (i.e., circumnavigate holes in the same way) as in the original graph, which may be useful when solving shortest paths on grid graphs and UDGs in \HYBRID or for similar problems in other models of computation.
Furthermore we give an $O(\log n)$ round algorithm for solving SSSP exactly in simple grid graphs and an $O(\log n)$ round algorithm for finding the distance from every node in a simple grid graph to a portal.

\subsection{Overview}

For an end-to-end overview of our approach, following an example graph, see \Cref{sec:intuition}.

In Section \ref{sec:closest_points_in_simple_grids}, we give an algorithm that solves SSSP in hole-free grid graphs (see Definition \ref{def:grid_graph} and Definition \ref{def:simple}) in $\bigO(\log n)$ rounds deterministically. We create two helper graphs, one representing its horizontal connections and one for the vertical connections, and we prove that SSSP solutions for these two graphs suffice to solve SSSP in the original grid graph. Further, we present a scheme that allows us to compute shortest paths of nodes to \emph{portal} sets (i.e., a connected straight path in the grid graph).

In Section \ref{sec:partitioning_the_graph}, we show that a grid graph can be decomposed into regions such that each region is simple (contains no holes) and \convex (for any nodes $s, t$ in a region, a shortest $st$-path lies entirely in their region), and with overlap only in portals (specifically called gates). We do this in stages: we decompose the graph into simple regions; we decompose these regions further so that they are ``tunnel'' shaped (i.e., bounded by at most $2$ gates); and finally we decompose them further still to ensure they are \convex.
We show that we can perform this decomposition quickly and the number of resulting regions is low.

In Section \ref{sec:landmark_graph}, we give a skeleton structure on this decomposition to facilitates routing. We mark vertices which lie on many shortest paths as \emph{landmarks}. We connect some landmarks in the same region to each other to form a \emph{landmark graph}, and show it can be computed quickly. Finally, we prove that for any grid nodes $s$ and $t$, a shortest $st$-path goes through the same regions as the shortest path from a landmark near $s$ to a landmark near $t$.

In Section \ref{sec:routing}, we combine our results to construct the routing scheme, first showing that each node can efficiently learn the whole landmark graph.
We then show that the region which a packet should enter next can be determined by combining information about the landmark graph, the target node's label, and local information about the region the message is currently in.
Nodes forward received packets to their neighbor which is closest to the packet's next region. We repeat this until the packet arrives at the target region, switching then to the routing scheme of \cite{Coy2022}.

\subsection{Related Work}

\textbf{Shortest Paths in Hybrid Networks.} Previous work in the \HYBRID model has mostly focused on shortest path problems \cite{Augustine2020,Kuhn2020,Feldmann2020,CensorHillel2021,Anagnostides2021}. In the $k$-sources shortest path ($k$-SSP) problem, all nodes must learn their distance in the (weighted) local network to a set of $k$ sources. Particular focus has been given to the all-pairs (APSP, $k=n$) and single-sources (SSSP, $k=1$) shortest-path problems. Note that solving the APSP problem gives a solution to the routing scheme problem. The complexity of APSP is essentially settled: \cite{Augustine2020,Kuhn2020} give an algorithm taking $\tilO(\!\sqrt{n})$ rounds, and this matches a lower bound of $\tilOm(\!\sqrt{k})$ to solve $k$-SSP even for polynomial approximations. This lower bound even matches a deterministic algorithm due to \cite{Anagnostides2021}, although only with an approximation factor of \smash{$\bigO(\frac{\log n}{\log \log n})$}.\footnote{\label{ftn:powerful_hybrid}These results are in the more powerful hybrid combination of \LOCAL and \NCC.} For $k$-SSP the $\tilOm(\!\sqrt{k})$ lower bound has been matched by \cite{CensorHillel2021} with a constant stretch algorithm, given sufficiently large $k$ (roughly $k \in \Omega(n^{2/3})$).$^{\ref{ftn:powerful_hybrid}}$ Whether there are any $\tilO(\!\sqrt{k})$ round $k$-SSP algorithms on general graphs for $1< k < n^{2/3}$ remains open.
The state-of-the-art algorithm for exact SSSP is provided by \cite{CensorHillel2021} and takes $\bigO(n^{1/3})$ rounds.$^{\ref{ftn:powerful_hybrid}}$ A recent result by \cite{Rozhon2022}, which solves SSSP by $\tilO(1)$ applications of an instruction set called ``minor-aggregation'' when given access to an oracle that solves the so called Eulerian Orientation problem, can be adapted for a $(1\p\eps)$ approximation of SSSP in $\tilO(1)$ rounds, as was shown in \cite{Schneider2023}.$^{\ref{ftn:powerful_hybrid}}$ An \textit{exact, deterministic} solution for SSSP in $\bigO(\log n)$ rounds has been achieved on specific classes of graphs (e.g.~cactus graphs, which includes trees) by \cite{Feldmann2020}. Another exact SSSP algorithm that takes $\tilO(\!\sqrt{SPD})$ rounds (where $SPD$ is the shortest-path-diameter of the local graph) is provided by \cite{Augustine2020}.

\textbf{Routing Schemes in Distributed Networks.} Our work builds on \cite{Coy2022}, in which they show how to compute a routing scheme for a UDG, by computing a routing scheme for a corresponding grid graph (see Section~\ref{sec:udgs_and_grids}). Their approach requires a simplifying assumption: the grid graph needs to be free of holes (see Def.~\ref{def:hole}), and this imposes a similar restriction on the underlying UDG. We remove that assumption in this work. In a recent article, \cite{Kuhn2022} considers computing routing schemes in the \HYBRID model on general graphs: they show that in $\tilO(n^{1/3})$ rounds one can compute exact routing schemes with labels of size $\tilO(n^{2/3})$ bits, or constant stretch approximations with smaller labels of $\bigO(\log n)$ bits. Interestingly, \cite{Kuhn2022} also gives lower bounds: they show that it takes $\tilOm(n^{1/3})$ rounds to compute exact routing schemes that hold for relabelings of size $\bigO(n^{2/3})$ and on unweighted graphs. They also give polynomial lower bounds for constant approximations on weighted graphs, implying that in order to overcome this lower bound in round complexity and label size, the restriction to a class of graphs is necessary!
The distributed round complexity of computing routing schemes was also considered in the \CONGEST model. In general, it takes $\tilOm(\!\sqrt{n}+D)$ rounds to solve the problem \cite{Sarma2012} (where $D$ is the graph diameter). This was nearly matched in a series of algorithmic results \cite{Lenzen2013a,Lenzen2015,Elkin2016}, for example, \cite{Elkin2016} gives a solution with stretch $\bigO(k)$, routing tables of size $\tilO(n^{1/k})$, labels of size $\tilO(k)$ in $\bigO\big(n^{1/2 + 1/k} \!+\! D_G\big) \cdot n^{o(1)}$ rounds.

\textbf{SINR Model as Local Network.} A well-accepted model for communication in wireless networks is the \textbf{S}ignal to \textbf{I}nterference and \textbf{N}oise \textbf{R}atio model, where a message is only received if the signal strength over the distance between sender and receiver is larger than the interference by other nodes plus some ambient noise. The latter implies a maximum range at which messages can be received, thus nodes that can communicate at zero interference induce a UDG. A result of Halld{\'{o}}rsson and Tonoyan \cite{HalldorssonT19} connects the SINR model to the \CONGEST model. They show that the SINR model allows to simulate the \CONGEST model on a locally sparse backbone structure which inherits the UDG property from the SINR model. Thus, by using the SINR model as local mode with our global mode ($\NCC_0$), we obtain a routing scheme on the backbone. This solution can be extended to nodes outside the backbone, by leveraging the ``download'' routine provided in \cite{HalldorssonT19}.

\section{Unit Disk Graphs and Grid Graphs}
\label{sec:udgs_and_grids}

In this section we mainly introduce concepts from \cite{Coy2022}, which we require in many of the following sections. We consider the class of \emph{Unit Disk Graphs} (UDG) defined as follows:

\begin{definition}[Unit Disk Graphs]
\label{def:udg}
    $G=(V,E)$ is a UDG if each node $v \in V$ is associated with a unique point in $\mathbb{R}^2$ and $u, v \in V: \{u, v\} \in E$ iff $\Vert u \m v \Vert_2 \leq 1$.
\end{definition}

\emph{Grid graphs} can be seen as a sparsifying structure for UDGs which can be easily simulated while preserving certain geometric properties and significantly simplifying the construction of algorithms for the original UDG (c.f., \cite{Coy2022}, more on that further below).

\begin{definition}[Grid Graphs]
    \label{def:grid_graph}
    A \emph{grid graph} $\Gamma = (V_\Gamma, E_\Gamma)$ is a graph such that the vertices $V$ uniquely correspond to points on a square grid $\mathbb Z^2$. Two vertices are connected by an edge in $E_\Gamma$ iff their corresponding points on the grid are horizontally or vertically adjacent.
\end{definition}

\subsection{Grid Graph Abstractions of Unit Disk Graphs}

We further explain the correspondence between UDGs and grid graphs.
As was shown in \cite{Coy2022}, one can compute and simulate a {grid-graph} abstraction $\Gamma = (V_\Gamma, E_\Gamma)$ of any input UDG using only local communication in $\bigO(1)$ rounds. In this simulated grid graph $\Gamma$, each grid node is represented with a close by node in the UDG. Any UDG node has to represent only a constant number of grid nodes (and can thus simulate all grid nodes it represents). 

\begin{theorem}[c.f.~\cite{Coy2022}]
    \label{thm:hybrid_on_grid}
    Let $G = (V,E)$ be a UDG. We can compute a grid graph $\Gamma = (V_\Gamma, E_\Gamma)$ in $\bigO(1)$ rounds, such that each grid node $g \in V_\Gamma$ is represented by some node $v \in V$ with $\Vert v \m g\Vert_2 \leq 1$. \emph{Every} node $v \in V$ has a representative $r \in V$ of some grid node $g \in V_\Gamma$ with $\hop_G(v,r) \leq 1$. A round of the \HYBRID model on $\Gamma$ can be simulated by the set of representatives in $V$ in $\bigO(1)$ rounds, and for any $s,t \in V$ with $\{s,t\} \notin E$ we have close-by representatives $g_s,g_r \in V_\Gamma$ such that from a shortest path $\Pi_{\Gamma}(g_s,g_t)$ from $g_s$ to $g_t$ in $\Gamma$ we can (locally and in $\bigO(1)$ rounds) construct a path $\Pi_G(s,t)$ in $G$ such that $|\Pi_G(s,t)| \leq 36 \cdot  \dist_G(s,t)$
\end{theorem}

The theorem above means that approximate paths with decent stretch on the UDG can be constructed from shortest paths on the grid graph abstraction $\Gamma$. In particular, this implies that any constant stretch routing scheme on $\Gamma$ also gives a constant stretch routing scheme on the underlying UDG. Therefore in order to obtain constant stretch routing schemes on UDGs it is sufficient to consider the problem on the (easier) class of grid graphs.

\begin{theorem}[c.f.~\cite{Coy2022}]
    \label{thm:routing_scheme_udg}
    Any algorithm that computes routing schemes with stretch $s$, labels of at most $x$ and local routing information at most $y$ bits on \emph{any grid graph} in $z$ rounds (where $x,y,z$ depend on $n$) implies an algorithm to compute a routing scheme with stretch $36 \cdot s$ on any UDG with labels of $\bigO(x)$ bits, local routing information of $\bigO(y)$ bits in $\bigO(z)$ rounds.
\end{theorem}

\subsection{Geometric Structure of Grid Graphs}

We start by defining paths and distances in grid graphs.

\begin{definition}[Paths]
    \label{def:paths}
    Let $\Gamma = (V_\Gamma,E_\Gamma)$ be an grid graph. We define a path $\Pi \subseteq E$ as a set of incident edges of $G$. For a $uv$-path $\Pi$ and a $vw$-path $\Pi'$ we denote the \textit{composite path} from $u$ to $w$ obtained by concatenating $\Pi$ and $\Pi'$ as $\Pi \circ \Pi'$. By $|\Pi|$ we denote the number of edges (or \emph{hops}) of a path $\Pi$. However, we will occasionally identify paths by the set of endpoints of its edges.
\end{definition}

\begin{definition}[Distances]
    \label{def:distances}
     The \emph{hop-distance} between two nodes $u, v \in V$ is defined as  $d_G(u,v) := \!\min_{\text{$u$-$v$-path } \Pi} |\Pi|$. Note that the hop distance is equal to the distance as all our edges have unit weight, hence we use the two interchangeably. Let $|\Pi_x|$ be the set of horizontal edges in the grid graph. Then $d_{x,G}(u,v) := \!\min_{\text{$u$-$v$-path } \Pi} |\Pi_x|$ is the horizontal distance between $u$ and $v$. Analogously we define the vertical distance $d_{y,G}$. We will drop the subscript $G$ when it is clear from the context.
\end{definition}


To analyze grid graphs we define some geometric structures, starting with \emph{portals} \cite{Coy2022}.

\begin{definition}[Portals]
    \label{def:portal}
    Let $E_v$ be the set of vertical edges of some grid graph $\Gamma = (V,E)$. Then the vertical portals are the connected components of the sub graph $(V,E_v)$. We say that portals $p_1$ and $p_2$ are \emph{adjacent} if $p_1$ is connected by an (horizontal) edge in $E$ to a node in $p_2$. Horizontal portals are based on the set of horizontal edges $E_h$ and the corresponding terms are defined analogously.
\end{definition}

Next, we introduce the concept of \emph{holes} in a grid graph. Intuitively, the area that is covered by a grid graph can be described by ``filling in'' the grid cells (those are unit squares $[a,a\!+\!1]\times[b,b\!+\!1]$ with $a,b \in \mathbb Z$) where all four corner nodes represent nodes in $\Gamma$ as well as all edges  of $\Gamma$. The holes are the connected areas in the Euclidean plane which are not ``filled''.

\begin{definition}
    \label{def:hole}
    Given a grid graph $\Gamma = (V_\Gamma, E_\Gamma)$. We define $S_\Gamma \subseteq \mathbb R^2$ as
    \begin{align*}
        & S_\Gamma^{\text{cells}} := \{ [a,a\!+\!1]\!\times\![b,b\!+\!1] \mid \text{nodes in $V_\Gamma$ at $(a,b), (a\!+\!1,b), (a,b\!+\!1), (a\!+\!1,b\!+\!1)$} \}\\
        & S_\Gamma^{\text{edges}} := \{ x \cdot u + (1\!-\!x) \cdot v \mid \{u,v\} \in E_\Gamma, x\in [0,1] \}\\
        & S_\Gamma := S_\Gamma^{\text{cells}} \cup S_\Gamma^{\text{edges}}
    \end{align*}
    Then all \emph{maximal connected components} in the complement $\mathbb{R}^2 \setminus S_\Gamma$ constitute the set of holes induced by $\Gamma$. We define the \emph{set of inner holes} $\mathcal H$ as all \textit{bounded}, maximal connected components in the complement $\mathbb{R}^2 \setminus S_\Gamma$. As the number of nodes in $\Gamma$ is finite, there is exactly one maximal connected component in $\mathbb{R}^2 \setminus S_\Gamma$ that is unbounded, which we call the \emph{outer hole}. 
\end{definition}

We now give some additional definitions pertaining to holes.





\begin{definition}
    \label{def:hole-boundary}
    Given a grid graph $\Gamma = (V_\Gamma, E_\Gamma)$ with inner holes $\mathcal H$. Then the boundary of some $H \in \mathcal{H}$ is the set of nodes whose corresponding points in the plane are on the geometric boundary $\partial H$ of $H$. We say that a node is incident to $H$ if it is part of its boundary. Similarly a node set (like a portal) is incident to $H$ if it contains an incident node.
\end{definition}


\begin{definition}
    \label{def:simple}
    A grid graph $\Gamma = (V_\Gamma, E_\Gamma)$ with $\mathcal{H} = \emptyset$ is called simple.
\end{definition}

Finally, \cite{Coy2022} gives an efficient algorithm for the special case of computing routing schemes in \textit{simple} grid graphs. 
This implies a constant stretch routing scheme for UDGs without holes by Theorem \ref{thm:routing_scheme_udg}.

\begin{theorem}[cf.\ \cite{Coy2022}]
    \label{thm:routing_scheme_grid_without_holes}
    An exact routing scheme for a simple grid graph $\Gamma$ using node labels and local space of $\bigO(\log n)$ bits
    can be computed in $\bigO(\log n)$ rounds in the \HYBRID model.
\end{theorem}

\section{Closest Point Computations in Simple Grid Graphs}\label{sec:closest_points_in_simple_grids}
    In this section, we aim to present two important subroutines for our main results. First, we present an $\mathcal O(\log n)$ Single Source Shortest Paths problem (SSSP) algorithm for simple grid graphs. The problem input is given by the grid graph in a distributed way (each node initially knows only its incident edges) and it is solved as soon all nodes know their distance to a dedicated source node.
    Afterwards, we use this algorithm to solve a problem we call the \emph{Single Portal Shortest Path} problem (SPSP). Here, we want to compute the distance of each node  to its corresponding closest point some dedicated portal (Def.\ \ref{def:portal}). We start with the following theorem.
    
    \begin{theorem}\label{simple_sssp}
        Given a simple grid graph $\Gamma$, we can solve SSSP on $\Gamma$ in $\mathcal O(\log n)$ rounds.
    \end{theorem}
    To achieve this runtime bound, we split the SSSP problem into its vertical and horizontal parts, i.e., we compute SSSP solutions with respect to vertical and horizontal distances separately and show that these can be combined to obtain the actual distance (which requires that $\Gamma$ is simple).     
    We start by defining two auxiliary graphs are used to compute distances in these two directions.

    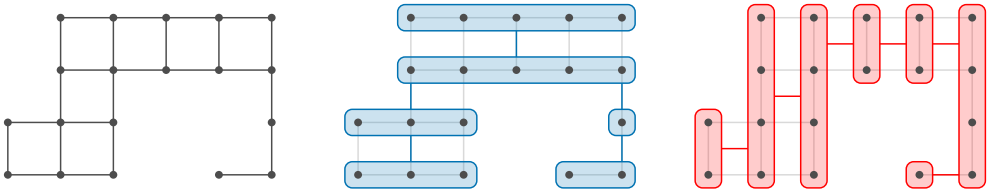
\begin{figure}[H]
        \centering
        \resizebox{!}{0.21\textwidth}{%
            \begin{tikzpicture}

    \def\gridpoints{
    (0.5, 0.5), (0.5, 1.5), (1.5, 0.5), (1.5, 1.5), (1.5, 2.5), (1.5, 3.5), (2.5, 0.5), (2.5, 1.5), (2.5, 2.5), (2.5, 3.5),  (3.5, 2.5), (3.5, 3.5), (4.5, 0.5), (4.5, 2.5), (4.5, 3.5), (5.5, 0.5), (5.5, 1.5), (5.5, 2.5), (5.5, 3.5)
    }
    
    \def\gridone{31}

    \foreach \p in \gridpoints {
        \node[circle,fill=black!70!white,inner sep=1.5pt] at \p {};
    }
    \foreach \p [count=\i]  in \gridpoints {
        \foreach \q [count=\j]  in \gridpoints {
            \ifnum \i<\j
                \tikzmath{
                    coordinate \A;
                    \A = \p;
                    coordinate \B;
                    \B = \q;
                    \dist1=veclen(\Ax-\Bx,\Ay-\By);
                }
                \ifdim \dist1 pt < \gridone pt 
                            \draw[thick,black!70!white] \p -- \q;
                \fi
            \fi
        }
    }

    \node at (0,0) {};
    \node at (6,4) {};

\end{tikzpicture}
        }
        \resizebox{!}{0.21\textwidth}{%
            \begin{tikzpicture}

    \def\gridpoints{
        (0.5, 0.5), (0.5, 1.5), (1.5, 0.5), (1.5, 1.5), (1.5, 2.5), (1.5, 3.5), (2.5, 0.5), (2.5, 1.5), (2.5, 2.5), (2.5, 3.5), (3.5, 2.5), (3.5, 3.5), (4.5, 0.5), (4.5, 2.5), (4.5, 3.5), (5.5, 0.5), (5.5, 1.5), (5.5, 2.5), (5.5, 3.5)
    }

    \def\gridone{31}
    \foreach \p [count=\i]  in \gridpoints {
        \foreach \q [count=\j]  in \gridpoints {
            \ifnum \i<\j
                \tikzmath{
                    coordinate \A;
                    \A = \p;
                    coordinate \B;
                    \B = \q;
                    \dist1=veclen(\Ax-\Bx,\Ay-\By);
                }
                \ifdim \dist1 pt < \gridone pt 
                            \draw[thick,black!15!white] \p -- \q;
                \fi
            \fi
        }
    }

    \draw[rounded corners,draw=deepcerulean, thick, fill=deepcerulean,fill opacity=0.2] (0.25, 0.25) rectangle (2.75, 0.75) {};
    \draw[rounded corners,draw=deepcerulean, thick, fill=deepcerulean,fill opacity=0.2] (4.25, 0.25) rectangle (5.75, 0.75) {};
    \draw[rounded corners,draw=deepcerulean, thick, fill=deepcerulean,fill opacity=0.2] (1.25, 3.25) rectangle (5.75, 3.75) {};
    \draw[rounded corners,draw=deepcerulean, thick, fill=deepcerulean,fill opacity=0.2] (1.25, 2.25) rectangle (5.75, 2.75) {};
    \draw[rounded corners,draw=deepcerulean, thick, fill=deepcerulean,fill opacity=0.2] (0.25, 1.25) rectangle (2.75, 1.75) {};
    \draw[rounded corners,draw=deepcerulean, thick, fill=deepcerulean,fill opacity=0.2] (5.25, 1.25) rectangle (5.75, 1.75) {};

    \foreach \p in \gridpoints {
        \node[circle,fill=black!70!white,inner sep=1.5pt] at \p {};
    }

    \draw[deepcerulean,thick] (1.5,0.75) -- (1.5, 1.25);
    \draw[deepcerulean,thick] (1.5,1.75) -- (1.5, 2.25);
    \draw[deepcerulean,thick] (5.5,0.75) -- (5.5, 1.25);
    \draw[deepcerulean,thick] (5.5,1.75) -- (5.5, 2.25);
    \draw[deepcerulean,thick] (3.5,2.75) -- (3.5, 3.25);

    \node at (0,0) {};
    \node at (6,4) {};

\end{tikzpicture}
        }
        \resizebox{!}{0.21\textwidth}{%
            \begin{tikzpicture}

    \def\gridpoints{
        (0.5, 0.5), (0.5, 1.5), (1.5, 0.5), (1.5, 1.5), (1.5, 2.5), (1.5, 3.5), (2.5, 0.5), (2.5, 1.5), (2.5, 2.5), (2.5, 3.5), (3.5, 2.5), (3.5, 3.5), (4.5, 0.5), (4.5, 2.5), (4.5, 3.5), (5.5, 0.5), (5.5, 1.5), (5.5, 2.5), (5.5, 3.5)
    }

    \def\gridone{31}
    \foreach \p [count=\i]  in \gridpoints {
        \foreach \q [count=\j]  in \gridpoints {
            \ifnum \i<\j
                \tikzmath{
                    coordinate \A;
                    \A = \p;
                    coordinate \B;
                    \B = \q;
                    \dist1=veclen(\Ax-\Bx,\Ay-\By);
                }
                \ifdim \dist1 pt < \gridone pt 
                            \draw[thick,black!15!white] \p -- \q;
                \fi
            \fi
        }
    }

    \draw[rounded corners,draw=tenn, thick, fill=tenn,fill opacity=0.2] (0.25, 0.25) rectangle (0.75, 1.75) {};
    \draw[rounded corners,draw=tenn, thick, fill=tenn,fill opacity=0.2] (1.25, 0.25) rectangle (1.75, 3.75) {};
    \draw[rounded corners,draw=tenn, thick, fill=tenn,fill opacity=0.2] (2.25, 0.25) rectangle (2.75, 3.75) {};
    \draw[rounded corners,draw=tenn, thick, fill=tenn,fill opacity=0.2] (5.25, 0.25) rectangle (5.75, 3.75) {};     
    \draw[rounded corners,draw=tenn, thick, fill=tenn,fill opacity=0.2] (4.25, 0.25) rectangle (4.75, 0.75) {};
    \draw[rounded corners,draw=tenn, thick, fill=tenn,fill opacity=0.2] (3.25, 2.25) rectangle (3.75, 3.75) {};
    \draw[rounded corners,draw=tenn, thick, fill=tenn,fill opacity=0.2] (4.25, 2.25) rectangle (4.75, 3.75) {};

    \foreach \p in \gridpoints {
        \node[circle,fill=black!70!white,inner sep=1.5pt] at \p {};
    }
    
    \draw[tenn,thick] (0.75,1) -- (1.25, 1);
    \draw[tenn,thick] (1.75,2) -- (2.25, 2);
    \draw[tenn,thick] (4.75,0.5) -- (5.25, 0.5);
    \draw[tenn,thick] (2.75,3) -- (3.25, 3);
    \draw[tenn,thick] (3.75,3) -- (4.25, 3);
    \draw[tenn,thick] (4.75,3) -- (5.25, 3);

    \node at (0,0) {};
    \node at (6,4) {};

\end{tikzpicture}
        }
        \caption{A simple grid graph (left) and the corresponding horizontal (center) and vertical (right) portal graphs.}
        \label{fig:portal_trees}
    \end{figure}
    
    \begin{definition}\label{def:portal_graph}
        Given a grid graph $\Gamma$, we define the \emph{vertical portal graph} $\mathcal{P}_v$ 
        to be the the graph with vertices corresponding to the vertical portals of $\Gamma$.
        Two vertices in $\mathcal{P}_v$ are connected by an edge iff their corresponding portals are adjacent (i.e., connected with a horizontal edge) in $\Gamma$.
        The horizontal portal graph $\mathcal{P}_h$ is defined analogously. 
    \end{definition}
    As we will see later, the distances in $\mathcal{P}_v$ correspond to horizontal distances taken by paths and those in $\mathcal{P}_h$ correspond to vertical distances, respectively.
    This motivates us to run SSSP on these graphs and combine the solutions.
    
    \begin{lemma}
    \label{lem:portal_graphs_of_simple_regions_are_trees}
        If a grid graph $\Gamma$ is simple, both its vertical portal graph $\mathcal{P}_v$ and its horizontal portal graph $\mathcal{P}_h$ are trees. 
    \end{lemma}
    \begin{proof}
        For simplicity, we state the proof for the vertical portal graph $\mathcal{P}_v$, the claim for $\mathcal{P}_h$ follows by symmetry.    
        We prove the claim by contradiction.
        Assume there was a cycle in $\mathcal{P}_v$.
        This cycle would correspond to a set of portals in $\Gamma$, such that two adjacent nodes in the cycle correspond to two adjacent portals in $\Gamma$.
        We pick two portals from this set that have the same $x$ coordinate arbitrarily.
        Such a pair must exist due to the existence of a cycle in $\mathcal{P}_v$.
        As the portals are maximal vertically connected components by definition, there is some space enclosed between them.
        Further, as the portals are part of a cycle in $\mathcal{P}_v$, they are connected by two paths in $P_v$.
        Starting from the northern portal of the pair we picked, one of those paths surround the area between the pair of portals in a clockwise direction, while the other surrounds it in a counter-clockwise direction.
        Together, they enclose it and it must correspond to the hole, which contradicts simplicity.
        
        \jw{I reworked the proof. It is a little handwavey now. I think it is fine, though. Do you agree?}
    \end{proof}
    This allows us to employ a slightly modified version of the hybrid SSSP algorithm for trees presented in \cite[Theorem 6]{Feldmann2020}.
    \begin{lemma}
    \label{lem:sssp_in_portal_graph}
        Given a simple grid graph $G$, we can solve SSSP on its vertical and horizontal portal graph in time $\mathcal O(\log n)$.
    \end{lemma}
    \begin{proof}
        The algorithm of \cite{Feldmann2020} can be used to compute a SSSP solution in trees in time $\mathcal O(\log n)$ in the \HYBRID model.
        Therefore, it suffices to construct a subgraph of $\Gamma$ that maintains the tree structure and the distances of $\mathcal{P}_v$. It consists of all nodes of $\Gamma$ and all vertical edges. For each pair of adjacent portals in $\mathcal{P}_v$ we add a \emph{single} arbitrary horizontal edge connecting the two portals. Note that identifying adjacent pairs of portals and one edge that connects those can be done in $\bigO(\log n)$ rounds in the \HYBRID model. We achieve this, by first employing \emph{pointer jumping} on the portal to aggregate and broadcast (Lemma \ref{lem:broadcast_and_aggregation}) the minimum identifier, acting as a portal identifier. Each node informs its horizontal neighbors about its own portal's identifier. Next the nodes communicate with their vertical neighbors, which identifiers they have received this way. This allows the bottommost node neighboring a specific portal to mark its edge in the corresponding direction as the edge connecting the two portals. As vertical portals are maximal vertically connected subgraphs, this selection is unique. 
        To preserve the distances in $\calP_v$ we simply assign a weight of $1$ to all horizontal $0$ to all vertical egdes.
        The procedure for $\mathcal{P}_h$ works analogous (exchange the words horizontal and vertical).     
    \end{proof}

    Next, we combine the SSSP solutions for $\mathcal{P}_v$ and $\mathcal{P}_h$ into a SSSP solution for $\Gamma$.
    \begin{lemma}
        Given a simple grid graph $\Gamma$, a starting node $s$ and SSSP solutions for its vertical and horizontal portal graphs $\mathcal{P}_v$ and $\mathcal{P}_h$ starting from the corresponding portals containing $s$, we can compute an SSSP solution for $\Gamma$ starting from $s$ in time $\mathcal O(\log n)$.
    \end{lemma}
    \begin{proof}
        Each node computes its distance to $s$ in $\Gamma$ by adding the two distances obtained by the SSSP computations in $\mathcal{P}_v$ and $\mathcal{P}_h$. 
        To show the correctness of these distances, we perform an inductive argument along the shortest paths. Let $w\in V$ be an arbitrary node and denote its computed vertical distance with $d_v(w)$ in and its computed horizontal distance with $d_h(w)$.
        Consider a shortest path $(s=w_1,\dots,w_\ell=w)$ from $s$ to $w$. Since both portals containing $s$ correspond to the starting nodes of the SSSP executions in $\mathcal{P}_v$ and $\mathcal{P}_h$, both distances $s$ obtained by the procedure described above are $0$, i.e., we have $d(s)=d_v(s)+d_h(s)=0$, which is our base case.
        Next, we argue that $d(w_n)=d_v(w_n)+d_h(w_n)$ implies $d(w_{n+1})=d_v(w_{n+1})+d_h(w_{n+1})$, for $n+1\le \ell$. 
        Note that node $w_{n+1}$ is node $w_n$'s successor in the shortest path and further away from $s$ than it by construction. If it is vertically adjacent to $w_n$, we have $d_v(w_{n+1})=d_v(w_n)+1$ and $d_h(w_{n+1})=d_h(w_n)$. Naturally, $d(w_{n+1})=d(w_n)+1$ also holds.
        Therefore, we have $d(w_{n+1})=d(w_n)+1=\left(d_v(w_n)+1\right)+d_h(w_n)=d_v(w_{n+1})+d_h(w_{n+1})$.
        The case where the shortest path makes a horizontal step is analogous.
        \jw{reworked the proof and removed one of the cases for easier reading}

        Since the distance values are correct and the predecessor pointers are picked to point to any neighbor closer to $s$, their correctness immediately follows.
    \end{proof}


    Using \Cref{lem:sssp_in_portal_graph}, we combine the SSSP solutions for the portal graphs $\mathcal{P}_v$ and $\mathcal{P}_h$ to compute an SSSP solution for the grid graph $\Gamma$, which allows us to conclude \Cref{simple_sssp}.
    We can extend \Cref{simple_sssp} to compute the distance of each node in a grid graph to a dedicated portal $P$. We refer to this problem as the \emph{Single Portal Shortest Path} (SPSP) problem. To solve the SPSP problem, we pick an arbitrary node of $P$ as the starting node of an SSSP execution and consider the weights of all edges of $P$ to be zero during that execution. 
    
    \begin{corollary}
    \label{lem:closest_points_in_simple_grids}
        Given a simple grid graph $\Gamma$ and a portal $P$ in $\Gamma$.
        We can compute the distances of every vertex in $\Gamma$ to $P$ and its neighbor that is closest to $P$ in time $\mathcal O(\log n)$.
    \end{corollary}

    We conclude this section with an observation on how to compute a shortest path tree using the distance values of an SSSP computation or an SPSP computation.

    \begin{observation}[Shortest Path Tree]\label{obs:shortest_path_tree}
        After computing SSSP (or SPSP), by having each node pick a predecessor pointer arbitrarily from its neighbors with minimal distance to $s$ we can establish the corresponding shortest path tree in a distributed sense (this step takes just $\mathcal O(1)$ time using local edges).
    \end{observation}
    

\section{Pathconvex Region Decomposition}
\label{sec:partitioning_the_graph}

In this section we establish that grid graphs can be partitioned into comparatively few sets of nodes (with some overlap at the borders) called \textit{regions} that are simple, i.e., they have no holes and \convex, which we define as follows.

\begin{definition}[Path Convexity]
    \label{def:convex}
    Let $\Gamma = (V,E)$ be a grid graph and let $R \subseteq V$. Then $R$ is called \emph{\convex} if for any pair $u,v \in R$ there is a shortest $uv$-path contained completely within that region.\footnote{This leans on the standard definition of convex node sets in the context of graphs. A node set is called convex if \textit{all} (instead of just one) shortest paths between any pairs of nodes are within the set.} We say that a region is \convex with respect to a set of points in that region, if there is a shortest path between any pair of those points that stays within the region.
\end{definition}

\begin{definition}[Region Decomposition]
    \label{def:region_decomposition}
    Let $\Gamma = (V,E)$ be a grid graph. A \textit{region decomposition} of $\Gamma$ is a family of connected node sets $\mathcal R = \{R_1, \dots , R_\ell\}$ (regions) with $V = R_1 \cup \dots \cup R_\ell$. 
    We call a region \emph{simple} if it has no inner holes (c.f., Definition \ref{def:hole}). We call $\mathcal R$ simple and \convex if all regions in $\mathcal R$  are simple and \convex (c.f., Definition \ref{def:convex}).
\end{definition}

For the implementation, we will achieve a simple and convex region decomposition $\mathcal R$ by disconnecting parts of $\Gamma$ and considering each connected component as a separate region. For this we create copies of nodes that are on the border of at least two regions and connect each copy only to neighbors in one of those region.

The construction breaks down into three main steps. In the first step, we decompose $\Gamma$ into simple regions. Second, we break those regions up further into ``tunnels'', which are defined by the property that they overlap in at most two portals (called gates) with all of their neighboring regions. In the final step we show that such tunnels have crucial properties that allows us to make them convex by subdividing them a constant number of times.
Ultimately, we prove the following theorem.

\begin{theorem}
\label{thm:convex_decomposition}
    For any grid graph $\Gamma$, a decomposition into $\bigO(|\mathcal H|)$ simple, \convex regions can be computed in $\bigO(|\mathcal H| + \log n)$ rounds in the \HYBRID model.
\end{theorem}

\textbf{Splitting Operations.} To achieve the desired region decomposition we split the the grid graph at strategic locations. The most basic such operation is to split a grid graph (or a region) at some portal, where each node will simulate two copies of itself, a ``right copy'' which has no left neighbor and a ``left copy'' which has no right neighbor. This establishes a new grid graph where nodes that have been split will act in the role of the nodes they simulate only, which blocks paths through the splitting portal and might disconnect $\Gamma$. 

If such a splitting portal touches the boundary of a hole, we often further split at a boundary node. in particular, we split the simulated boundary node that is on the ``same side'' as the hole (say the left copy if the hole is left of the portal) into a ``top'' and ``bottom'' copy, which do not have a bottom or top neighbor, respectively. This is where up to three regions may intersect in the resulting region decomposition and can also be used to break up cycles around a hole, thereby making the resulting regions simple. We describe such splitting operations in detail and prove that they can be conducted efficiently in Appendix \ref{appsec:splitting_portals}, Definition \ref{def:splitting_procedure} and Lemma \ref{lem:splitting_procedure_overhead}.

We often talk of the region decomposition in the sense of Definition \ref{def:region_decomposition} after conducting such a split. It is given implicitly by the connected components in the grid graph formed by the simulated nodes after splitting at a portal or a node according to Definition \ref{def:splitting_procedure}. As a consequence of these splits, during intermediate steps, we occasionally end up with (unconnected) copies of some grid node within the same region thus we have to slightly generalize our notion of grid graphs to accommodate for those.

Since we split at portals, we obtain the property that regions overlap only in node sets that form portals (i.e., maximal vertically connected components). To distinguish those from ordinary portals (Def.\ \ref{def:portal}), we call these \emph{gates}, which we formally define as follows.

\begin{definition}[Gate]
    \label{def:gate}
    Let $\Gamma$ be a (simulated) grid graph after splitting at some vertical portal (according to Appendix \ref{appsec:splitting_portals} Definition \ref{def:splitting_procedure}) has been conducted (the case with a horizontal portal is symmetric). Then a node on that portal is called a \emph{gate node}. A maximal vertically or horizontally connected component of gate nodes within a given region of the resulting decomposition after the split is called a \emph{gate}.
\end{definition}

Note that in the subsequent sections, gates will be vertical until the last step.
Furthermore, we classify connected segments of the boundary nodes of some region, which are not on gates. We will refer to these as \emph{walls}, formally defined as follows.

\begin{definition}[Wall]
    \label{def:wall}
    Let $\calR$ be a region decomposition of $\Gamma$ established by conducting splitting operations as defined in Appendix \ref{appsec:splitting_portals} Definition \ref{def:splitting_procedure}. For $R \in \calR$ We denote \emph{connected} segments of hole boundary nodes in $R$ (see Def.\ \ref{def:hole-boundary}) that are \emph{not} gate nodes (cf. Def.\ \ref{def:gate}) as \emph{walls}. Note that walls and gates alternate on the boundary of $R$. Also note that different wall segments  might share nodes.
\end{definition}

\subsection{Decomposition into Simple Regions}
\label{sec:simple_decomposition}

The first step is to split our grid graph into \emph{simple} regions, i.e., regions without holes. We show that a relatively simple region decomposition can achieve this goal and give the implementation details of the construction separately.

\begin{definition}[Splitting $\Gamma$ into Simple Regions]
\label{def:construction-Gamma} For each inner hole $H \in \mathcal H$ of $\Gamma$ we repeat the following. Let $v_H$ be the leftmost node on the boundary of $H$ (we make $v_H$ unique by choosing the northernmost among leftmost boundary-nodes). Let $P_H$ be the unique vertical portal with $v_H \in P_H$. We conduct splits at $P_H$ and $v_H$ (for the details, see Definition \ref{def:splitting_procedure} case two and three).
In general, $P_H$ might contain leftmost nodes of boundaries of different inner holes. In that case we conduct the vertical split at the northernmost node of each such hole, as described in the third case of Definition \ref{def:splitting_procedure} (and each such inner hole needs not be considered further in the iteration over $\calH$). Figure \ref{fig:simple_regions} shows an example of the procedure.
\end{definition}

\begin{figure}
    \centering
    \caption{Decomposition of grid graph into simple regions by splitting at certain portals \& nodes.}
    \label{fig:simple_regions}
    \begin{subfigure}{0.45\textwidth}
        \centering        \includegraphics[width=\textwidth,page=1]{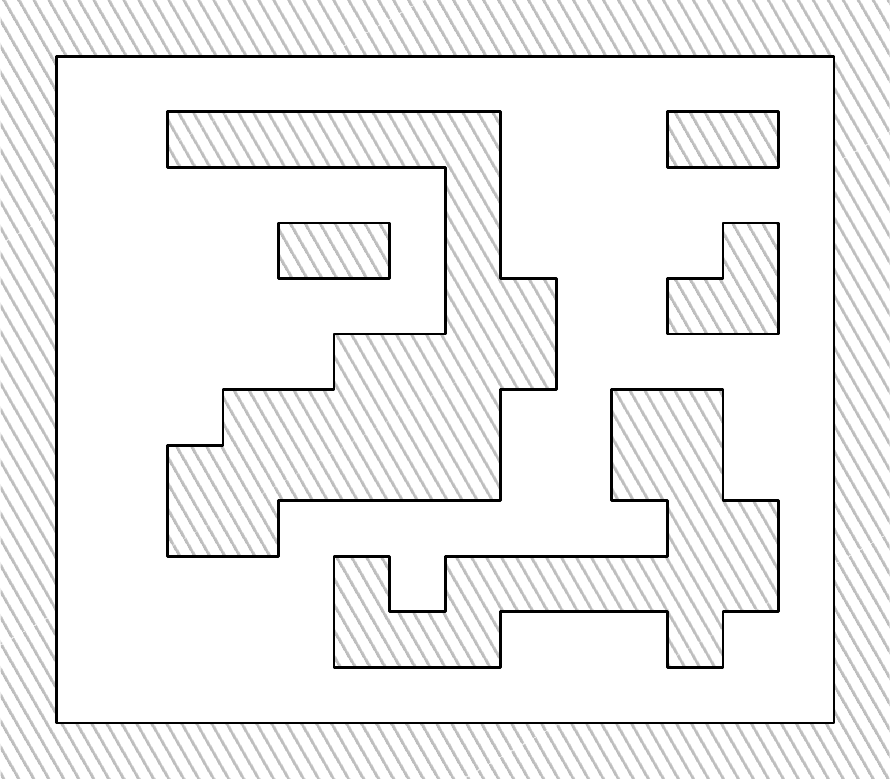}
        \caption{A stylised grid graph with holes, where white areas are occupied by grid nodes and shaded areas correspond to holes (see Def.\ \ref{def:hole}).\\}
    \end{subfigure}
    \hspace{0.08\textwidth}
    \begin{subfigure}{0.45\textwidth}
        \centering        \includegraphics[width=\textwidth,page=2]{figures/simple_regions.pdf}
        \caption{Red lines mark portals through a leftmost node of some hole at which a split occurs. Further splits take place at a dedicated leftmost node of each hole (details in Def.\ \ref{def:splitting_procedure}).}   
    \end{subfigure}
\end{figure}

The connected components that result from the above construction form regions that are simple. The idea is that each such split will make the grid graph at the portal $P_H$ horizontally impassable and at $v_H$ vertically impassable. Roughly speaking, the split at a portal and a node on the portal and hole boundary creates a ``thin'' hole such that in the resulting graph two holes ``merge'', and become a single hole. This decreases the overall number of \textit{inner} holes by one. By repeating this split for each inner hole we will be left with one hole in the end, which is the outer hole.

However, splitting portals this way leads to ``infinitely thin hole sections'' which does not fit with our definition of holes (\Cref{def:hole}). Therefore, we use as an auxilliary construction a grid graph with finer granularity $1/3$. This construction has the same topological properties, and the Definition \ref{def:hole} can be applied again. It can be modified such that it becomes simple, and finally it can be transformed back into a grid graph with normal granularity such the result equals \textit{exactly} the one from the construction above. In the following, we give the  detailed work on this, in particular we show the following two lemmas:

\begin{lemma}[Correctness of Simple Decomposition]
\label{lem:simple_decomp}
    The construction in Definition \ref{def:construction-Gamma} decomposes $\Gamma$ into at most $|\calH| \p 1$ simple regions.
\end{lemma}

\begin{lemma}[Computing the Simple Decomposition]
\label{lem:computing_simple_decomp}
    A \emph{simple} region decomposition can be computed for any grid graph $\Gamma$ in $\bigO(\log n)$ rounds in the \HYBRID model.
\end{lemma}

For the proof, we start with the Definition of the grid graph of granularity 1/3.

\begin{definition}[$1/3$-Grid Transformation] We transform a grid graph $\Gamma = (V_\Gamma, E_\Gamma)$ into a grid graph $\Gamma' = (V_{\Gamma'}, E_{\Gamma'})$ that has granularity $\frac{1}{3}$ as follows. For each node $v \in V_\Gamma$ at position $(a,b) \in \mathbb Z^2$, we create nine nodes in $V_\Gamma'$ at positions $(a \p x,b \p y)$ with $x,y \in \{-\frac{1}{3}, 0, \frac{1}{3}\}$. We say that $u',v' \in V_{\Gamma'}$ belong to the group of $v \in V_\Gamma$ if they were created from $v$ (i.e., rounding their coordinates to the closest integer gives the position of $v$). For each pair $u',v'\in V_\Gamma$ We add an edge $\{u',v'\}$ to $E_{\Gamma'}$ between two nodes $u',v' \in V_{\Gamma'}$ if $\Vert u'\m v'\Vert_2 = \frac{1}{3}$. An edge $\{u',v'\} \in E_{\Gamma'}$ is assigned weight 0 if $u',v'$ belong to the same group, else weight 1.
\end{definition}

We will decompose the transformed graph $\Gamma' = (V_{\Gamma'}, E_{\Gamma'})$ with a procedure similar to the one proposed for $\Gamma$ and make an argument that it is simple, i.e., no region has any holes. Afterwards we ``round'' the nodes of $\Gamma'$ back to the integer grid and argue that a) the region decomposition is still simple and b) That the final result is exactly the same as doing the procedure in Definition \ref{def:construction-Gamma} for $\Gamma$. The reverse Transformation works as follows.

\begin{definition}[Reverse Transformation] We transform a grid graph $\Gamma' = (V_{\Gamma'}, E_{\Gamma'})$ that has granularity $\frac{1}{3}$ back into a grid graph $\Gamma = (V_\Gamma, E_\Gamma)$ with granularity 1. Note that the resulting graph may have nodes sharing the same coordinate. For $p \in \mathbb Z^2$ Let $V_{p}$ be the set of nodes in $V_{\Gamma'}$ whose coordinates are $p$ when rounded to the closest integer. Let $\tilde u, \tilde v$ be two connected components in $V_{p}$. We create a new node $V_\Gamma$ at $p$, for each connected component (this is where the nodes occupying the same coordinate can occur).
Finally we define the edges of $\Gamma$. Let $p_1, p_2 \in \mathbb Z^2$ and let $\tilde u$, $\tilde v$ be connected components in $V_{p_1}$ and $V_{p_2}$, respectively, i.e., $\tilde u, \tilde v \in V_{\Gamma}$. We add an edge $\{\tilde u,\tilde v\}$ between nodes $\tilde u, \tilde v \in V_\Gamma$, if 
there are two nodes $u \in \tilde u$ and $v \in  \tilde v$ with $\{u,v\} \in E_{\Gamma'}$, i.e., $u,v$ share an edge in $\Gamma'$.
\end{definition}

The next definition gives an iterative construction on how to make $\Gamma'$ simple.


\begin{definition}[Construction to remove $H$ in $\Gamma'$]
\label{def:construction-Gamma-prime}
Let $\Gamma'$ be the $1/3$ Transformation of a grid graph $\Gamma$. For each inner hole $H \in \mathcal H$ of $\Gamma$ we repeat the following. Let $v_H$ be the splitting node, i.e., the chosen leftmost node on the boundary of $H$. Let $v'_H \in V_{\Gamma'}$ be the middle node (the one with coordinates in $\mathbb{Z}$) of the 9 nodes created from $v_H$ in $\Gamma'$. Let $P'_H$ be the unique portal in $\Gamma'$ with $v'_H \in P'_H$. We remove all nodes on $P'_H$ from $\Gamma'$ as well as the node right of $v'_H$ (which is at the boundary of some hole in $\Gamma'$). If there is a splitting node $v_{H'}$ of another hole $H'$ with $P_y$ we repeat the same procedure for $v_{H'}$ in $\Gamma'$.
\end{definition}



\begin{proof}[Proof of \Cref{lem:simple_decomp}]
    When splitting at some hole $H \in \calH$ (according to the procedure in Definition \ref{def:splitting_procedure}) we introduce at most one additional connected component, i.e. one additional region, in case $P_H$ is a node separator\footnote{A node separator w.r.t.\ two nodes $u,v$ is a set of nodes, the removal of which decomposes the graph into (at least) two connected components, where $u,v$ are in different components.} in its current region. Note that, even though we also split at $v_H$, the region to the right of $P_H$ will remain connected via the boundary of $H$ (we end up with two copied nodes at the position of $v_H$ in the same region, though). Hence the number of regions is at most $|\calH| \p 1$.
    
    The idea to prove the simple property is that each split of the grid graph as described in the procedure above (Def.\ \ref{def:construction-Gamma}), will make the portal $P_H$ impassable in horizontal direction and $v_H$ impassible in vertical direction. Loosely speaking, in the resulting grid graph at least two holes will ``merge'', and become one hole. After repeating this for all holes, we obtain a grid graph with a single hole, namely the outer hole, i.e., it is simple (Def.\ \ref{def:simple}).
    
    However, the idea to ``merge'' holes via an impassable portal is not reconcilable with our definition of holes \ref{def:hole}. Thus we use a $1/3$-transformation to $\Gamma'$, where we can show that at least two holes merge using our formal definition with an analogous construction in $\Gamma'$ (Def.\ \ref{def:construction-Gamma-prime}). We then argue that the reverse transformation of the resulting simple graph with granularity $1/3$ corresponds to the same grid graph that we obtain with the original procedure has the same topology, that is, the same number of holes per region, i.e., none.
    
    Let $H$ be said hole with splitting node $v_H$. Let $\Gamma'$ be the the $1/3$-transformation of $\Gamma$ \textit{after} removing $H$ as per the construction explained above. After the construction, the transformed nodes from the Portal $P_y$ form a ``channel'', which connects $H$ with at least one other hole, as otherwise $v_H$ could not have been the leftmost node. Hence the number holes in $\Gamma'$ decreases by at least one.
    
    After repeating the above for all nodes, the remaining nodes in $\Gamma'$ only have a single hole left, the outer hole.  This is equivalent to the fact that for any pair of nodes, two paths in $\Gamma'$ between those nodes (if any exist) can be continuously transformed into one another within $S_{\Gamma'}$ (see Def.\ \ref{def:hole}), i.e., without ``transforming'' paths through holes.
    The same follows for any pair of nodes in $\Gamma$, as $\Gamma$ does not admit the passing of the same portals which we ``channeled'' in $\Gamma'$. This is equivalent to the fact that $\Gamma$ has only one hole, the outer hole.
\end{proof}

\begin{lemma}
    \label{lem:node_copies_in_region}
    After the construction to remove holes in $\Gamma$ (Def.\ \ref{def:construction-Gamma}), the only nodes that share coordinates are two copies of splitting nodes $v_H$.
\end{lemma}
\begin{proof}
     For a contradiction assume that after executing the construction (Def.\ \ref{def:construction-Gamma}) in $\Gamma$, the left and right node copies, say $v_\ell$ and $v_r$ of $P_H$ are in the same connected component (region). Then there is a $v_\ell,v_r$-path within that region. Since $P_H$ can not be crossed, that path must enclose at least one other inner hole (which touches on of the end nodes of $P_H$). But then this path would also have to cross the splitting portal $P_{H'}$ of at least one inner hole $H'$, which is a contradiction.
\end{proof}

We show that the procedure in Definition \ref{def:construction-Gamma} can be conducted in $\bigO(\log n)$ rounds.

\begin{proof}[Proof of \Cref{lem:computing_simple_decomp}]
    In  Appendix~\ref{appsec:identifying_holes} and Appendix~\ref{appsec:splitting_portals} we describe a few basic primitives that allows us to conduct broadcasts and aggregations on paths and cycles in $\bigO(\log n)$ rounds. Note that these procedures can be done in parallel in case such structures intersect only $\bigO(1)$ times in each node. Also note that portals and boundaries of holes form paths and cycles with only $\bigO(1)$ intersections per node.
    
    Using these primitives, it is easy to establish splitting nodes $v_H$ and splitting Portals $P_H$ (see Definition \ref{def:construction-Gamma}) by broadcasting the minimum $x$-coordinate on the boundary of $H$ (breaking ties by maximum $y$ coordinate) and subsequently informing the nodes on $P_H$ of their role as nodes on a splitting portal. We do this in parallel for each $H$ and $P_H$ (note that only one broadcast has be conducted on $P_H$ in case $P_H$ is splitting portal for multiple holes). Determining the roles of nodes is the main part, the subsequent splitting procedure and simulation overhead takes just $\bigO(1)$ rounds by Lemma \ref{lem:splitting_procedure_overhead}.
\end{proof}

\subsection{Decomposition into Tunnel Regions}
\label{sec:tunnel_decomposition}

Our next step is to ensure that each region is a \emph{tunnel}, which we define as a region that has at most two gates (see Def.\ \ref{def:gate}).
Recall that as a consequence of the previous Section~\ref{sec:simple_decomposition}, we start out with a grid graph $\Gamma$ that is decomposed into simple regions (cf. Lemma \ref{lem:simple_decomp}).

\begin{remark}
\label{rem:junction_structure}
    Recall that the boundary of each region is composed of alternating segments of walls and vertical gates (cf.\ Def's \ref{def:gate}, \ref{def:wall}). Being a member of a wall or a gate of some region is a condition that each node can determine locally. Furthermore, walls and gates can both compute unique identifiers in $O(\log n)$ rounds, using our procedures for aggregation and broadcast in Appendix \ref{appsec:pointer_jumping} (Lemma \ref{lem:broadcast_and_aggregation}). 
\end{remark}

In the following, we say that two vertical portals (cf.\ Def.~\ref{def:portal}) are adjacent to each other, if at least two nodes of either portal are endpoints of a horizontal grid edge.
The notion of walls and gates allows us to define \emph{junction portals} that we require for the next stage of our decomposition. Informally, a junction portal is a portal at which a simple region ``diverges'' into at least three ``tunnels'' (although there are some degenerate cases for such junction portals, where a gate cuts away one of these ``tunnels''). Formally we define (see also \Cref{fig:junction_portal,fig:junctions_case_ii}):

\begin{definition}[Junction Portals]
\label{def:splitting_portals}
    Let $R$ be a simple region. A vertical portal $P$ in $R$ is a \emph{junction} portal if:
    \vspace*{-1mm}
    \begin{enumerate}[itemsep = -1.5mm, label=\roman*.]
        \item $P$ has at least $3$ adjacent portals each intersecting at least 2 distinct walls; or
        \item $P$ is a gate, and has at least $2$ adjacent portals each intersecting at least 2 distinct walls.
    \end{enumerate}
\end{definition}

The idea is to perform portal splitting operations (as described in Appendix \ref{appsec:splitting_portals} Definition \ref{def:splitting_procedure}) on each junction portal and on specific nodes on the junction portal in order to separate these divergent tunnels. More specifically the construction works as follows. 

\begin{definition}[Splitting at Junction Portals]
    \label{def:junction_portal_splitting}
    Let $P$ be a junction portal. We note that \Cref{def:splitting_portals} implies that there are at least $2$ adjacent portals which each intersect multiple distinct walls to the left of $P$ or to the right of $P$, or both. Suppose that portals $P_1, P_2,\dots, P_k$ are such portals with the property of intersecting distinct walls \emph{to the left} of $P$ and let these portals be ordered from north to south (the procedure for those to the right is analogous).\\
    We first conduct a splitting operation at $P$ (Def.\ \ref{def:splitting_procedure} case 1). 

    Then, for each portal $P_i, 1 \leq i < k \m 1$ we choose the bottommost node on $P$ which is adjacent to some node on $P_i$ and split at this node (see Def.\ \ref{def:splitting_procedure} case 2+3). Note that this node must coincide with a hole boundary, since one of the two portals does not extend further south.
    Finally, note that this procedure effectively splits off each region that is bordered by one of the portals $P_i$ since $R$ is simple.
\end{definition}

The goal of this section is the following three lemmas pertaining to the number of regions created, and the correctness and computational complexity of the construction:

\begin{lemma}[Computation of Tunnel Decomposition]
\label{lem:computation_of_junction_splitting}
    Given a decomposition into simple regions, finding and splitting all junction portals can be done in $O(\log n)$ rounds.
\end{lemma}

\begin{lemma}[Correctness of Tunnel Decomposition]
\label{lem:correctness_of_junction_splitting}
    All of the regions resulting from the procedure of finding splitting portals and splitting them are bounded by at most two walls and two gates. 
\end{lemma}

\begin{lemma}[Regions after Splitting Junctions]
\label{lem:number_of_regions_after_junction_splitting}
    After decomposing all junctions into tunnels with the procedure above (Definition \ref{def:junction_portal_splitting}), the resulting number of regions is $O(|\mathcal{H}|)$.
\end{lemma}

We first give a proof of \Cref{lem:computation_of_junction_splitting}, which states that junction portals can be found and split in parallel in $O(\log n)$ rounds of the \HYBRID model:

\begin{proof}[Proof of \Cref{lem:computation_of_junction_splitting}]
    Nodes can determine whether their portal is incident to multiple distinct walls in $O(\log n)$ rounds as follows. Nodes broadcast any incident wall IDs to all other nodes on their portal. As soon as a node receives two wall IDs it instead broadcasts message ``yes'', indicating that the portal is incident to multiple distinct walls. If a node receives the message ``yes'', it stops what it was doing and starts broadcasting the message ``yes'' for $O(\log n)$ rounds. If no node receives two wall IDs or the message ``yes'' in $O(\log n)$ rounds then nodes can conclude that their portal is \emph{not} incident to multiple walls. 
    Each node can then inform their neighbors in one round whether its portal is incident to multiple distinct walls.
    
    In the same way, each portal can determine whether it satisfies the criteria to be a junction portal (all nodes on a portal know whether or not the portal is a gate, and therefore know which case of \Cref{def:splitting_portals} is required).
    Finally, each node $v$ can locally determine whether it is split if its portal is a junction portal. If the node to the right of $v$ (resp. to the left of $v$) is on a portal incident to multiple distinct walls, and the node below and to the right to $v$ is not, then $v$ is to be split. 
    Finally, the splitting can be performed according to Def.\ \ref{def:splitting_procedure} in $O(\log n)$ rounds, per \Cref{lem:splitting_procedure_overhead}.
\end{proof}

Next, we introduce a tree structure which will be useful for the proofs of \Cref{lem:number_of_regions_after_junction_splitting,lem:correctness_of_junction_splitting}. This structure is used for the purpose of proving claims only, we need not actually construct it.

\begin{definition}[Portal Tree Without Cavities]
\label{def:portal_tree_without_cavities}
    Given a simple region $R$, consider the vertical portal graph $\mathcal{P}$ of the region $R$ as in \Cref{def:portal_graph}. As $R$ is simple, by \Cref{lem:portal_graphs_of_simple_regions_are_trees} $\mathcal{P}$ is a tree.
    Remove all leaves of $\mathcal{P}$ which correspond to portals which only intersect one wall. Repeat this until no such leaves remain in $\mathcal{P}$. We call this the \emph{portal tree without cavities} of $R$.
\end{definition}

A visual example of a portal tree without cavities is given in \Cref{fig:portal_tree_without_cavities}.

\begin{figure}[p]
\centering
\label{fig:junction_images}
    \caption{An explanation of junction portals and the portal tree without cavities.}
    \centering
    \begin{subfigure}{0.45\textwidth}
        \centering
        \includegraphics[width=\textwidth,page=1]{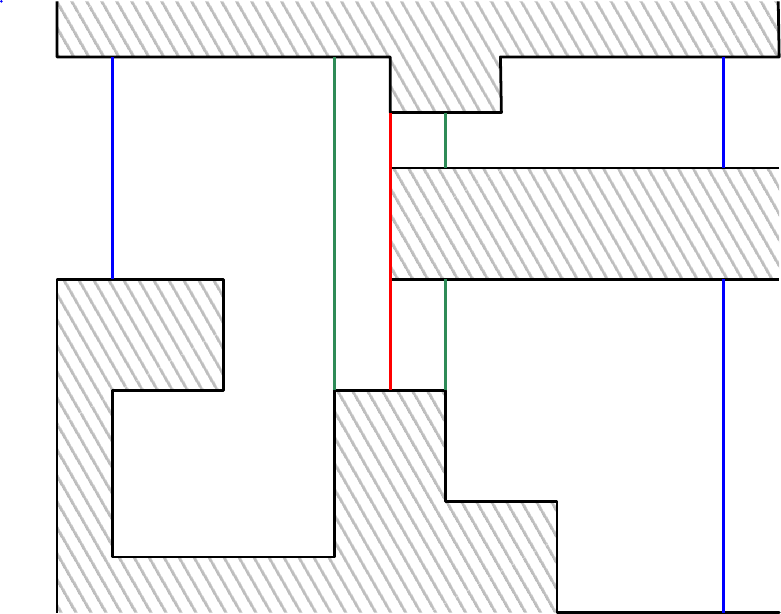}
        \caption{An example of case~(i) of \Cref{def:splitting_portals}. \textcolor{blue}{Blue} portals are pre-existing gates. The \textcolor{red}{red} portal is a junction portal, because it has three adjacent portals (shown in \textcolor{seagreen}{green}) which are incident to multiple distinct walls.\\} 
        \label{fig:junction_portal}
    \end{subfigure}
    \hfill
    \begin{subfigure}{0.45\textwidth}
        \centering
        \includegraphics[width=\textwidth,page=2]{figures/junctions.pdf}
        \caption{An example of a vertical portal tree (corresponding to the example in \Cref{fig:junction_portal}). The \textcolor{orange}{orange} nodes and edges would be removed, giving us a portal tree without cavities. In the remaining graph, \textcolor{blue}{blue} nodes correspond to gates, and the \textcolor{red}{red} node corresponds to a junction portal.}
        \label{fig:portal_tree_without_cavities}
    \end{subfigure}
    
    \vspace*{5mm}
    
    \begin{subfigure}{0.5\textwidth}
        \centering
        \includegraphics[width=0.57\textwidth]{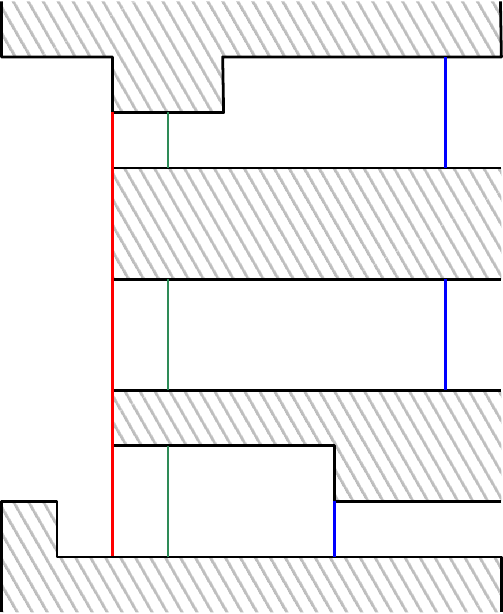}
        \caption{An example of case~(ii) of \Cref{def:splitting_portals}. The \textcolor{blue}{blue} portals are pre-existing gates. The \textcolor{red}{red} portal is also a pre-existing gate, and it is adjacent to three portals which intersect multiple distinct walls (in \textcolor{seagreen}{green}), so it is a junction portal.}
        \label{fig:junctions_case_ii}
    \end{subfigure}
\end{figure}

\begin{lemma}
\label{lem:properties_of_portal_tree_without_cavities}
    Consider a simple region $R$ and an associated portal tree without cavities $\mathcal{P}$ (as in \Cref{def:portal_tree_without_cavities}). The following statements all hold:
    
    \begin{enumerate}
        \item Since we have only removed leaves from the portal tree, $\mathcal{P}$ remains connected (or empty).
        \item All portals corresponding to nodes in $\mathcal{P}$ are incident to multiple distinct walls.
        \item All remaining leaves in $\mathcal{P}$ correspond to gates.
        \item The portal tree $\mathcal{P}$ contains all junction portals in $R$.
        \item Iff a node in $\mathcal{P}$ has degree at least $3$, or corresponds to a gate and has degree at least $2$, then it corresponds to a junction portal.
    \end{enumerate}
\end{lemma}
\begin{proof}
    We take the statements in order:
    \begin{enumerate}
        \item Removing leaves does not disconnect a tree.
        
        \item We will prove this by contradiction. Fix a portal $P$ which corresponds to a node in $\calP$, and suppose it is only incident to one wall, $W$.
        Let $p_\uparrow$ and $p_\downarrow$ be the top and bottom points of $P$; both are incident to $W$. Since region boundaries are alternating sequences of walls and gates, we must have that $p_\uparrow$ and $p_\downarrow$ are connected by some sub-path of $W$. Call this sub-path $W'$ and suppose wlog that $W'$ meets $p_\uparrow$ and $p_\downarrow$ from the left. Consider $R'$: the subregion of $R$ bounded by $P$ and $W'$. Clearly $R'$ must be simple, and therefore no other hole can appear in $R'$%
        . 
        
        Note that $P$ cannot have degree $1$ in $\calP$, otherwise we would have removed it from $\calP$. So $P$ must have at least two adjacent portals which are incident to multiple walls: let $P_1$ and $P_2$ be two such portals. Furthermore $P_1, P_2$ must be to the right of $P$ since the cavity $R'$ is to its left. 
        If either $P_1$ or $P_2$ of $P$ would start and end with nodes in $W$ this would be the beginning of another cavity $R''$, and by the same argument in the previous paragraph would have already been removed from $\calP$. Therefore, $P_1$ and $P_2$ must have at least one wall other than $W$ as one of their endpoints.
        
        In order that $P$ can be incident to \emph{two} portals to its right, it must have an the interior of a hole immediately to the right of $P$ that ``separates'' these. By our assumption, since the boundary of this hole intersects $P$ it must belong to the wall $W$. And so $P_1$ and $P_2$ must both have one endpoint in $W$ to the right of $P$ as well. Combined with the above statement this implies that $P_1, P_2$ each have exactly one endpoint in $W$ (which is to the right of $P$) and another endpoint in another wall. Since $W$ must be contiguous, the section of $W$ to the right of $P$ must connect to either $p_\uparrow$ or $p_\downarrow$ and thus intersect both endpoints of at least one of the two portals $P_1,P_2$. This is a contradiction to the fact that both $P_1$ and $P_2$ must have an endpoint in a wall which is not $W$. 
        
        \item All remaining leaves must be incident to multiple distinct walls, by the previous statement.
        For a contradiction, fix a leaf which corresponds to a portal $P$ but is not a gate. Note that $P$ is only incident to one portal which remains in $\mathcal{P}$, suppose wlog that it is to the right of $P$.
        
        If there is no portal to the left of $P$ in $R$, then $P$ is a section of $W_1$ (by our assumption that $P$ is not a gate). Some node along $P$ must be part of a different wall $W_2$, otherwise $P$ would have been removed. 
        Let $p_\uparrow$ be the node to the right of the topmost node on $P$ which intersects $W_2$, and let $p_\downarrow$ be the node to the right of the bottom-most node on $P$ which intersects $W_2$. $p_\uparrow$ and $p_\downarrow$ are both on different vertical portals to the right of $P$, and both of their portals intersect $W_1$ and $W_2$. This contradicts that $P$ is a leaf in $\calP$.
        
        If there is a portal to the left of $P$, then it must be incident to only one wall, $W$ (otherwise $P$ would not be a leaf). If both endpoints of $P$ are incident to $W$, then the logic from the previous paragraph applies: some node on $P$ must be incident to a different wall, implying that $P$ has two portals to its right which intersect distinct portals. 
        So suppose this is not the case, and that one of the endpoints of $P$ (wlog the top-most endpoint) is incident to a different wall $W'$. Let $v$ be the topmost point on $P$ which intersects $W$. The point immediately to the left of $v$ must lie on a vertical portal which intersects $W$ (since $v$ intersects $W$) and $W'$. This is a contradiction, since by assumption no portal to the left of $P$ is incident to multiple distinct walls.

        \item We argue that nodes which we removed could only be adjacent to at most one portal which is incident to multiple walls, and therefore fail to satisfy the definition in \Cref{def:splitting_portals}.
        This is the case because nodes only have degree $1$ when we remove them: they may previously have been connected to nodes which we later removed, but the nodes which we later removed must only have been incident to one wall.
        
        \item This claim immediately follows from statement 2, since this mirrors the definition of junction portals given in \Cref{def:splitting_portals}.
    \end{enumerate}
\end{proof}

We now prove that the splitting procedure is correct (\Cref{lem:correctness_of_junction_splitting}).

\begin{proof}[Proof of \Cref{lem:correctness_of_junction_splitting}]
    Consider, for the sake of contradiction, a region $R$ has more than three gates after all junction portals have been split. Consider the portal tree described in \Cref{def:portal_tree_without_cavities}.
    
    If this tree has at least $3$ leaves, then it must (graph-theoretically) have a node with degree at least $3$. But this node is a junction portal by \Cref{lem:properties_of_portal_tree_without_cavities}. If this tree has only $2$ leaves, it must have an internal node (with degree at least $2$) which corresponds to a gate, due to our assumption that $R$ has three gates. But by \Cref{lem:properties_of_portal_tree_without_cavities}, this must also have been a junction portal. Therefore both cases lead to a contradiction and no such region $R$ can exist.
\end{proof}

Next, we observe that splitting a junction portal can lead to a ``dead-end'' region, bounded by one gate and one wall. Note that a junction portal which creates a dead-end region must have both endpoints incident to the same wall. An example of this case is given in \Cref{fig:dead-end}. Note that each junction portal can only create one dead-end region, since if there were a dead-end region on each side, then there would be no adjacent portals which connect two distinct walls (since each region is simple).
Note also that if there is a dead-end region on one side, then, since the region is simple and therefore consists of alternating series of walls and gates, there can be no portal on that side which connects two distinct walls. Therefore the junction portal is not split on that side, and at most one dead-end region can be created on that side.

\begin{observation}
\label{obs:dead_ends}
    When splitting at a junction portal, this can result in at most one ``dead-end'' region being created: that is, a region which is bounded by one gate and one wall.
    
    This happens when the junction portal only has one portal to its left (say) which only intersects one wall. This cannot happen on both sides because then it would not be a junction portal.
\end{observation}
\begin{proof}
    Let $P$ be a junction portal, $W$ be a wall, and let there be a dead-end region $R$ to the left (w.l.o.g.) of $P$ which is bounded by $P$ and $W$. Since both ends of $P$ are on a wall, they must be on $W$.

    First, notice that $R$ must be the only region to the left of $P$, since, if there was a region to the left of $P$ which intersected two distinct walls, then it would not be a dead end, and we would not cut $P$ between this region and $R$, so $R$ would be a part of that region (we would call it a ``cavity'').
    
    If there was a dead-end region also to the right of $P$, then since the endpoints of $P$ are on $W$, this region must also be bounded solely by $P$ and $W$. But we reach a contradiction, since this means that there could not be any walls aside from $W$ in the region to the right of $P$, since the regions are simple and region is only bounded by $P$ and $W$, and therefore $P$ could not be a junction portal.
\end{proof}

\begin{figure}
    \centering
    \label{fig:dead-end}
    \includegraphics[width=0.35\textwidth]{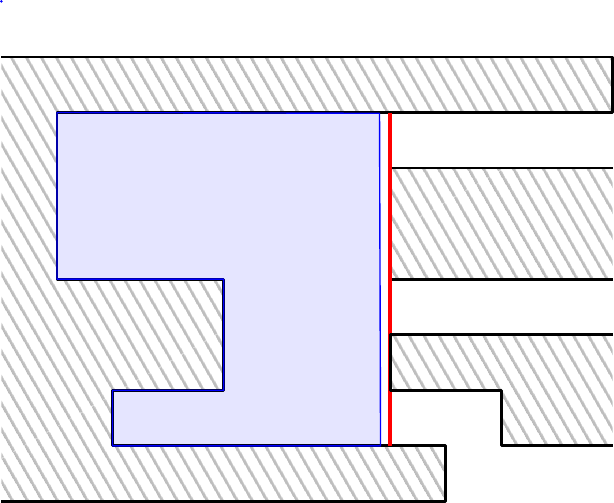}
    \caption{An example of a dead-end region, in \textcolor{blue}{blue}. The junction portal (in \textcolor{red}{red}) splits to the right, because there are at least three neighboring portals which connect distinct walls. But it does not split to the left, and so a dead-end region is formed.}
\end{figure}

Next, we prove \Cref{lem:number_of_regions_after_junction_splitting}: that the number of regions resulting from splitting junction portals is not too large.

\begin{proof}
    We argue that a $k$-junction introduces at most $\bigO(k)$ new regions and splitting portals. The claim then follows since the sum of gates of all junctions are in $\bigO(|\mathcal H|)$. 
    
    Consider a region $R$ which was a $k$-junction (i.e.~was bounded by $k$ gates) after the decomposition into simple regions, and consider its portal tree without cavities $\mathcal{P}$ as in \Cref{def:portal_tree_without_cavities}.
    
    By Lemma \ref{lem:properties_of_portal_tree_without_cavities} (3.) leaves of $\mathcal P$ must be gates of $R$ thus $\mathcal{P}$ has $\ell \leq k$ leaves. Note that the remaining $k-\ell$ gates of $R$ correspond to internal nodes of $\mathcal P$ (recall that these correspond to case (ii) of Def.\ \ref{def:splitting_portals}).
    The $\ell$ leaves imply that there are, in the worst case, $\ell - 2$ of degree at least $3$ in $\mathcal{P}$. Each of these is incident to one dead end by \Cref{obs:dead_ends}. Then, each of these junction portals is incident to $4$ regions created by splitting at junction portals, giving $4\ell - 8$ regions.
    Suppose (again in the worst case), that the $k-\ell$ internal nodes which are gates only have degree $2$. These gates are incident to $2$ regions each created by splitting at these junction portals, giving $2k-2\ell$.
    Summing these, and recalling that $\ell \leq k$, gives that at most $\bigO(k)$ regions are created from a $k$-junction, as required.
\end{proof}

Finally, we conclude with a lemma which will be useful for next section: during the tunnel decomposition we get rid of almost all node copies within a tunnel, except for the two copies of a splitting node.

\begin{lemma}
    \label{lem:node_copies_after_split}
    In the tunnel decomposition obtained above, there can be at most one pair of nodes that are copies of each other in each region. If that is the case, this pair of nodes forms the two point shaped gates in their tunnel region.
\end{lemma}
\begin{proof}
    As the region decomposition is simple, we can not have any pair of horizontally copied nodes in the same region as these would need to be connected by a path encircling the hole north- or south-adjacent to their corresponding portal. Hence we only need to consider the vertically copied nodes introduced by a horizontal node split. As all splits added in this section divide regions that are already simple, none of the resulting vertically copied nodes can be in the same region. Hence all of them must be created during the splitting into simple regions described in \autoref{def:construction-Gamma}. At that point we can still have more than one pair of vertically copied nodes, if the leftmost node of two different hole boundaries align. As those two holes induce two different walls and there must be at least one more wall, which may be both above and below those holes, we must have a junction portal splitting the two pairs of vertically copied nodes.
\end{proof}

\subsection{Path-Convex Decomposition}
\label{sec:path_convex_decomp}

Finally, we make the tunnel regions that we obtained in the previous section \convex (cf.~Def.~\ref{def:convex}), by splitting them at appropriate portals. 
For any two nodes in such a tunnel with a shortest path in $\Gamma$ that travels outside the tunnel, the additional portals with which we separate the region will also separate that pair of nodes. We will also see that it is sufficient to further separate tunnels at a constant number of vertical or horizontal portals.

The main question that we answer in this section is \textit{where} the tunnels shall be split in order to make them \convex. While the construction of the splits is not too complicated, the main challenge is the proof of correctness, i.e., all ``offending'' node pairs whose shortest path runs outside of the region must be separated. We answer this question in stages. 

First, we impose the assumption that we have a tunnel region $T$ with gates that are single nodes $g$ and $g'$, see \Cref{fig:tunnel_splitting_point_gates}. We show that we can split $T$ into \convex regions using only a horizontal and a vertical portal $P_x$ and $P_y$. Roughly, if $d_x = d_{x,T}(g,g')$ is the horizontal distance from $g$ to $g'$ in $T$ then all nodes which are at distance $\frac{d_x}{2}$ to $g$ and $g'$ will be part of $P_x$ (details in Def.~\ref{def:splititng_portals_tunnel_node_gates}), which we show forms a vertical portal (cf., Lemma \ref{lem:splitting_portal_property}). The horizontal portal $P_y$ is defined symmetrically. We split $T$ at $P_x$ and $P_y$ (cf., Fig.\ \ref{fig:splitting_portal}).

As first part of our proof we show that two nodes $u,v$ that end up in the same region $R$ after splitting at $P_x$ and $P_y$ that both lie on (possibly different) shortest paths from $g$ to $g'$ have a shortest path within $R$ (see Lemma \ref{lem:split_region_diameter_special}). Loosely speaking, this covers the well behaved case where $u,v$ are not located in ``cavities'' of $T$. We then move on to this harder case and show the same claim for any pair of nodes $u,v \in R$ (see Lemma \ref{lem:split_region_diameter_general}) by using fundamental property of grid graphs (see Lemma \ref{lem:shortest_path_property}) that lets us fall back to the easier case.

In the last stage of the construction, we remove the assumption that gates are point-shaped. We distinguish two cases, depending on whether there is a horizontal portal connecting the two gates (see Definition \ref{def:splititng_portals_tunnel} and Figure \ref{fig:tunnel_splitting_cases}). We then give a decomposition procedure for each case, whereas the first case (a) is well behaved and the second case (b) has the property that for any pair of nodes in a region in the ``middle part'' a hypothetical shortest path that goes outside $T$ can always leave and enter $T$ through two nodes $g,g'$ in the gates $G,G'$, which allows us to fall back to the proof of correctness for this case and split the ``middle part'' at $P_x,P_y$.

We conclude with lemmas pertaining to the number of regions created, the correctness, and the computational complexity.

\begin{lemma}[Regions after Convex Decomposition]
\label{lem:regions_per_tunnel}
    The construction in Definition \ref{def:splititng_portals_tunnel} produces at most 10 regions for every tunnel which we split.
\end{lemma}

\begin{lemma}[Correctness of Path-Convex Decomposition]
    \label{lem:convex_tunnel}
    The construction in Definition \ref{def:splititng_portals_tunnel} (illustrated in Figure \ref{fig:tunnel_splitting_cases}) decomposes $T$ into \convex regions.
\end{lemma}

\begin{lemma}[Computation of Path-Convex Decomposition]
    \label{lem:convex_region_complexity}
    The construction of the portals in Def.\ \ref{def:splititng_portals_tunnel} and \ref{def:splititng_portals_tunnel_node_gates} and the described splitting of the tunnel region $T$ takes $\bigO(\log n)$ rounds.
\end{lemma}

The lemmas from the three sections of the region decomposition can then be combined to obtain the proof for Theorem \ref{thm:convex_decomposition}: that our construction is efficient and results in a \convex region decomposition that is linear in the number of holes.

\begin{proof}[Proof of Theorem \ref{thm:convex_decomposition}.]
    We obtain a simple region decomposition in the claimed number of rounds by \Cref{lem:simple_decomp}. These simple regions can be transformed into tunnel regions, i.e.~regions with at most two gates, due to Lemma \ref{lem:correctness_of_junction_splitting}. Finally, we have shown that we can further break up tunnels into  convex regions by Lemma \ref{lem:convex_tunnel}.
    The computation of this decomposition requires $O(\log n)$ rounds in the \HYBRID model, by \Cref{lem:computing_simple_decomp,lem:computation_of_junction_splitting,lem:convex_region_complexity}.
    The number of regions is $\bigO(|\calH|)$ by Lemmas \ref{lem:simple_decomp}, \ref{lem:number_of_regions_after_junction_splitting} and \ref{lem:regions_per_tunnel}.
\end{proof}

It remains to prove \Cref{lem:convex_tunnel,lem:convex_region_complexity}.
We will subdivide the approach into several stages, where we iteratively remove some of the simplifying assumptions by reducing the more general case to the simplified one.

\begin{figure}
    \centering
    \includegraphics[scale=0.6]{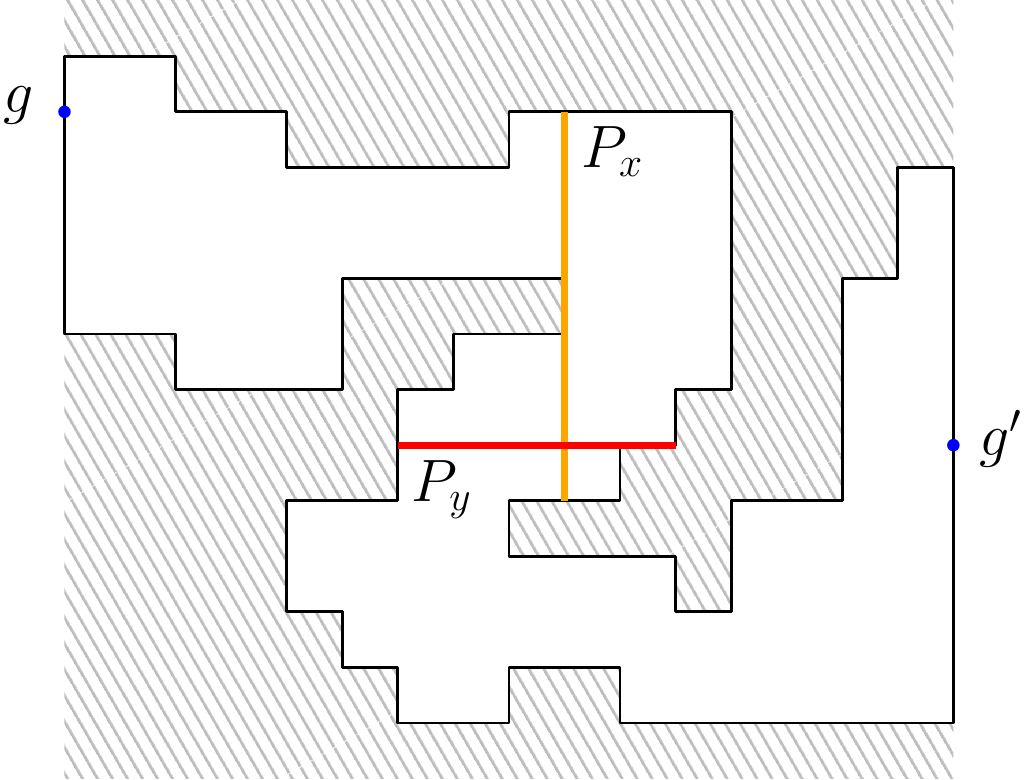}
    \caption{Splitting a tunnel with gates consisting of single nodes.}
    \label{fig:tunnel_splitting_point_gates}
\end{figure}

In the first stage we assume the special case that the two gates of a tunnel are points. We use the following construction to make $T$ \convex. We start with the following definition. A visual depiction of the splitting procedure described in \Cref{def:splititng_portals_tunnel_node_gates} is given in \Cref{fig:tunnel_splitting_point_gates}.

\begin{definition}[Splitting Portals for tunnels with point shaped gates]
    \label{def:splititng_portals_tunnel_node_gates}
    Let $T$ be a tunnel region with two gates $g,g' \in T$ each consisting of a single node. Let $d_x := d_x(g,g')$ and $d_y := d_y(g,g')$, which are well defined since a tunnel $T$ is also simple. Let $P_x = \{v \in T \mid d_x(g, v) = \lceil \tfrac{d_x}{2} \rceil, d_x(g', v) = \lfloor \tfrac{d_x}{2} \rfloor \}$ and analogously $P_y = \{v \in T \mid d_y(g, v) = \lceil \tfrac{d_y}{2} \rceil, d_y(g', v) = \lfloor \tfrac{d_y}{2} \rfloor \}$.
\end{definition}

\begin{lemma}
    \label{lem:splitting_portal_property}
    The set $P_x$, respective $P_y$ (see Def.\ \ref{def:splititng_portals_tunnel_node_gates} and Figure \ref{fig:tunnel_splitting_point_gates}) forms a vertical, respective a horizontal portal in the tunnel region $T$ (the simplicity property of tunnels suffices here). Each portal is a node separator  w.r.t.\ $g$ and $g'$.
\end{lemma}

\begin{proof}
    By Definition \ref{def:portal}, we have to show that $P_x,P_y$, are maximal connected components in the graph induced by horizontal or vertical edges, respectively. Due to symmetry, it suffices to consider $P_x$. Let $a \in P_x$ and let $P_a$ be the maximally vertically connected component of $a$. Any vertically adjacent point to $a$ must be also be in $P_x$ since the horizontal distance to $g$ and $g'$ does not change, thus $P_a \subseteq P_x$.
    
    Assume there is a node $b \in P_x \setminus P_a$. Note that since the horizontal distances of $a,b$ to $g$ and $g'$ add up to exactly $d_x(g,g')$, $a$ and $b$ must lie on $gg'$ paths $\Pi_a,\Pi_b$  that are shortest w.r.t.\ the horizontal distance. The path $\Pi_b$ can not visit any point in $P_a$ since otherwise either $b \in P_a$ or $b \notin P_x$ (for the latter, consider the definition of $P_x$). Due to vertical maximality, $P_a$ meets the boundary of $T$ at both ends. Thus $\Pi_a$ and $\Pi_b$ must together enclose some boundary nodes of $T$, which means $T$ has a hole, a contradiction.
    
    It remains to show that $P_x$ is a node separator. Assuming that there would be a $gg'$-path that does not intersect $P_x$ would again induce a hole in $T$, since $P_x$ is maximal vertically connected.
\end{proof}

\begin{remark}
    \label{rem:tunnel-regions-point-gates}
    The portals $P_x$ and $P_y$ characterized in Definition \ref{def:splititng_portals_tunnel_node_gates} result in regions consisting of nodes categorized by their location with respect to $P_x$ and $P_y$. We obtain at most four regions. Each region is formed by the set of nodes that are on the same side of both node separators $P_x$ and $P_y$. Formally, we assume that nodes on the portals $P_x$ and $P_y$ are part of all regions which they are adjacent to. On the algorithmic level, we split $T$ at $P_x$ and $P_y$, see Figure \ref{fig:tunnel_splitting_point_gates} where nodes on $P_x$ and $P_y$ are ``copied'' with one copy in each resulting region, as in \Cref{def:splitting_procedure}.
\end{remark}

We identify the following easy case.

\begin{fact}
    \label{fct:convexity_easy_case}
    Given a region $R$ of some grid graph $\Gamma$ that is bordered by a vertical and horizontal gate $G,G'$ which intersect in a node. Then there is \emph{no} shortest path in $\Gamma$ between two points in $R$ that leaves and then reenters through any of those gates as such a path can always be made shorter. In particular, this implies that a region that has just one or at most two gates that intersect, must be \convex.
\end{fact}

Another important property is the following.

\begin{lemma}
    \label{lem:convexity_retained}
    Let $R$ be a \convex region and let $P$ be a portal through $R$. If we split $R$ at $P$ the resulting regions remain \convex.
\end{lemma}

\begin{proof}
    Let $u,v \in R$ both be in one of the resulting regions after splitting at $P$. Since $R$ was \convex there is a shortest $uv$-path $\Pi$ inside $R$. Assume $\Pi$ would cross $P$, i.e., leave the new region of $u,v$ and reenter it through nods on $P$, then $\Pi$ could clearly be shortened along $P$, which contradicts the fact that $\Pi$ is shortest.
\end{proof}

We continue with the first step of the overall proof, namely grid nodes $u,v \in T$ that lie on a shortest path from $g$ to $g'$. In particular, this means that $u,v$ are not located in ``cavities'' of the tunnel region, which is a more complicated case, which we will carefully break down into this simpler case afterwards. We start with the following technical lemma that shows that two points that are on shortest $gg'$-paths and on $P_x$ and $P_y$, are relatively close. Recall that $d_x := d_x(g,g')$ and $d_y := d_y(g,g')$.

\begin{lemma}
    \label{lem:split_region_portal_diameter}
    Let $T$ be a tunnel region with two gates $g,g'$ that are single nodes. Let $\Pi,\Pi'$ be two shortest $gg'$-paths. Let $u \in P_x \cap \Pi$ and $v \in P_y \cap \Pi'$ (interpreting $\Pi,\Pi'$ as set of nodes on the path; for $P_x, P_y$ see Def.\ \ref{def:splititng_portals_tunnel_node_gates}). Then $d_x(u,v) \leq \lceil\tfrac{d_x}{2}\rceil$ and $d_y(u,v) \leq \lceil\tfrac{d_y}{2}\rceil$.
\end{lemma}\jw{a picture would probably help somewhere here (i.e. here or at a related lemma)}
\ps{I agree, maybe I'll still do one, but priority for making pictures for the appendix is low}

\begin{proof}
    By symmetry, it suffices to prove $d_x(u,v) \leq \lceil\tfrac{d_x}{2}\rceil$. Let $u' \in P_x \cap \Pi'$. Since $T$ is simple, we have $|\Pi|_x = |\Pi'|_x = d_x$ by Lemma \ref{lem:length_of_shortest_paths_in_simple_regions}. In particular, $\Pi'_{gu'} = \lceil\tfrac{d_x}{2}\rceil$ for the sub path from $g$ to $u'$ (by the definition of $P_x$ and since $\Pi'_{gu'}$ is shortest too). Furthermore, $P_x$ separates $g,g'$ by Lemma \ref{lem:splitting_portal_property}. Assume that $v$ is on the same side as $g$, otherwise the subsequent estimations admit swapping $g$ and $g'$. Then
    \[
        d_x(u,v) \leq \underbrace{d_x(u,u')}_{=0 \text{, as }u,u'\in P_x} \!\!\!+\, d_x(u',v) = \underbrace{d_x(u',v)}_{\text{both on }\Pi'} \leq d_x(u',g) = |\Pi'_{gu'}|_x = \lceil\tfrac{d_x}{2}\rceil.
    \]
    Note that the requirement $v \in P_y$ is only needed for the proof for the vertical distance.
\end{proof}

The lemma above already shows that leaving $T$ is \textit{not} required to get from $u$ to $v$ on some shortest path, given that both points are on shortest $gg'$-paths and on $P_x$ and $P_y$. We will now generalize the latter condition for two such points within the same region of $T$ after splitting at $P_x,P_y$.

\begin{lemma}
    \label{lem:split_region_diameter_special}
    Let $T$ be a tunnel region with two gates $g,g'$ that are single nodes. Let $\Pi,\Pi'$ be two shortest $gg'$-paths. Let $u$ on $\Pi$ and $v$ on $\Pi'$ and $u, v \in R$ for some region $R$. Then $d_{R,x}(u,v) \leq \lceil\tfrac{d_x}{2}\rceil$ and $d_{R,y}(u,v) \leq \lceil\tfrac{d_y}{2}\rceil$.
\end{lemma}

\begin{proof}
    We show $d_x(u,v) \leq \lceil\tfrac{d_x}{2}\rceil$, the vertical distance case is analogous. Note that $\Pi$ and $\Pi'$ must intersect $P_x$ and $P_y$ in points $x_1 \in P_x$, $y_1 \in P_y$, $x_2 \in P_x $ and $y_2 \in P_y$, respectively. 
    
    First we show that w.l.o.g., we can assume that $y_1=y_2$. If $y_1 \neq y_2$, then, by Lemma \ref{lem:compose_paths_with_line} one of the two paths $\Pi_{gy_1} \circ \overline{y_1y_2} \circ \Pi'_{y_2g'}$ or $\Pi'_{gy_2} \circ \overline{y_2y_1} \circ \Pi_{y_1g'}$ is a shortest $uv$-path as well; let this path be $\widetilde\Pi$. Since $u,v \in R$ are on the same side w.r.t.\ the separator $P_y$, we have that $u$ or $v$ is on $\widetilde \Pi$. We replace $\Pi := \widetilde \Pi$, if $u$ is on $\widetilde \Pi$, or $\Pi' := \widetilde \Pi$, if $v$ is on $\widetilde \Pi$. Note that the preconditions of this Lemma still hold. \textit{Additionally} we obtain the property that $\Pi$ and $\Pi'$ both intersect in some point $R \cap P_y$.
    
    Note that, $\Pi$ and $\Pi'$ do not enter, exit and subsequently reenter $R$ (see Fact \ref{fct:convexity_easy_case}). This means that there are points $u_1, u_2$, $v_1, v_2$, such that $\Pi_{u_1 u_2}$ and $\Pi'_{v_1 v_2}$ stay within $R$.
    The points of entry and exit into $R$ can only be on $g,g',P_x,P_y$. The only case where the two entry (or exit points) of $\Pi,\Pi'$ into $R$ may differ, is if they are located on $P_x$, due to our previous statement and due to the fact that $g,g'$ are points. We will assume that $u_1 \neq u_2$ is possible but $v_1 = v_2$ (for instance $u_1, u_2 \in P_x$ and $v_1,v_2 \in P_y$ or $v_1,v_2 = g'$) since the following proof allows to analogously switch roles of the point pairs. More importantly, we observe that $d_x(u_1, u_2) = 0$, meaning they are ``the same'' in terms of vertical distance to other points.
    
    

    We split the sub paths $\Pi_{x_1,u_2}$ and $\Pi'_{v_1,v_2}$ into the following segments:
    \begin{align*}
        |\Pi_{u_1,u_2}|_x & = \alpha + \beta, \text{ where } \alpha := d_x(u_1,u), \beta := d_x(u,u_2)\\ 
        |\Pi'_{v_1,v_2}|_x & = \gamma + \delta, \text{ where }
        \gamma := d_x(v_1,v), \delta := d_x(v,v_2)
    \end{align*}
    Since $u_1,v_1 \in P_x$ and $v_2,u_2 \in P_y$, we can apply Lemma \ref{lem:split_region_portal_diameter} thus:
    \begin{equation*}
        \alpha + \beta \leq \lceil\tfrac{d_x}{2}\rceil \text{ and }
        \gamma + \delta \leq \lceil\tfrac{d_x}{2}\rceil
    \end{equation*}
    This implies 
    \begin{equation}
        \label{eq:sub_path_estimations}
        \alpha + \gamma \leq \lceil\tfrac{d_x}{2}\rceil \text{ or } \beta + \delta \leq \lceil\tfrac{d_x}{2}\rceil
    \end{equation}
    as otherwise one of the inequalities of \eqref{eq:sub_path_estimations} would be false. 
    If the former of the two inequalities \eqref{eq:sub_path_estimations} is true then applying the triangle inequality gives
    \[
        d_{x,R}(u,v) \leq |\Pi_{u,u_1}|_x + \underbrace{d_x(u_1, v_1)}_{=0} + |\Pi_{v_1,v}|_x = \alpha + \gamma \leq \lceil\tfrac{d_x}{2}\rceil\vspace*{-3mm}
    \]
    If the latter of inequalities \eqref{eq:sub_path_estimations} is true then we get\vspace*{-2mm}
    \[
        d_{x,R}(u,v) \stackrel{u_2=v_2}{\leq} |\Pi_{u,u_2}|_x + |\Pi_{v_2,v}|_x = \beta + \delta \leq \lceil\tfrac{d_x}{2}\rceil \qedhere
    \]  
\end{proof}

In the next stage, we show that we have path convexity for arbitrary points $u,v$ in a given region $R$ (but still gates consist of a single node). We will show that a shortest $uv$-path that leaves $R$ will necessarily also have to leave $T$, and we will see that it can not be shorter than a path that stays inside $T$, partly relying on previous lemmas.
One caveat is that after splitting the tunnel $T$ the boundaries to the outer hole can be of arbitrary shape, in particular, we could have highly complex cavities attached to such regions, like, e.g., spirals.

To show that such structures do not matter in terms of path convexity, we exploit a fundamental property of grid graphs, namely that for any three of grid nodes $a,b,v$, there is a shortest $ab$-path $\Pi$ and a point $v'$ on $\Pi$ such that any shortest $av'v$- or $bv'v$-path is a shortest $av$- or $bv$-path, respectively (see Appendix \ref{sec:grid_properties} Lemma \ref{lem:shortest_path_property}). This allows us to fall back to our previous result in Lemma \ref{lem:split_region_diameter_special} where $u,v$ lie on shortest $gg'$-paths.

\begin{lemma}
    \label{lem:split_region_diameter_general}
    Let $T$ be a tunnel region with two gates $g,g'$ that are single nodes. Let $u, v \in R$ for some region $R$ of $T$ after splitting at $P_x,P_y$. Let $\widetilde\Pi$ be a $uv$-path that is partially outside of $R$. Then $|\widetilde\Pi| \geq d_R(u,v)$ (which implies path convexity).
\end{lemma}

\begin{proof}
    We will prove that for any such path $\widetilde \Pi$ there is one that stays within $R$ and is at most as long as $\widetilde \Pi$.
    First assume that $\widetilde \Pi$ leaves $R$ but not $T$, i.e., $\widetilde \Pi$ has a grid node that is not in $R$ but none that is not in $T$. This means $\widetilde \Pi$ does not cross through $g$ or $g'$. Then, by Fact \ref{fct:convexity_easy_case}, there must be a shorter path that stays within $R$.
    Furthermore, $\widetilde \Pi$ must be a simple path (no node visited twice) since otherwise $\widetilde \Pi$ could be made shorter by short-cutting the loop.
    
    This implies that if $\widetilde \Pi$ leaves $T$ at $g$ it can only reenter at $g'$ or vice versa, which is the case left to consider. Let $\widetilde\Pi_T$ be the subpaths of $\widetilde\Pi$ inside $T$. W.l.o.g., we will assume that $\widetilde\Pi_T$ connects $g$ with $u$ and $g'$ with $v$ (otherwise we switch the roles of $g,g'$), we refer to this as property (A).
    We show that for such $\widetilde\Pi$ we can fall back on Lemma \ref{lem:split_region_diameter_general} by employing Lemma \ref{lem:shortest_path_property} in order to prove $|\widetilde\Pi| \geq d_R(u,v)$. Let $\Pi,\Pi'$ be two shortest $gg'$-paths that are \textit{closest} to $u$ and $v$, respectively, and let $u'$ and $v'$ be two nodes on $\Pi,\Pi'$ that are closest to $u$ and $v$, respectively. By Lemma \ref{lem:shortest_path_property} we have $d_T(g,u) = |\Pi_{gu'}|  + d_R(u',u)$ and $d_T(g',v) = |\Pi'_{g'v'}|  + d_R(v',v)$.
    
    Note that there is always some shortest $gg'$-path that intersects $R$. This is clear if $P_x,P_y$ do not intersect, since each region of $T$ has to be crossed. If they do intersect, then by definition of $P_x,P_y$ there is a shortest $gg'$-path that crosses that intersection and thus $R$ (note that we consider the nodes in the portals $P_x$ and $P_y$ that delimit a region $R$ as part of that region, see Remark \ref{rem:tunnel-regions-point-gates}). Therefore, the definition of $u',v'$ implies that both are in $R$. This, in turn, implies that $P_x$ and $P_y$ are each crossed by \textit{one} of the two path segments $\Pi_{gu'}, \Pi'_{g'v'}$ (could be the same path segment for both), which by Definition \ref{def:splititng_portals_tunnel_node_gates} of $P_x$ and $P_y$ means that we have property (B): \smash{$|\Pi_{gu'}| + |\Pi'_{g'v'}| \geq \lfloor\frac{d_x}{2}\rfloor + \lfloor\frac{d_y}{2}\rfloor$}.
    \begin{align*}
        |\widetilde\Pi_T|  \stackrel{\text{(A)}}{\geq}  \hspace*{2.5mm} & d_T(g,u) + d_T(g',v)\\
         \stackrel{\text{Lem.} \ref{lem:shortest_path_property}}{=} & |\Pi_{gu'}|  + d_R(u',u) +  |\Pi'_{g'v'}| + d_R(v',v)\\
         \stackrel{\text{(B)}}{=} \hspace*{2.5mm} & \lfloor\tfrac{d_x}{2}\rfloor + \lfloor\tfrac{d_y}{2}\rfloor  + d_R(u',u) + d_R(v',v)\\
         = \hspace*{3mm} & \lceil\tfrac{d_x}{2}\rceil + \lceil\tfrac{d_y}{2}\rceil - \tfrac{\mathbbm{1}_{[d_x \text{ odd}]} + \mathbbm{1}_{[d_y \text{ odd}]}}{2} + d_R(u',u) + d_R(v',v)\\
         \stackrel{\text{Lem.}\ref{lem:split_region_diameter_special}}{\geq} & d_{R,x}(u',v')  + d_{R,y}(u',v') + d_R(u',u) + d_R(v',v) - \tfrac{\mathbbm{1}_{[d_x \text{ odd}]} + \mathbbm{1}_{[d_y \text{ odd}]}}{2}\\
         \stackrel{\text{Lem.}\ref{lem:length_of_shortest_paths_in_simple_regions}}{=} & d_{R}(u',v')  + d_R(u',u) + d_R(v',v) - \tfrac{\mathbbm{1}_{[d_x \text{ odd}]} + \mathbbm{1}_{[d_y \text{ odd}]}}{2}\\
         \stackrel{\Delta \text{-ineq.}}{\geq} & d_{R}(u,v) - \tfrac{\mathbbm{1}_{[d_x \text{ odd}]} + \mathbbm{1}_{[d_y \text{ odd}]}}{2}
    \end{align*}
    Note that we lower bounded the part of $\widetilde\Pi$ inside $T$. If $g$ and $g'$ are distinct points then $|\widetilde\Pi| \geq |\widetilde\Pi_T| + 1 \geq d_{R}(u,v)$. If $g$ and $g'$ happen to be unconnected copies of a node (see Lemma \ref{lem:node_copies_after_split}) then any $|\widetilde\Pi|_x, |\widetilde\Pi|_y$ must both be even, thus $|\widetilde\Pi_T| \geq d_{R}(u,v)$.
\end{proof}

Finally, we generalize our approach to handle tunnels $T$ with gates $G,G'$ that are vertical portals. Note that when we created the tunnels in \Cref{sec:tunnel_decomposition} we only used vertical gates to delimit those. Also, note that our construction to make tunnels with single node gates \convex can not be used out of the box, in particular since $P_y$ from Definition \ref{def:splititng_portals_tunnel_node_gates} is no longer well defined. 

We describe here two fundamental cases of the shape of $T$ and give a construction to make convex in either case, see \Cref{fig:tunnel_splitting_cases}.

\begin{fact}
    \label{fct:tunnel_case_distinction}
    Let $T$ be a tunnel with gates $G, G'$. Let $g_1, \dots ,g_\ell$ be the set of nodes on $G$ from bottom to top (i.e., sorted by ascending $y$-coordinate). Let $d(i) := d_{T,y}(g_i,G')$ be the vertical distance in $T$ from $g_i$ to $G'$. Consider $d(i)$ the as a function over $i \in [\ell]$. We have intervals $[1, \ldots, i_1\!-\!1]$, $[i_1, \ldots, i_2]$ and $[i_2\!+\!1, \ldots, \ell]$ where $d(i)$ is monotonously decreasing, assuming its minimum and monotonously increasing, respectively. Note that the first and third interval can be empty ($i_1=1$ or $i_2 = \ell$) and $i_1=i_2$ is possible, in which case we define $i':= i_1$. This property of $d(i)$ is due to the fact that $T$ is simple and $G$ and $G'$ are vertical portals (our claims are illustrated by Figure \ref{fig:tunnel_splitting_cases}). We distinguish two cases
    \begin{enumerate}
        \item[(a)] $d(i)$ strictly decreases until $d(i_1) = 0$ for some $i_1 \in [\ell]$ then stays constant at 0 until some $i_2 \in [\ell]$ after which $d(i)$ strictly increases (includes $i_1 = i_2$ and $d(i_1) = 0$).
        \item[(b)] $d(i)$ strictly decreases until assuming an absolute minimum $d(i') > 0$ for $i' \in [\ell]$ and then strictly increases (excludes $d(i') = 0$).
    \end{enumerate}
\end{fact}

\begin{figure}[H]
    \centering
    \begin{subfigure}{0.42\textwidth}
        \centering
        \includegraphics[width=.95\textwidth]{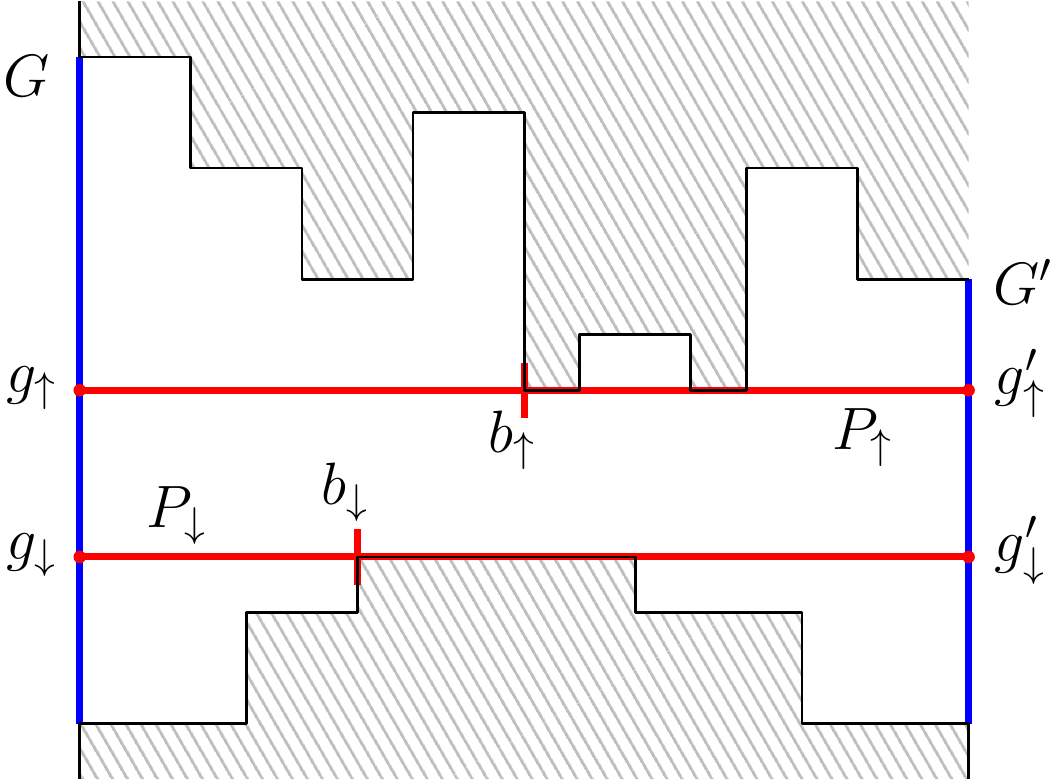}
        \caption{}
        \label{fig:tunnel_splitting_case_a}
    \end{subfigure}
    \hspace{0.1\textwidth}
        \begin{subfigure}{0.42\textwidth}
        \centering
        \includegraphics[width=.95\textwidth]{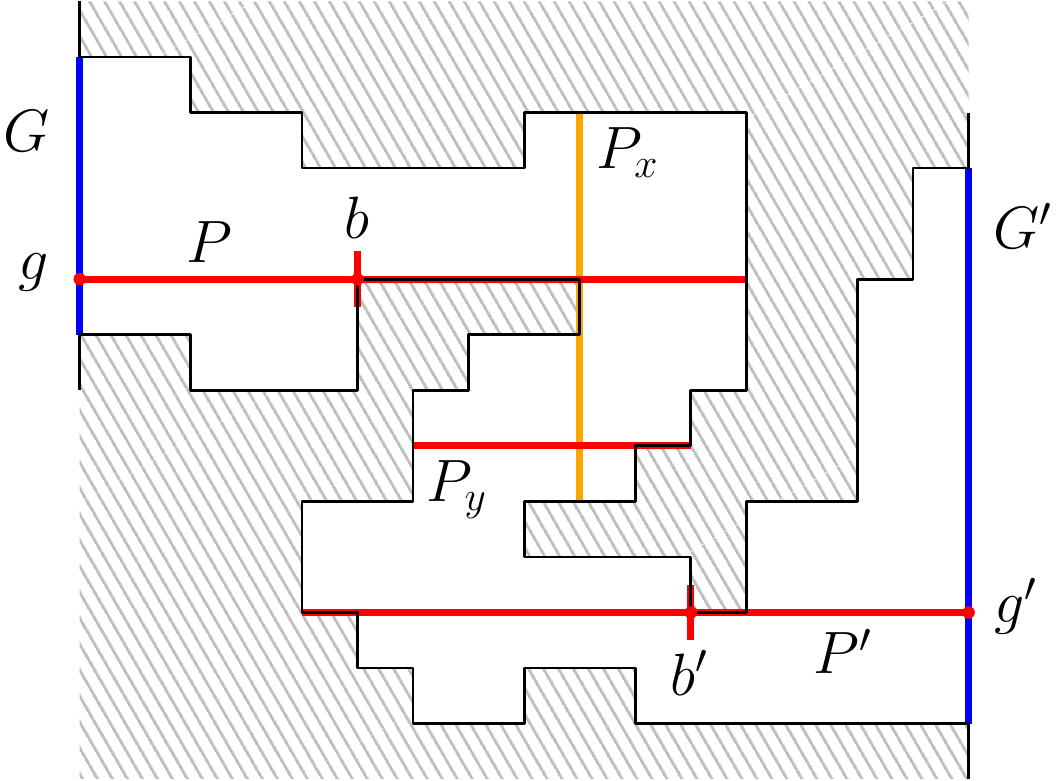}
        \caption{}
        \label{fig:tunnel_splitting_case_b}
    \end{subfigure}
    \caption{Splitting tunnels into \convex regions in two cases according to Def.\ \ref{def:splititng_portals_tunnel}.}
    \label{fig:tunnel_splitting_cases}
\end{figure}

We split the tunnels differently depending on which case of \Cref{fct:tunnel_case_distinction} our tunnel falls into (see \Cref{fig:tunnel_splitting_cases} for a visualization).

\begin{definition}
    \label{def:splititng_portals_tunnel}
    Based on this case distinction in Fact \ref{fct:tunnel_case_distinction} we make the following construction to decompose $T$ into \convex regions (consult Figure \ref{fig:tunnel_splitting_cases} for an illustration).
    \begin{enumerate}
        \item[(a)] Let $g_{\downarrow} := g_{i_1}$ and $g_{\uparrow} := g_{i_2}$. Note that we have two nodes $g_{\downarrow}', g_{\uparrow}'$ on $G'$, that fulfill exactly the same conditions on $G'$ as $g_{\downarrow}, g_{\uparrow}$ on $G$. The pairs are ``facing'' each other horizontally, i.e., we have $d_y(g_{\downarrow},g'_{\downarrow}), d_y(g_{\uparrow},g'_{\uparrow}) = 0$. We split $T$ at the two horizontal portals $P_\downarrow, P_\uparrow$ in $T$ (!) between the pairs $g_\downarrow,g_\downarrow'$ and $g_\uparrow,g_\uparrow'$ (possibly $P_\downarrow = P_\uparrow$). Furthermore, in case $P_\uparrow$ touches a boundary node $b_\uparrow\neq g_\uparrow,g_\uparrow'$ of a hole \emph{above} $P_\uparrow$, then we split the region above $P_\uparrow$ at the leftmost such node in horizontal direction. We do this symmetrically for the region below $P_\downarrow$.
        \item[(b)] Let $g := g_{i'}$ the node on $G$ where $d(i')$ is minimal and let $g'$ be the node with the same property on $G'$. We split $T$ at horizontal portals $P$ and $P'$ in $T$(!) through $g$ and $g'$. Let $b,b'$ be the leftmost node where $P,P'$ touch the boundary with the hole located above $P$ or below $P'$, respectively. We conduct a horizontal split at $b$ and $b'$, respectively. Then we split $T$ at $P_x$, $P_y$ that are defined with respect to the two nodes $g,g'$ as in Definition \ref{def:splititng_portals_tunnel_node_gates}.
    \end{enumerate}
\end{definition}

We now prove that the regions given by the splitting procedure in \Cref{def:splititng_portals_tunnel}

We conclude with the proofs of \Cref{lem:regions_per_tunnel,lem:convex_tunnel,lem:convex_region_complexity}.

\begin{proof}[Proof of \Cref{lem:regions_per_tunnel}]
    Case (b) potentially produces the most regions. $P_x$ intersects at most 3 horizontal portals which gives at most 8 regions. Splitting at $b,b'$ adds at most 2 more regions.
\end{proof}

\begin{proof}[Proof of \Cref{lem:convex_tunnel}]
    \textit{Case (a):} In non-degenerate cases we obtain the 5 regions shown in the example in Figure \ref{fig:tunnel_splitting_case_a}. The middle one forms a rectangle, which is clearly \convex. The others are bordered by at most 2 portals each, thus are \convex by Fact \ref{fct:convexity_easy_case}.
    
    \textit{Case (b):} If we disregard the split at $P_x$, then the regions above and below $P$ and $P'$ are \convex by Fact \ref{fct:convexity_easy_case}. This property is retained for the resulting regions after potentially splitting these two regions again at $P_x$, by Lemma \ref{lem:convexity_retained}. We end up with two regions that are bordered, only by $P$, $G$ and $P'$, $G'$, respectively, which are also \convex by Fact \ref{fct:convexity_easy_case}.
    
    Let $T'$ be the connected area that is enclosed by $P,P'$ (such that $P_y$ is \emph{inside} $T'$) with the nodes on the portals $P,P'$ part of $T'$. The remaining area that is unaccounted for resembles our tunnel with single node gates. In particular, for any node pair $u,v$ inside $T'$ any $uv$-path $\Pi$ that leaves the \textit{tunnel} $T$ has to intersect $P$ \emph{and} $G$ when leaving $T$ and $P'$ \emph{and} $G'$ when reentering. Since $P, G$ and $P',G'$ intersect in $g$ and $g'$, respectively, there is also a $uv$-path $\widetilde \Pi$ that goes through $g$ and $g'$ and is at most as long as $\Pi$ (see Figure \ref{fig:tunnel_splitting_case_b}). 
    
    We can therefore consider $T'$ as tunnel region with single node gates $g,g'$, since these points of entry and exit suffice to obtain shortest paths that travel outside of $T'$. By Lemma \ref{lem:split_region_diameter_general}, after splitting at $P_x$ and $P_y$, the regions that $T'$ decomposes into are \convex.
\end{proof}

\begin{proof}[Proof of \Cref{lem:convex_region_complexity}]
    We first invoke \autoref{lem:closest_points_in_simple_grids} for the two gates $G,G'$, after which every node in $T$ will know their distance to these gates. This takes $\bigO(\log n)$ rounds.
    In particular, this allows the nodes on $G$ to determine their role, in particular, whether they are the nodes $g_\downarrow,g_\uparrow$ or $g$ described in Fact \ref{fct:tunnel_case_distinction}, just by looking at the distance to $G'$ of their neighbors and using the properties of the distance function $d$ defined in Fact \ref{fct:tunnel_case_distinction}. We proceed analogously for the nodes on $G'$.
    
    In case (a) we construct the two horizontal portals $P_\downarrow, P_\uparrow$ in $T$ with endpoints at $g_\downarrow,g_\uparrow$ using our procedures from Sections \ref{appsec:pointer_jumping} and \ref{appsec:identifying_holes}. Using these subroutines, we can also identify the nodes $b_\downarrow,b_\uparrow$ described in case (a). We conduct a split at $b_\uparrow$ on $P_\uparrow$ with respect to the hole above $b_\uparrow$, as described in Section \ref{appsec:splitting_portals}.
    
    Let us consider case (b). Constructing $P,P'$ and $b,b'$ is analogous to case (a). Since $T$ is simple we can use Corollary \ref{cor:broadcast_and_aggregation_trees} and broadcast the distances $d_x = d_{x,T}(g,g')$ and $d_y = d_{y,T}(g,g')$ to all nodes in $T$. Since these also know their horizontal and vertical distances to $G$ and $G'$, nodes can locally decide whether they are part of the portals $P_x$, $P_y$ (see Def.\ \ref{def:splititng_portals_tunnel_node_gates}).
    We have thus constructed all portals and nodes at which we want to split our regions (meaning that all nodes know about their roles in the splitting process). We apply the construction described in Def.\ \ref{def:splitting_procedure}, where the round complexity of $\bigO(\log n)$ is guaranteed by Lemma \ref{lem:splitting_procedure_overhead}.
\end{proof}

\section{Landmark Graph}
\label{sec:landmark_graph}

Our next task is to provide a skeleton of the regionalization which facilitates optimal routing. To do this, we construct what we call a \emph{landmark graph}. The intuition is that we mark certain gate nodes as \emph{landmarks} if they are likely to appear on a shortest path between two regions, and so we can (at least approximately) reduce the problem of finding shortest paths to the problem of finding shortest paths between landmarks.

We connect pairs of landmarks which lie in the same region with virtual edges. These edges have weights corresponding to the distance between the landmarks (calculated using the SSSP subroutine developed in \Cref{sec:closest_points_in_simple_grids}). We call the resulting graph the \emph{landmark graph}. Our goal is to distribute the landmark graph to every node so that routing decisions can be made locally: this is made possible by its relatively small size.

For this section, we assume that we have a grid graph $\Gamma$ which has been regionalized by \Cref{thm:convex_decomposition}. We call this a \emph{regionalized grid graph}; we refer to the set of regions as $\regions$.
This section introduces landmarks and presents various useful properties of them, shows that the landmark graph is of small size, shows that shortest paths in the landmark graph pass through the same regions as shortest paths in the underlying grid graph, and finally demonstrates that the landmark graph can be computed quickly.

First, we define which nodes are marked as landmarks. An example of the placement of landmarks, colored by type, is given in \Cref{fig:landmarks}.

\begin{figure}
\label{fig:landmark_images}
    \centering
    \begin{subfigure}{0.42\textwidth}
        \centering
        \includegraphics[width=0.9\textwidth,page=1]{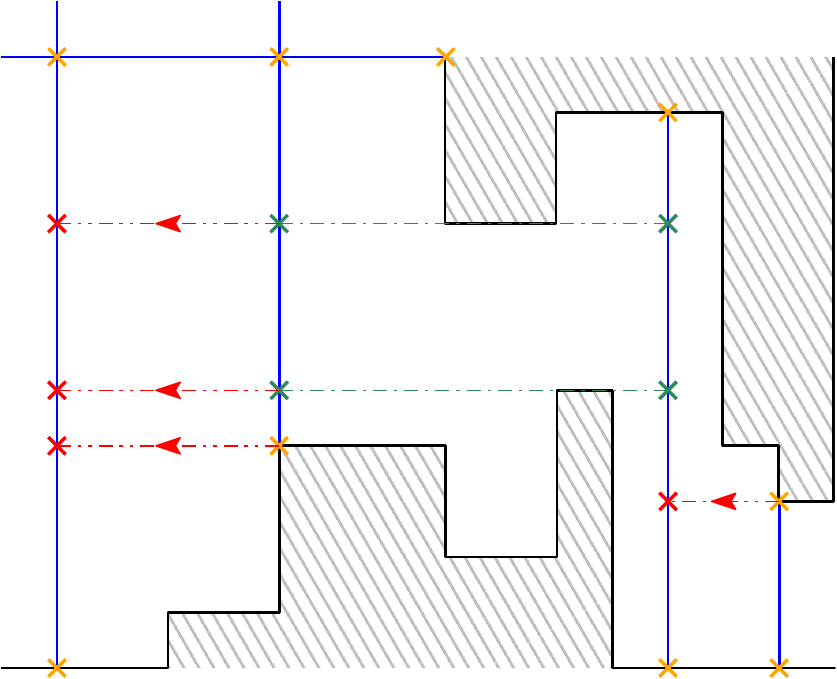}
        \caption{An example of the placement of landmarks. Portal-end landmarks are in \textcolor{orange}{orange};  overhang-induced landmarks are in \textcolor{seagreen}{green}; and projected landmarks are in \textcolor{red}{red}.} 
        \label{fig:landmarks}
    \end{subfigure}
    \hspace{0.1\textwidth}
    \begin{subfigure}{0.42\textwidth}
        \centering
        \includegraphics[width=0.9\textwidth,page=2]{figures/landmarks.pdf}
        \caption{An example of a landmark path from $s$ to $t$. The \textcolor{orange}{Orange} edges are edges in $E_\Lambda$, and the \textcolor{seagreen}{green} edges are composed of multiple edges in $E_\Gamma$.}
        \label{fig:landmark_path}
    \end{subfigure}
    \caption{An overview of our landmark construction.}
\end{figure}

\begin{definition}
\label{def:landmark}
    Given a regionalized grid graph $\Gamma$, consider a vertex $v \in V_\Gamma$ which lies on a \emph{gate}
    $G$.
    Let $P$ be the portal perpendicular to $G$ passing through $v$ (note that if $v$ lies at the intersection of four gates, it is a landmark of the first type), and let $R$ be one of the regions incident to $G$.
    Note that $v$ is one of the endpoints of $P \cap R$ (the portal $P$ restricted to the region $R$).

    Then $v$ is a landmark if one of the following holds:
    \begin{enumerate}[itemsep = -2mm, label=\roman*.]
        \item $v$ is an \textbf{endpoint} of $G$.
        \item $v$ is an \textbf{overhang-induced landmark}.
        Let $u$ be the other endpoint of $P \cap R$. Then $v$ is an overhang-induced landmark if for some $p \in P \cap R$:
        \begin{itemize}[itemsep=-1.5mm]
            \item $p$ lies on a wall $W$; and
            \item $u$ lies either on $W$ or a gate which is not $G$
        \end{itemize}
        We describe $p$ as an ``overhang'' (though note: $p$ is not marked as a landmark).
        \item $v$ is a \textbf{projection landmark}: that is, if any node on $P$ is a landmark of either of the first two types.
    \end{enumerate}
\end{definition}

Next, we define the landmark graph as follows:

\begin{definition} 
\label{def:landmark_graph}
    Suppose we have a regionalized grid graph $\Gamma=(V_\Gamma, E_\Gamma)$ with a set of landmarks $\landmarks \subseteq V_\Gamma$. 
    We define the \emph{landmark graph} as $\Lambda=(V_\Lambda, E_\Lambda)$, where $\{u, v\} \in E_\Lambda$ if $u, v \in V_\Lambda$ and either:
    (i) $u$ and $v$ are on the same portal and there is no landmark between them; or
    (ii) $u$ and $v$ are on different gates incident to the same region, and $v$ is the closest landmark on its gate to $u$.
\end{definition}

Note that we do not connect \emph{all} pairs of landmarks which lie in the same region: to do so would give $|E_\Lambda| = \omega(|\holes|^2)$, and this would adversely affect the running time. We show in \Cref{sec:landmark_graph_properties} (Lemma~\ref{lem:edges_in_landmark_graph_suffice}) that the edges that we select are sufficient to preserve the properties which we need. We present the following observation, which gives us that $|E_\Lambda| = O(|V_\Lambda|)$:

\begin{observation}\label{obs:landmarks_on_gates_and_have_constant_degree}
    \autoref{def:landmark_graph} has two rules adding edges to the landmark graph. The first adds edges between adjacent landmarks on a gate. The second adds edges to the closest point on gates bordering the same region. As each region has a constant number of gates according to \autoref{obs:gates_per_region}, we can conclude that each landmark has constant degree.
\end{observation}

Next, we define the notion of a landmark path (see also \Cref{fig:landmark_path}). Note that the following definition does not require that $s$ and $t$ are landmarks, merely vertices in the grid graph.

\begin{definition}
\label{def:landmark_path}
    Suppose we have a regionalized grid graph $\Gamma$ with a set $\landmarks \subseteq V_\Gamma$ of landmarks and two vertices $s, t$ which lie in different regions in $\regions$.
    We define a \emph{landmark path} from $s$ to $t$ as any path which goes from $s$ to some landmark $\ell_s$ via edges in $E_\Gamma$, then goes from $\ell_s$ to some other landmark $\ell_t$ via edges in $E_\Lambda$ (Definition~\ref{def:landmark_graph}), and then goes from $\ell_t$ to $t$ via edges in $E_\Gamma$.
    We say that a landmark path induces a sequence of regions $(R_0, R_1\dots R_m)$: these are the regions which the underlying grid edges of the landmark path pass through, in order.\footnote{In the special case that two consecutive regions $R_i, R_{i+1}$ meet at a corner, we include one of the regions which is adjacent to both $R_i$ and $R_{i+1}$ in the sequence, in between them.}
\end{definition}

We now give an overview of the rest of the section.
In \Cref{sec:landmark_graph_properties}, we give some miscellaneous properties of the landmark graph.
In \Cref{appsec:routing_properties_landmark_graph}, we give several lemmas relating to the routing-related properties of landmark graphs, and we conclude with a proof of the following theorem%
:

\begin{theorem}\label{thm:landmark_path_correctness}
    Let $\Gamma$ be a regionalized grid graph with a set $\landmarks \subseteq V_\Gamma$ of landmarks, and two vertices $s, t$ in different regions in $\regions$.
    Some shortest path in the grid graph from $s$ to $t$ crosses exactly the same sequence of regions as those induced by the shortest landmark path from $s$ to $t$.
\end{theorem}


Finally, we claim that the landmark graph is also easy to compute. In particular, we have the following lemmas regarding its computation, which we prove in \Cref{appsec:computation_landmark_graph}.

\begin{lemma}
\label{lem:computing_landmarks}
    Given a regionalized grid graph $\Gamma$, nodes can be informed whether they are landmarks (as in Definition~\ref{def:landmark}) in $O(\log n)$ rounds.
\end{lemma}

\newcommand{\landmarkgraphconstructionlemma}{
    Given a regionalized grid graph $\Gamma$ with a set of landmarks $\landmarks$ (computed as in Lemma~\ref{lem:computing_landmarks}),
    the landmark graph $\Lambda=(V_\Lambda, E_\Lambda)$ can be computed, such that every landmark node knows its adjacent landmarks in $\Lambda$ in $O(\log n)$ rounds.
}
\begin{lemma}\label{lem:landmark_graph_construction}
\landmarkgraphconstructionlemma
\end{lemma}

\subsection{Properties of the Landmark Graph}
\label{sec:landmark_graph_properties}

In this section give some miscellaneous properties of landmarks, the landmark graph, and landmark paths which will be useful in the two subsequent subsections.

We begin with a lemma which says that, if we start at a landmark in a region $R$, the closest point in another region (equivalently, the closest point on another gate bounding $R$) is also a landmark. See also \Cref{fig:closest_point_landmark} for some visual intuition for the lemma.

\begin{lemma}
\label{lem:closest_point_is_landmark}
Let $\Gamma$ be a regionalized grid graph with a set of landmarks $\landmarks \subseteq V_\Gamma$. Let $u$ be a landmark incident to a region $R$, and let $G$ be a gate incident to $R$, on which $u$ does not lie.

The closest vertex on $G$ to $u$ is a landmark.
\end{lemma}
\begin{proof}
Suppose wlog that $G$ is horizontal and $R$ is to the north of $G$. Let $(v_1, v_2\dots v_k)$ be the set of nodes of $G$ from left to right. For each $v_i$, let $P_i$ denote the vertical portal on which $v_i$ lies, restricted to $R$, and let $R_V$ denote $\bigcup P_i$, i.e.~the union of the vertical portals through the nodes on $G$. Let $\Pi$ denote a shortest path from $u$ to $G$.

Observe that (i) $\Pi$ must enter $R_V$ (since $G \in R_V$), and (ii) from the first point on $\Pi$ which lies on some portal $P_i \subseteq R_V$, the shortest path to $G$ is along $P_i$.

If $u \in P_i$ for some $P_i \in R_V$, then, since $u$ lies on a gate (as it is a landmark), and there is a portal connecting $u$ and $G$, either $v_i$ is a landmark of the first two types, or $v_i$ must be a projected landmark (case~(iii) of \Cref{def:landmark}).

Otherwise $u \not \in R_V$. Let $q$ be the last node on $\Pi$ which is not in $R_V$, and let $p \in P_i$ (for some $P_i$) be the first node which is. It suffices to argue that $v_i$ is a landmark. Observe that the grid edge from $q$ to $p$ must be horizontal (if it were vertical then clearly $q \in P_i$, which is a contradiction); suppose wlog that $q$ is to the left of $p$.\\
Since $q \in R_V$, there is some grid point which is not a grid node on the line segment between $q$ and $G$. The grid point immediately to the right of this grid point must be a grid node, however, since it lies on $P_i$: call this node $p'$. It is clear that $p'$ is incident to a wall $W$ (since the grid point to its left is unoccupied). We argue that $p'$ is an overhang, and that $v_i$ is therefore an overhang-induced landmark (case~(ii) of \Cref{def:landmark}).\\
Let $v'_i$ be the endpoint of $P_i$ which is not $v_i$. If $v'_i$ lies on a gate, or $v'_i$ is a wall of node of the wall $W' \neq w$ then $v_i$ is an overhang-induced landmark and we are done. Suppose $v'_i$ lies on $W$. Let $R'$ be the region bounded by $W$ and the segment of $P_i$ between $p'$ and $v'_i$. Since $R'$ contains no landmarks, $\Pi$ must have entered this region to reach $q$ and then returned to the boundary of the region to reach $p$. So $\Pi$ could be shortened by not entering $R'$, which contradicts that $\Pi$ is a shortest path.
\end{proof}

\begin{figure}
    \centering
    \includegraphics[width=0.37\textwidth]{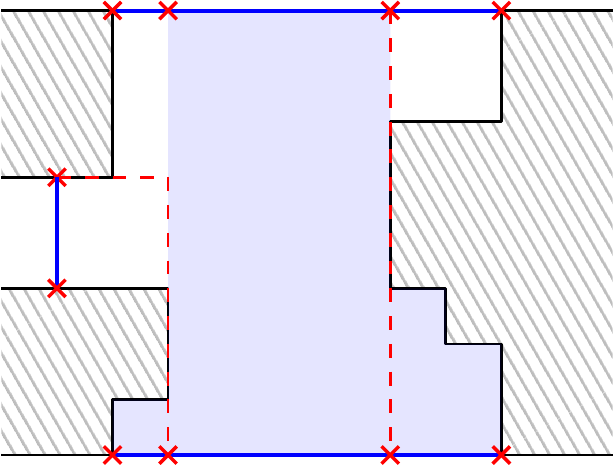}
    \caption{A visual explanation of \Cref{lem:closest_point_is_landmark}. \textcolor{blue}{Blue} lines are gates, \textcolor{red}{red} crosses are landmarks, and two example paths between landmarks are depicted with \textcolor{red}{red} dotted lines. The bottom gate corresponds to $G$, and $R_V$ is denoted by a blue shaded area. Notice that if landmarks on other gates lie inside $R_V$ then there must be a corresponding landmark on $G$; if landmarks on other gates lie outside $R_V$ then the must enter from the left or right, and then the shortest path is direct to $G$.}
    \label{fig:closest_point_landmark}
\end{figure}

A useful corollary follows from the case distinction of this lemma:

\begin{corollary}
\label{cor:closest_point_is_not_landmark}
    Let $\Gamma$ be a regionalized grid graph, let $u$ be a point in a region $R$, and let $G$ be a gate bounding $R$, on which $u$ does not lie.
    
    If the closest vertex on $G$ to $u$ is not a landmark, the shortest path from $u$ to $G$ is a straight line.
\end{corollary}

Next, we argue that in a simple region $R$ and given a point $u$ in $R$, it must be the case that there's a unique closest point to $u$ on any gate bounding $R$. We start by giving a slightly more general lemma which then immediately implies this.

\begin{lemma}
\label{lem:distances_are_convex}
Let $\Gamma$ be a regionalized grid graph. Let $u$ and $v$ be vertices in the same region $R$, let $v$ lie on a gate bounding $R$ and let $v_1$ and $v_2$ be the neighbors of $v$ on its portal. Then $d_\Gamma(u, v_1) > d_\Gamma(u, v)$, or $d_\Gamma(u, v_2) > d_\Gamma(u, v)$.
\end{lemma}
\begin{proof}
    Let $\Pi_1^*$ be a shortest path in $R$ from $u$ to $v_1$, and let $\Pi_2^*$ be a shortest path in $R$ from $u$ to $v_2$. Let $p$ be the last point that $\Pi_1^*$ and $\Pi_2^*$ have in common, when traversing both paths from $u$ to their respective destinations, and let the subpath of $\Pi_1^*$ starting at $p$ and ending at $v_1$ be $\Pi_1$ (resp. $\Pi_2$).
    Note that if $v$ is a point on either of $\Pi_1$ or $\Pi_2$ we are done (as we just terminate the path in question at $v$---it must be shorter than the full path), and so assume it is not: this implies that the area enclosed by $\Pi_1$, $(v_1, v)$, $(v, v_2)$, and $\Pi_2$ is a simple cycle in the grid graph (as otherwise, $\Pi_1$ or $\Pi_2$ could be shortened).
    
    Consider a line segment $\ell$ which starts at $v$ and extends into $R$, perpendicular to $P_v$, until it intersects either $\Pi_1$ or $\Pi_2$. Note that all points on $\ell$ are within the simple cycle formed by the paths and vertices $v_1, v, v_2$, and therefore all vertices along $\ell$ exist (i.e.~none lie within holes).
    Suppose without loss of generality that $\ell$ intersects $\Pi_1$ at point $v^*$, and the length of $\ell$ from $v$ to $v^*$ is $x$. Now consider the path formed by concatenating the path $\Pi^*_1$ from $u$ up to the point $v^*$, with the line segment $\ell$ (from $v^*$ to $v$). Since $d_\Gamma(v^*, v_1) \geq x + 1$ but $\ell$ is of length $x$, this path is shorter than $\Pi^*_1$ and goes from $u$ to $v$, proving the lemma.
\end{proof}

\begin{corollary}
\label{cor:unique_closest_point}
    Let $\Gamma$ be a regionalized grid graph. Let $u$ be a vertex in a region $R$, and let $G$ be a gate incident to $R$.
    
    There is a \emph{unique} closest point to $u$ on $P$. Let this point be $v$. Given another point $w$ on the same portal as $v$, $d_\Gamma(u, w) = d_\Gamma(u, v) + d_\Gamma(v, w)$.
\end{corollary}

Next, we show that the edges that we have selected to be part of our landmark graph are sufficient to represent shortest paths between all pairs of landmarks which border a region.

\begin{lemma}
\label{lem:edges_in_landmark_graph_suffice}
Let $u$ and $v$ be landmarks in the same region $R$. Some shortest path in $\Gamma$ from $u$ to $v$ can be expressed as the concatenation of the underlying paths of edges in $E_\Lambda$.
\end{lemma}
\begin{proof}
This follows from Lemmas~\ref{lem:closest_point_is_landmark}~and~\ref{lem:distances_are_convex}. First, recall that  since $R$ is path-convex there exists some shortest path between $u$ and $v$ in $\Gamma$ which lies entirely in $R$.

Let $P_u$ and $P_v$ be the portals on which $u$ and $v$ lie. If $P_u = P_v$, then the lemma statement is trivially true, since there are edges in $E_\Lambda$ between adjacent landmarks on the same portal, and the shortest path between $u$ and $v$ is along the portal on which they both lie.

If $P_u \neq P_v$, then let $v^*$ be the closest point on $P_v$ to $u$, which, by Lemma~\ref{lem:closest_point_is_landmark}, is a landmark (and hence $(u, v^*)$ is an edge in $E_\Lambda$). We claim that some shortest path between $u$ and $v$ first goes from $u$ to $v^*$, and then goes from $v^*$ to $v$ along $P_v$; equivalently, $d_\Gamma(u, v^*) + d_\Gamma(v^*, v) = d_\Gamma(u, v)$. This follows easily from \Cref{cor:unique_closest_point}, and this path can be constructed using edges in $E_\Lambda$.
\end{proof}

Finally for this section, we compute the size of the landmark graph. This will be particularly useful when we consider the computational complexity of distributed our landmark graph and setting up our routing scheme in Section~\ref{sec:routing}.

\begin{lemma}
\label{lem:landmark_properties}
    Given a regionalized grid graph $\Gamma$ with $|\holes|$ holes and a set of landmarks $\landmarks$ as in Definition~\ref{def:landmark}, the following properties all hold:
    \begin{itemize}
        \item Each landmark is adjacent to $O(1)$ regions.
        \item Each region has $O(k)$ adjacent landmarks.
        \item The landmark graph $\Lambda$ has $|V_\Lambda| = O(k^2)$ nodes
        \item The landmark graph has $|E_\Lambda| = O(k^2)$ edges.
    \end{itemize}
\end{lemma}
\begin{proof}
    We begin by observing that since landmarks are always placed on portals (cf.~\Cref{def:landmark}), and each portal bounds at most a constant number of regions (follows from the splitting procedure), the first property holds.

    For the second and third properties, we count the three different origins of landmarks separately:
    \begin{itemize}
        \item \textbf{Endpoints of a gate:} From Theorem~\ref{thm:convex_decomposition} we know that $|\regions| = O(|\holes|)$. Consider the graph where vertices are regions and there is an edge between regions for each portal that connects them. Clearly this graph is planar and therefore the number of portals bounding regions is $O(|\holes|)$ as well. Trivially, since each region is bounded by a constant number of portals, each region has a constant number of adjacent landmarks of this type.
        \item \textbf{Overhang-induced regions:} By the definition of overhangs there can be at most two overhang-induced landmarks for each pair of horizontal (resp. vertical) gates bounding a region. Since the number of gates bounding a region is constant, there are a constant number of overhang-induced portals per region, and the total over the graph is $O(|\holes|)$.
        \item \textbf{Projected landmarks:} Clearly we must have $O(|\holes|^2)$ of these, since there are $O(|\holes|)$ overhangs in total and each could be ``projected'' onto any gate bounding any other region. There are $O(|\holes|)$ regions and each is bounded by $O(1)$ gates, giving the claimed bound.
    \end{itemize}
    
    The number of landmark edges in each region is $O(|\holes|)$, as by \Cref{obs:landmarks_on_gates_and_have_constant_degree} each landmark has at most a constant number of neighbours in each region. Since there are $O(|\holes|)$ regions, the bound in the fourth property follows immediately.
\end{proof}

\subsection{Routing Properties of the Landmark Graph}
\label{appsec:routing_properties_landmark_graph}

Next, we show that our choice of where to place landmarks give us useful properties for routing. We conclude the section with the proof of \Cref{thm:landmark_path_correctness}.

First, we show that, if we fix a sequence of regions through which a path must travel, a shortest path can be obtained by repeatedly routing to the closest point in the next region in the sequence.

\begin{lemma}
\label{lem:closest_point_routing}
    We say that an $s$-$t$ path is \emph{shortest-possible} relative to a sequence of pairwise-adjacent regions $(R_0, R_1,\dots R_m)$, where $s \in R_0$ and $t \in R_m$, if no shorter path passes through exactly those regions, in order.

    A shortest-possible $s$-$t$ path can be obtained by routing from $s$ to the closest point in $R_1$, then to the closest point in $R_2$, and so on: and when finally $R_m$ is entered, taking a shortest path to $t$.
\end{lemma}
\begin{proof}
    Let $\Pi^*$ be a shortest-possible path from $s$ to $t$, and let $\Pi$ be a path which follows the strategy outlined in the lemma. We will show that $|\Pi| = |\Pi^*|$. We will notate the length of path $\Pi$ between points $x$ and $y$ as $|\Pi(x, y)|$ (and analogously for the length of $\Pi^*$).
    
    Suppose for each $R_i$, $\Pi^*$ enters $R_i$ at some point $v^*_i$, and $\Pi$ enters $R_i$ at $v_i$. Note that since pairs of adjacent regions are connected by a single portal, $v_i$ and $v^*_i$ are on the same portal. It suffices to show that $|\Pi^*(s, v^*_m)| = |\Pi(s, v_m)| + d_\Gamma(v_m, v^*_m)$, since when routing to $t$ in the final step, clearly $d_\Gamma(v_m, t) \leq d_\Gamma(v_m, v^*_m) + d_\Gamma(v^*_m, t)$ (as we can route via $v^*_m$ if this is shortest). 
    
    We prove this by induction. Note that as a consequence of \Cref{cor:unique_closest_point}, we have that $|\Pi^*(s, v^*_1)| = |\Pi(s, v_1)| + d_\Gamma(v_1, v^*_1)$, hence the base case is satisfied. Now by two applications of \Cref{cor:unique_closest_point} we have, for all $v_i, v_{i+1}, v^*_i, v^*_{i+1}$, a kind of rectangle equality:
    \[d_\Gamma(v_i, v_{i+1}) + d_\Gamma(v_{i+1}, v^*_{i+1}) = d_\Gamma(v_i, v^*_i) + d_\Gamma(v^*_i, v^*_{i+1})\]
    And this gives us:
    \begin{align*}
        |\Pi^*(s, v^*_{i+1})| &= |\Pi^*(s,v^*_i)| + d_\Gamma(v^*_i, v^*_{i+1})\\
        &= |\Pi(s, v_i)| + d_\Gamma(v_i, v^*_i) + d_\Gamma(v^*_i, v^*_{i+1})\\
        &= |\Pi(s, v_i)| + d_\Gamma(v_i, v_{i+1}) + d_\Gamma(v_{i+1}, v^*_{i+1})\\
        &= |\Pi(s, v_{i+1})| + d_\Gamma(v_{i+1}, v^*_{i+1})
    \end{align*}
    as required.
\end{proof}

Next, we show that if \emph{any} shortest path between two nodes contains a landmark, then there is a shortest path between those two nodes which is a landmark path.
\begin{lemma}
\label{lem:landmark_on_shortest_path}
    If there is a landmark on some shortest path in the grid graph from $s$ to $t$, then some shortest path in the grid graph from $s$ to $t$ is a landmark path.
\end{lemma}
\begin{proof}
    Firstly, note that if $v$ is a point on a shortest $st$-path, then this shortest $st$-path is composed of a shortest $sv$-path and a shortest $vt$-path.
    
    Then, consider a landmark $\lambda$ on a shortest $st$-path $\Pi$, and let $\Pi_{\lambda t}$ be the subpath of $\Pi$ from $\lambda$ to $t$. We can replace this subpath by a new path $\Pi^*_{\lambda t}$ obtained by the strategy given in Lemma~\ref{lem:closest_point_routing}, using the regions passed through by $\Pi_{\lambda t}$ to ensure the cost does not increase relative to $\Pi_{\lambda t}$. Note that our path starts at a landmark: by Lemma~\ref{lem:closest_point_is_landmark}, each time we route to the closest point in the next region in the sequence, this point must also be a landmark. As Lemma~\ref{lem:edges_in_landmark_graph_suffice} shows that we can replace shortest paths between landmarks by edges in $E_\Lambda$, we can see that our path $\Pi^*_{\lambda t}$ is a landmark path.
    
    We can apply the same construction to replace $\Pi_{s \lambda}$ with a landmark path $\Pi^*_{s \lambda}$, and finally it is easy to see that concatenating $\Pi^*_{s \lambda}$ with $\Pi^*_{\lambda t}$ gives a landmark path of no greater length than $\Pi$: that is to say, this landmark path is a shortest $st$-path in the grid graph.
\end{proof}

\begin{figure}
    \centering
    \includegraphics[width=0.5\textwidth]{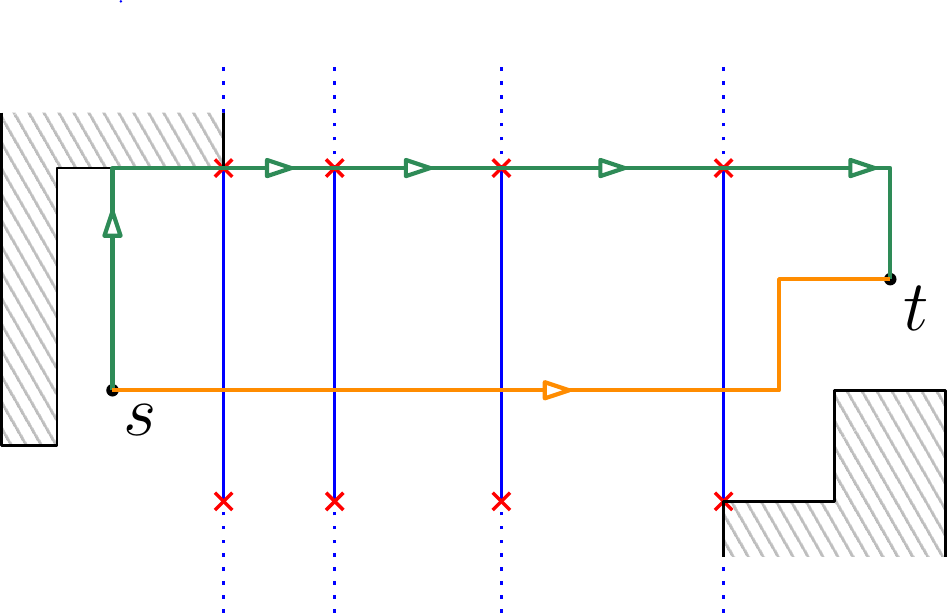}
    \caption{An example of the case where there are no landmarks on a shortest path from $s$ to $t$.}
    \label{fig:special_routing_case}
\end{figure}

We conclude the subsection with the proof of \Cref{thm:landmark_path_correctness}, that the shortest $st$-landmark path crosses the same sequence of regions as a shortest $st$-path:

\begin{proof}[Proof of \Cref{thm:landmark_path_correctness}]
    First, note that if any shortest $s$-$t$ path contains a landmark, then by Lemma~\ref{lem:landmark_on_shortest_path} we are immediately done.
    
    Suppose no shortest $s$-$t$ path contains a landmark, and fix an arbitrary shortest path $\Pi'$. Let $\Pi$ be a shortest-possible path relative to the same sequence of regions $(R_1, R_2,\dots R_m)$ induced by $\Pi'$. By \Cref{lem:closest_point_routing}, we can construct $\Pi$ such that it is comprised of a shortest path from $s$ to the region $R_1$, followed by a shortest path to the region $R_2$, etc., when the path enters $R_m$, it is then followed by a shortest path to $t$.
    
    We now argue that, since $\Pi$ does not pass through any landmarks, it must have a very specific structure. For each $R_i$, let $v_i$ be the first node on $\Pi$ in $R_i$. Since $\Pi$ does not contain a landmark by assumption, by Corollary~\ref{cor:closest_point_is_not_landmark} the shortest path from $s$ to $v_1$ must be a straight line. This logic also applies to the shortest path from $v_1$ to $v_2$, and from $v_2$ to $v_3$, all the way to $v_m$, and note further that all of the portals on which the $v_i$ lie must be parallel to each other (if the path were to ``turn'' between some $v_i$ and $v_{i+1}$, it would necessarily reach a landmark by contrapositive of \Cref{cor:closest_point_is_not_landmark}). Finally, let the closest point to $t$ on the gate separating $R_{m-1}$ from $R_m$ be $t^*$. A shortest path from $v_m$ to $t$ may go via $t^*$ first (by Lemma~\ref{cor:unique_closest_point}). To summarise: $\Pi$ consists of a straight line from $s$ to $v_m$, and then a path from $v_m$ to $t$ (via $t^*$).
    
    It remains to show that the shortest landmark path induces the same sequence of regions. Suppose without loss of generality that the gates which $\Pi$ crosses are all vertical, and let $\lambda^\uparrow_i$ and $\lambda^\downarrow_i$ be the closest landmarks above and below $v_i$ on its gate, respectively. We first argue that all of the $\lambda^\uparrow_i$ lie in a horizontal line (and that all points along this line are vertices in the grid graph). We argue that the shortest path between $\lambda^\uparrow_1$ and $\lambda^\uparrow_2$ is a straight line: the argument applies analogously to all other $\lambda^\uparrow_i$ and all $\lambda^\downarrow_i$.
    Let $v^i_1$ be the vertex $i$ steps north of $v_1$. Iteratively consider the sequence of pairs of points $((v^1_1, v^1_2), (v^2_1, v^2_2)\dots)$. If any of the points on the line between $v^i_1$ and $v^i_2$ are incident to a hole boundary, then $v^i_1=\lambda^\uparrow_1$ and $v^i_2=\lambda^\uparrow_2$ (by case (ii) of Definition~\ref{def:landmark}: the points are landmarks). Otherwise, suppose that (without loss of generality), $v^i_1 = \lambda^\uparrow_1$. Then, by case (iii) of Definition~\ref{def:landmark}, $v^i_2=\lambda^\uparrow_2$ (by projection). If neither of these hold, then proceed to $(v^{i+1}_1, v^{i+1}_2)$. Eventually the end of one of the portals will be reached and both landmarks will be found at the same ``height''.
    
    We claim that $t$ is neither more north than $\lambda^\uparrow_m$ nor more south than $\lambda^\downarrow_m$. To see why, recall that $t^*$ is the closest point to $t$ on the gate containing $v_m$. If there was a landmark between $t^*$ and $v_m$ on their gate, it would lie on $\Pi$, which is a contradiction. Therefore the line segment $(t^*, v_m)$ does not cross $\lambda^\uparrow_m$ or $\lambda^\downarrow_m$. Combining this with an application of \Cref{cor:closest_point_is_not_landmark} shows our claim to be true.
    
    Finally, we show that the shortest landmark path induces the same sequence of regions as the shortest $st$-path. Consider the two landmark paths $L^\uparrow$, $L^\downarrow$, with $L^\uparrow$ going from $s$ to $\lambda^\uparrow_1$ to $\lambda^\uparrow_2\dots$ to $\lambda^\uparrow_m$ to $t$, and $L^\downarrow$ defined analogously. Note that both of these paths induce the series of regions $(R_1, R_2,\dots R_{m-1}, R_m)$. Both of these paths will travel the same distance horizontally (the minimum possible), but one may travel less than the other vertically. We observe that, since there are no landmarks inside the box bounded by $\lambda^\uparrow_1$, $\lambda^\uparrow_m$, $\lambda^\downarrow_1$, and $\lambda^\downarrow_m$, any other landmark path must travel at least the same distance horizontally but a greater distance vertically.
\end{proof}

\subsection{Computation of the Landmark Graph}
\label{appsec:computation_landmark_graph}

We conclude the section on the landmark graph by showing that the landmark graph can be computed efficiently in the \HYBRID model, proving the two lemmas to that effect which we gave at the beginning of the section.

\begin{proof}[Proof of \Cref{lem:computing_landmarks}]
    Consider the three cases where a node could be a landmark in Definition~\ref{def:landmark}:
    \begin{itemize}
        \item \textbf{Endpoints of a gate:} Nodes already know whether they are the endpoints of a gate: they become aware of this during the region decomposition process.
        \item \textbf{Overhang-induced landmarks:} Nodes on gates know the ID of the hole boundary or gate which the portal perpendicular to the gate terminates at. Nodes know whether they are incident to a hole boundary and can broadcast this information to all nodes on their incident portals in $O(\log n)$ rounds. After this, nodes on gates can locally determine whether they are overhang-induced landmarks as they have all the necessary information.
        \item \textbf{Projected landmarks:} Nodes can broadcast whether they are a landmark of the first two types to the rest of the portals on which they lie in $O(\log n)$ rounds, per Lemma~\ref{lem:broadcast_and_aggregation} (note that it suffices to know that \emph{some} node on the same portal is a landmark). Nodes already know, from the region decomposition, whether they lie on a gate.
    \end{itemize}
\end{proof}

\begin{proof}[Proof of \Cref{lem:landmark_graph_construction}]
    We compute $E_\Lambda$ with a constant number of applications of \Cref{lem:closest_points_in_simple_grids}. This approach gives a running time of $O(\log n)$ rounds), so it remains to argue that this is possible. Note that we don't compute the underlying paths of edges in $E_\Lambda$ (we address this in \Cref{sec:routing}), it suffices to know \emph{which} landmarks the edges are between.
    
    First, note that nodes can easily identify edges in $E_\Lambda$ which go between landmarks on the same portal in $O(\log n)$ rounds, using simple \HYBRID primitives to discover the closest landmarks to them on their own gate. It remains to show that edges in $E_\Lambda$ which go between landmarks on different gates can be computed in $O(\log n)$ rounds.
    
    Fix an arbitrary region $R$, fix a gate $G$ incident to $R$, fix a landmark $\lambda$ which is not on $G$, and let $\lambda_G$ be the closest point to $\lambda$ on $G$, which by \Cref{lem:closest_point_is_landmark} is also a landmark. We argue that it suffices to run SSSP starting at each landmark of the first two types of \Cref{def:landmark} on $G$, and that we do not need to run it for landmarks of the third type on $G$.
    
    To see why, suppose that $\lambda_G$ is a landmark of the third type (a projected landmark). Therefore, by reasoning in \Cref{lem:closest_point_is_landmark}, there is a straight line path from $\lambda$ to $\lambda_G$. Since the node at $\lambda_G$ knows that it is a projected landmark, it can broadcast this information to all nodes on the portal on which it lies that is perpendicular to $G$. Since $\lambda$ lies on this portal, and both nodes know that they are landmarks on gates incident to $R$, both nodes will know that there is an edge between them in $E_\Lambda$.
    
    Finally, we observe that there are a constant number of landmarks of the first two types on each gate. This means that, for each gate, we need only run a constant number of SSSP computations to compute all edges incident to the gate. Since there are at most a constant number of gates per region, this gives a total running time of $O(\log n)$.
\end{proof}
    
\section{Routing}\label{sec:routing}
In this section we combine our previous results and use them to formulate our routing scheme.
Recall that, according to \autoref{def:routing_schemes}, for a routing scheme each node $v$ needs to learn its label $\lambda(v)$ and routing function $\rho(v)$. After computing the region decomposition and the landmark graph as presented in the previous sections, the nodes use these structures to learn the information required for their own label and routing table. Specifically, the routing tables will be identical for each node and will consist of a version of the landmark graph that is labeled with additional information linking it to the actual grid graph. This means that each node has to learn the landmark graph and these additional labels. To use the landmark graph to make routing decisions, the nodes need to add themselves and the target node to the landmark graph. This requires them to know the distances to close landmarks for both themselves and the target node. Hence we make this information part of the node labels. To distinguish different cases of the algorithm, the node's label additionally contains its region identifier and its node identifier.

In the following we start by proving that the nodes do not need to learn \emph{all} landmarks adjacent to their region in order to themselves and their nearby landmarks to their copy of the landmark graph. Instead, it suffices to learn two landmarks (see \autoref{def:important_landmarks}) for each gate adjacent to the region.
As each gate can have up to $O(|\holes|)$ many landmarks, only have to learn two of them significantly reduces the size of the node labels required.
Afterwards, we describe how the nodes compute and distribute the data required for the node labels and routing tables. Finally, we describe how a node can use the data it learned this way to forward a packet towards a target node.


To prepare for this, we first fix the following observations that are commonly used in the proofs in this section:
\begin{observation}\label{obs:regions_per_node}
    Each node is part of $\bigO(1)$ regions. This follows from the constructions of \autoref{def:construction-Gamma}, \autoref{def:junction_portal_splitting} and \autoref{def:splititng_portals_tunnel}.
\end{observation}
\begin{observation}\label{obs:gates_per_region}
    Each region has $\bigO(1)$ gates. This follows from the facts that tunnels have exactly $2$ gates; that according to \autoref{lem:regions_per_tunnel} we split tunnels into at most $10$ regions; and that each of the splits is entirely horizontal or vertical.
\end{observation}
\begin{observation}\label{obs:landmark_degree}
    The degree of each landmark is $\bigO(1)$.
    This is because \autoref{def:landmark_graph} has two rules introducing edges to the landmark graph. The first connects landmarks closest to each other on the same gate and the second connects landmarks closest to each other on different gates. As there are $\bigO(1)$ gates per region according to \autoref{obs:gates_per_region}, both rules add only a constant number of edges to the landmark graph for each landmark.
\end{observation}
We first use the regionalization algorithm described in \autoref{sec:partitioning_the_graph} to divide the graph into simple, path-convex regions.
Then, we build a landmark graph as described in \autoref{sec:landmark_graph}. Together, these constructions take $\bigO(|\mathcal{H}|+\log n)$ time (cf.\ \autoref{thm:convex_decomposition} and \autoref{lem:landmark_graph_construction}).

Afterwards, we describe which landmarks adjacent to their region the nodes need to include in their label, and how they can learn them. To enable nodes to make routing decisions, the label of the target node must contain the landmarks adjacent to the target node's region. As there could be $\bigO(|\mathcal{H}|)$ landmarks that fit this description, we need to find a small subset of these landmarks that suffice without introducing any stretch. We formalize this notion of important landmarks with the following definition and show that they suffice for routing. 
\begin{definition}\label{def:important_landmarks}
    Given a node $v$, for each gate adjacent to $v$'s region, consider the unique (\autoref{cor:unique_closest_point}) closest node to $v$. We call the closest landmark in each direction induced by the gate \emph{important}.
\end{definition}

\begin{figure}[H]
    \centering
    \resizebox{!}{0.33\textwidth}{%
        \begin{tikzpicture}
    \fill [pattern=north east lines, pattern color=gray] (0,.5) -- (0,-4) -- (4,-4) -- (4,-3) -- (5,-3) -- (5,-4.5) -- (-.5,-4.5) -- (-.5,.5);
    \fill [pattern=north east lines, pattern color=gray] (6,-4.5) -- (6,0) -- (2,0) -- (2,-1) -- (1,-1) -- (1,0.5) -- (6.5,0.5) -- (6.5,-4.5);
    \draw[thick] (0,.5) -- (0,-4) -- (4,-4) -- (4,-3) -- (5,-3) -- (5,-4.5);
    \draw[thick] (6,-4.5) -- (6,0) -- (2,0) -- (2,-1) -- (1,-1) -- (1,.5);

    \node[draw=tenn,fill=tenn,circle,inner sep=2pt] at (1,-2) {};
    \draw[thick,tenn] (1,-2) -- (3,-2);
    \draw[very thick,tenn] (2.975,-2) -- (2.975,-1);
    \draw[very thick,tenn] (2.975,-2) -- (2.975,-3);
    \node[draw=tenn,fill=tenn,circle,inner sep=2.8pt] at (3,-1) {};
    \node[draw=tenn,fill=tenn,circle,inner sep=2.8pt] at (3,-3) {};

    \draw[deepcerulean,thick] (3,0) -- (3,-4);
    \node[draw=deepcerulean,fill=deepcerulean,circle,inner sep=2pt] at (3,0) {};
    \node[draw=deepcerulean,fill=deepcerulean,circle,inner sep=2pt] at (3,-1) {};
    \node[draw=deepcerulean,fill=deepcerulean,circle,inner sep=2pt] at (3,-3) {};
    \node[draw=deepcerulean,fill=deepcerulean,circle,inner sep=2pt] at (3,-4) {};
\end{tikzpicture}
    }
    \caption{A region boundary with landmarks (blue) with a node and the corresponding important landmarks (red). To obtain these, consider the closest point on the region boundary to the node and mark the closest landmark in each direction.}
    \label{fig:important_landmarks}
\end{figure}

\begin{lemma}
\label{lem:important_landmarks_suffice}
     For each landmark path starting from a node $v$ and ending with a node $t$, such that $v$ and $t$ are in different regions, there is a landmark path of equal length that has important landmarks as its first and last landmarks, i.e., the landmark path starts with $v$ and an important landmark of $v$ and ends with an important landmark of $t$ followed by $t$. For each node, there are $\bigO(1)$ important landmarks in total.
\end{lemma}
\begin{proof}
    Let $v'$ be the closest node to $v$ on a given gate. \autoref{cor:unique_closest_point} yields that for any landmark $\lambda$ on that gate, $d(v,\lambda)=d(v,v')+d(v',\lambda)$. If $\lambda$ is non-important, there must be an important landmark $\lambda'$ closer to $v'$. As $\lambda$ and $\lambda'$ are on the same gate, we can conclude $d(v,\lambda)=d(v,v')+d(v',\lambda')+d(\lambda',\lambda)$. Hence, given a landmark path starting with $v$ and $\lambda$, we obtain a landmark path of equal length starting with $v$ and $\lambda'$. The argument for the end of the landmark path is analogous. The fact that there are $\bigO(1)$ important landmarks follows directly from \autoref{obs:gates_per_region}, since there are at most $2$ important landmarks per gate.
\end{proof}
Hence, it suffices to include these important landmarks in the node labels.
To allow the nodes to make routing decisions using only locally stored information, we need to further augment the landmark graph with some labels. To this end, we start by computing a region identifier for each region and let each node learn its region's identifier.
\begin{lemma}\label{lem:region_identifiers}
    Each node can learn the identifier of its region, i.e. the minimum identifier of any (virtual) node in the region in time $\bigO(\log n)$. Each region identifier is unique and of size $\bigO(\log n)$.
\end{lemma}
\begin{proof}
    Our goal is to pick the minimum identifier of any node in a region as its region identifier. Per \autoref{lem:landmark_properties}, each node may be part of multiple regions. Hence, each occurrence of a node must have a unique identifier. To achieve this, any node that is part of multiple regions picks a unique suffix for each of them, and then appends this suffix to the identifier of the node. Since each node is part of a constant number of regions, the resulting identifiers are still of length $\bigO(\log n)$.

    Now, we can aggregate the resulting identifiers to find the minimum.
    To aggregate among a region, we perform two aggregations and a broadcast. In the first aggregation, each node learns the minimum identifier of its horizontal portal. In the next step, each node of the region boundary aggregates the results of the first step. This way, each boundary node learns the minimum identifier. In the third step, the rightmost node of each horizontal portal (that must be part of the region boundary) broadcasts that minimum to each node of that portal.

    Since all these operations can be performed in $\bigO(\log n)$ using pointer-jumping (\autoref{lem:broadcast_and_aggregation}), we have a total runtime of $\bigO(\log n)$.
\end{proof}

Next, we have the nodes learn the distance to their adjacent gates and we use the SSSP tree resulting from this execution, to further obtain their distances to their important landmarks.

\begin{lemma}\label{lem:gate_sssp}
    Each node can learn the distance to each gate adjacent to its region and its neighbor that is closest to that region in time $O(\log n)$. 

    After learning about their adjacent gate, each node can infer the distance to each of its important landmarks in time $\bigO(\log n)$. $\bigO(\log n)$ bits suffice to store all distances learned this way.
\end{lemma}
\begin{proof}
    As each region has $\bigO(1)$ gates according to \autoref{obs:gates_per_region}, we can perform the algorithm from \autoref{lem:closest_points_in_simple_grids} a constant number of times to achieve the first claim.

    For the second claim, using that we only have a constant number of gates per region again, we can describe the algorithm for a single gate.

    After the execution of the algorithm of \autoref{lem:closest_points_in_simple_grids} each node learns its distance to the gate and its predecessor (\autoref{obs:shortest_path_tree}). Specifically, we obtain a shortest path tree structure connecting each node in an adjacent region to the gate.
    Next, we have all non-landmark gate nodes learn the distances to the closest landmark in each direction. This can be achieved in $\bigO(\log n)$ time by performing pointer-jumping on the gates. Finally, each gate node performs the Euler tour technique of \autoref{lem:euler_tour} on its subtree of the gate's shortest path tree. This allows it to distribute its distance values to the closest landmark in each direction in time $\bigO(\log n)$. By adding the distances to the gate to the distances of gate nodes to their closest landmarks, the nodes can infer the distance to their important landmarks.
\end{proof}

To enable the nodes to make routing decisions, they need to learn which gate leads to which region and we require further labeling of the gates and the edges and landmarks of the landmark graph. The next three lemmas describe how to establish the required information.
\begin{lemma}\label{lem:gate_labels}
    After establishing region identifiers and performing a closest point computation, gate nodes can distribute their knowledge on what regions they are adjacent to to all nodes closest to them in $\bigO(\log n)$ rounds.
\end{lemma}
\begin{proof}
    Due to the establishment of region identifiers in \autoref{lem:region_identifiers}, each gate node knows which regions it is adjacent to. Additionally, each gate node obtains an SSSP subtree during the closest point computation. They can execute the euler tour technique (\autoref{lem:euler_tour}) on these subtrees to transform them into a path graph and employ pointer-jumping (\autoref{lem:pointer_jumping_structure}) to distribute the identifiers of the regions they are adjacent to.
\end{proof}
\begin{lemma}\label{lem:gate_identifiers}
    Each node on a gate can learn the label of its gate (i.e. the minimum identifier of any node on it and a bit for tiebreaking) in time $\bigO(\log n)$. Each gate identifier is unique and of size $\bigO(\log n)$.
\end{lemma}
\begin{proof}
    The nodes of each gate perform pointer-jumping aggregating the minimum of their identifiers. As each node can be part of a horizontal and a vertical gate, we add a bit indicating the orientation for tiebreaking. The resulting identifiers are of size $\bigO(\log n)$ and are obtained in $\bigO(\log n)$ rounds. The gate identifiers can be distributed the same way the adjacent regions were distributed in \autoref{lem:gate_labels}.
\end{proof}
\begin{lemma}\label{lem:landmark_and_edge_labels}
    After establishing region identifiers and gate identifiers, we can label each edge of the landmark graph with the region identifiers of the regions the edge lies in and each landmark with the gate identifiers it lies on in $\bigO(1)$ rounds. The total size of these labels is $\bigO(|H|^2\cdot \log n)$ bits. 
\end{lemma}
\begin{proof}
    To establish the edge labels, adjacent landmarks communicate which regions they are part of, keeping only the regions containing both. As each landmark is of constant degree according to \autoref{obs:landmark_degree} this can be done in $\bigO(1)$ rounds. Further, each edge is part of two regions yielding a label size of $\bigO(\log n)$. During the establishment of the region identifiers, each landmark already learned which gates it lies on. Combining the identifiers of those gates yields a label of $\bigO(\log n)$ bits. As there are a total of $\bigO(|\mathcal{H}|^2)$ edges in the landmark graph and a total of $\bigO(|\mathcal{H}|^2)$ landmarks, this yields a total size of $\bigO(|H|^2\cdot \log n)$ bits. 
\end{proof}

 With these labels established, we can now show that they indeed suffice for the nodes to make routing decisions. 
\begin{lemma}
\label{lem:next_gate}
    The labels established in \autoref{lem:gate_labels}, \autoref{lem:gate_identifiers} and \autoref{lem:landmark_and_edge_labels} suffice to allow a node to decide which gate should be crossed next.
\end{lemma}   
\begin{proof}
    After computing the landmark path locally, $v$ checks the edge label of its first edge $e$ that does not lie within $v$'s region. If there is a gate connecting $v$'s region to one of  $e$'s regions (it has multiple regions, if both of its endpoints are on the same gate), that gate is the next gate. Otherwise the landmark path crosses the landmark before $e$ diagonally (i.e. the landmark is at the intersection of two gates the the next region is diagonally opposed to the current one) and $v$ picks an arbitrary gate connected to that landmark. This process is depicted in \autoref{fig:next_gate}.
\end{proof}

\begin{figure}[H]
    \centering
    \resizebox{!}{0.33\textwidth}{%
        \begin{tikzpicture}
    \fill [pattern=north east lines, pattern color=gray] (0,.5) -- (0,-4) -- (5,-4) -- (5,-4.5) -- (-.5,-4.5) -- (-.5,.5);
    \fill [pattern=north east lines, pattern color=gray] (6,-4.5) -- (6,0) -- (1,0) -- (1,0.5) -- (6.5,0.5) -- (6.5,-4.5);
    \draw[thick] (0,.5) -- (0,-4) -- (5,-4) -- (5,-4.5);
    \draw[thick] (6,-4.5) -- (6,0) -- (1,0) -- (1,.5);

    \draw[tenn,thick] (.5,-1) -- (1.5,-2) to [bend right] (3,-2) to (4.5,-4) to (5.5,-4);
    \node[draw=tenn,fill=tenn,circle,inner sep=2pt] at (.5,-1) {};
    \node[draw=tenn,fill=tenn,circle,inner sep=2pt] at (5.5,-4) {};
    \node[tenn] at (2.25, -2.4) {$e_1$};
    \node[tenn] at (3.95, -2.8) {$e_2$};
    \node[tenn] at (.7,-.8) {$s$};
    \node[tenn] at (5.7,-3.8) {$t$};

    \draw[deepcerulean,thick] (1.5,0) -- (1.5,-4);
    \node[draw=deepcerulean,fill=deepcerulean,circle,inner sep=2pt] at (1.5,0) {};
    \node[draw=deepcerulean,fill=deepcerulean,circle,inner sep=2pt] at (1.5,-4) {};
    \draw[deepcerulean,thick] (3,0) -- (3,-4);
    \node[draw=deepcerulean,fill=deepcerulean,circle,inner sep=2pt] at (3,0) {};
    \node[draw=deepcerulean,fill=deepcerulean,circle,inner sep=2pt] at (3,-4) {};
    \draw[deepcerulean,thick] (4.5,0) -- (4.5,-4);
    \node[draw=deepcerulean,fill=deepcerulean,circle,inner sep=2pt] at (4.5,0) {};
    \node[draw=deepcerulean,fill=deepcerulean,circle,inner sep=2pt] at (4.5,-4) {};
    \draw[deepcerulean,thick] (0,-2) -- (6,-2);
    \node[draw=deepcerulean,fill=deepcerulean,circle,inner sep=2pt] at (0,-2) {};
    \node[draw=deepcerulean,fill=deepcerulean,circle,inner sep=2pt] at (1.5,-2) {};
    \node[draw=deepcerulean,fill=deepcerulean,circle,inner sep=2pt] at (3,-2) {};
    \node[draw=deepcerulean,fill=deepcerulean,circle,inner sep=2pt] at (4.5,-2) {};
    \node[draw=deepcerulean,fill=deepcerulean,circle,inner sep=2pt] at (6,-2) {};
    
    \node[deepcerulean] at (2.25,-1.75) {$R_1$};
    \node[deepcerulean] at (3.75,-1.75) {$R_2$};
    \node[deepcerulean] at (2.25,-3.75) {$R_3$};
    \node[deepcerulean] at (3.75,-3.75) {$R_4$};
\end{tikzpicture}
    }
    \caption{A packet is to be sent from $s$ to $t$. As the edge $e_1$ is aligned with the horizontal gate, it is considered to be in both region $R_1$ and region $R_3$. As $R_1$ is adjacent to the region containing $s$, it is the next region. Once the packet reaches $R_1$, $e_2$ will be considered to determine the next region. $e_2$ entirely lies within region $R_4$, which is not adjacent to $R_1$. Hence it is possible to either cross the gate towards region $R_2$ next or to cross the gate towards $R_3$ next.}
    \label{fig:next_gate}
\end{figure}
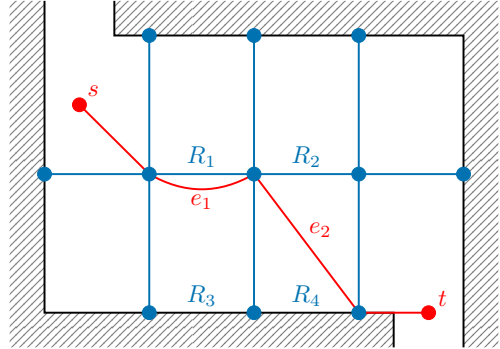

Next, we have every node learn the entire landmark graph including these labels. This gives each node enough knowledge about the grid graph to make local routing decisions. 

\begin{lemma}\label{lem:distribute_landmark_graph}
    Each node can learn the landmark graph and the labels from \autoref{lem:landmark_and_edge_labels} in time $\bigO(|\mathcal{H}|^2 + \log n)$.
\end{lemma}
\begin{proof}
    To achieve the runtime of the theorem, we employ the technique from \ref{lem:broadcast_and_aggregation}. As it requires a path graph to work, we start by establishing one using the construction of \autoref{lem:ncc0_to_ncc_in_grids}. We can use that path graph to fix an ordering on the landmarks by performing pointer-jumping on the non-landmark nodes to reduce the distances between landmarks to $\bigO(\log n)$ and sending the landmarks' identifiers over the established shortcuts afterwards. This takes $\bigO(\log n)$ rounds.
    As each landmark now knows its successor, the first landmark can start broadcasting its neighbors over the butterfly network. Once it is done, it can inform that successor to do the same. By pipelining this process, and as $\bigO(|\mathcal{H}|^2)$ messages of size $\bigO(\log n)$ have to be sent according to \autoref{lem:landmark_properties} and \autoref{lem:landmark_and_edge_labels}, the entire broadcast takes $\bigO(|\mathcal{H}|^2+\log n)$ rounds.
\end{proof}

Applying the algorithm from \cite{Coy2022}, we obtain the following corollary.
\begin{corollary}
    For each simple region, we can compute exact routing tables in time $O(\log n)$.
\end{corollary}

We conclude with a corollary combining the round complexity of each stage of the pre-processing (by \autoref{thm:convex_decomposition}, \autoref{lem:landmark_graph_construction}, and the results in this section):
\begin{corollary}\label{cor:preprocessing_runtime}
    The preprocessing takes $\bigO(|\mathcal{H}|^2+\log n)$ time.
\end{corollary}
As well as the landmark graph, nodes require some more information about the target node to successfully route a packet. Specifically, each node needs to know the target's distance to its adjacent landmarks, and its region label. We include these in the node labels and assume it to be included in a routing request\footnote{Note that each node knows all information required to construct its own label and any node knowing that node's identifier could then request this label via a global connection.}. Each node label consists of the node's distance to important landmarks ($O(\log n)$ bits, \autoref{lem:landmark_properties}) and its region label ($O(\log n)$ bits).



Additionally to the node labels, each node makes use of some additional information to forward a packet. Specifically, it uses its distance to important landmarks (already part of node label), its distance to region boundaries ($O(\log n)$ bits), its region routing tables ($O(\log n)$ bits) and the Landmark graph ($O(|\mathcal{H}|^2\cdot\log n)$ bits, \autoref{lem:landmark_properties}). All of this information has been collected as a part of the preprocessing and is now stored locally by the nodes.
This directly implies the following corollary:
\begin{corollary}\label{cor:label_size_and_memory_requirement}
    Node labels have $O(\log n)$ bits and routing tables have $O(|\mathcal{H}|^2\cdot\log n)$ bits.
\end{corollary}

To present our routing scheme, we describe how a node $v=(v.id, v.D, v.rid)\neq t$ forwards a packet with destination $t=(t.id, t.D, t.rid)$. $v.id$ and $t.id$ correspond to the nodes' identifiers, $v.D$ and $t.D$ are the distances to their important landmarks and $v.rid$ and $t.rid$ denote their region identifier. See \autoref{alg:route-packet} for the pseudocode.

If the packet is already in the correct region, i.e. $v$ and $t$ have the same region identifier, we can use the routing algorithm for simple grid graphs given in \cite{Coy2022}. Otherwise, $v$ locally augments its landmark graph by adding itself and $t$ to it, and adding connections to the adjacent landmarks according to the distance values stored in $v.D$ and $t.D$. Using the augmented landmark graph, $v$ can locally compute a shortest $vt$-path and thereby infer the next gate (as in \autoref{lem:next_gate}). Each node can learn which of its neighbors is closest to a specific gate of its region in $O(\log n)$ rounds (see \autoref{lem:gate_sssp}), it can forward the packet along a shortest path towards that gate.
%

To conclude \autoref{thm:grid_graph_routing_scheme}, we note that \autoref{thm:landmark_path_correctness} yields that we pick the correct sequence of gates and \autoref{lem:closest_point_routing} yields that repeatedly routing the package to the closest point of the next gate in that sequence results in a shortest path.
The runtime, the label size and the additional memory requirement result from \autoref{cor:preprocessing_runtime} and \autoref{cor:label_size_and_memory_requirement}.
Finally, \cite[Theorem 3.8]{Coy2022} yields the extension to UDGs described in \autoref{cor:udg_routing_scheme}.

\begin{algorithm}[H]
    \caption{\code{forward-packet}}
	\label{alg:route-packet}
	\begin{algorithmic}[1]
        \If{$v.rid=t.rid$}
            \State Forward packet using region routing table
        \Else
            \State Augment landmark graph using $v.D$ and $t.D$
            \State Compute shortest path from $v$ to $t$ in augmented landmark graph
            \State Send packet towards next gate according to shortest path
        \EndIf
	\end{algorithmic}
\end{algorithm}

\section{Conclusion}

We believe that there are several interesting directions for interesting follow-up work. Efficiently computing compact routing schemes in more general classes of geometrically interesting graphs (for example planar graphs or visibility graphs) is a natural next step. We suspect that an extension to $3$ and higher dimensions might be quite difficult (in particular the geometry required will certainly be more challenging), but the $3$-dimensional case could have practical applicability in sensor networks and swarm robotics.

\phantomsection
\addcontentsline{toc}{section}{References}
\bibliographystyle{plainurl} 
\bibliography{long_version}

\begin{thebibliography}{10}

\bibitem{Anagnostides2021}
Ioannis Anagnostides and Themis Gouleakis.
\newblock {Deterministic Distributed Algorithms and Lower Bounds in the Hybrid
  Model}.
\newblock In Seth Gilbert, editor, {\em 35th International Symposium on
  Distributed Computing (DISC 2021)}, volume 209 of {\em Leibniz International
  Proceedings in Informatics (LIPIcs)}, pages 5:1--5:19, Dagstuhl, Germany,
  2021. Schloss Dagstuhl -- Leibniz-Zentrum f{\"u}r Informatik.
\newblock URL: \url{https://drops.dagstuhl.de/opus/volltexte/2021/14807}, \href
  {https://doi.org/10.4230/LIPIcs.DISC.2021.5}
  {\path{doi:10.4230/LIPIcs.DISC.2021.5}}.

\bibitem{Asadi2016}
Arash Asadi, Vincenzo Mancuso, and Rohit Gupta.
\newblock An sdr-based experimental study of outband d2d communications.
\newblock In {\em IEEE INFOCOM 2016 - The 35th Annual IEEE International
  Conference on Computer Communications}, pages 1--9, 2016.
\newblock \href {https://doi.org/10.1109/INFOCOM.2016.7524372}
  {\path{doi:10.1109/INFOCOM.2016.7524372}}.

\bibitem{AGG+18}
John Augustine, Mohsen Ghaffari, Robert Gmyr, Kristian Hinnenthal, Fabian Kuhn,
  Jason Li, and Christian Scheideler.
\newblock Distributed computation in node-capacitated networks.
\newblock In {\em Proc. of the 31st ACM Symposium on Parallelism in Algorithms
  and Architectures (SPAA 2019)}, pages 69--79, 2019.

\bibitem{Augustine2020}
John Augustine, Kristian Hinnenthal, Fabian Kuhn, Christian Scheideler, and
  Philipp Schneider.
\newblock Shortest paths in a hybrid network model.
\newblock In {\em Proceedings of the Thirty-First Annual ACM-SIAM Symposium on
  Discrete Algorithms}, SODA '20, page 1280–1299, USA, 2020. Society for
  Industrial and Applied Mathematics.

\bibitem{BoseMSU01}
Prosenjit Bose, Pat Morin, Ivan Stojmenovic, and Jorge Urrutia.
\newblock Routing with guaranteed delivery in ad hoc wireless networks.
\newblock {\em Wireless Networks}, 7(6):609--616, 2001.
\newblock \href {https://doi.org/10.1023/A:1012319418150}
  {\path{doi:10.1023/A:1012319418150}}.

\bibitem{CastenowKS20}
Jannik Castenow, Christina Kolb, and Christian Scheideler.
\newblock A bounding box overlay for competitive routing in hybrid
  communication networks.
\newblock In {\em Proc. of the 21st International Conference on Distributed
  Computing and Networking (ICDCN 2020)}, pages 14:1--14:10, 2020.
\newblock \href {https://doi.org/10.1145/3369740.3369777}
  {\path{doi:10.1145/3369740.3369777}}.

\bibitem{CensorHillel2021}
Keren Censor-Hillel, Dean Leitersdorf, and Volodymyr Polosukhin.
\newblock {Distance Computations in the Hybrid Network Model via Oracle
  Simulations}.
\newblock In Markus Bl\"{a}ser and Benjamin Monmege, editors, {\em 38th
  International Symposium on Theoretical Aspects of Computer Science (STACS
  2021)}, volume 187 of {\em Leibniz International Proceedings in Informatics
  (LIPIcs)}, pages 21:1--21:19, Dagstuhl, Germany, 2021. Schloss Dagstuhl --
  Leibniz-Zentrum f{\"u}r Informatik.
\newblock URL: \url{https://drops.dagstuhl.de/opus/volltexte/2021/13666}, \href
  {https://doi.org/10.4230/LIPIcs.STACS.2021.21}
  {\path{doi:10.4230/LIPIcs.STACS.2021.21}}.

\bibitem{CensorHillel2021a}
Keren Censor-Hillel, Dean Leitersdorf, and Volodymyr Polosukhin.
\newblock On sparsity awareness in distributed computations.
\newblock In {\em Proceedings of the 33rd ACM Symposium on Parallelism in
  Algorithms and Architectures}, SPAA '21, page 151–161, New York, NY, USA,
  2021. Association for Computing Machinery.
\newblock \href {https://doi.org/10.1145/3409964.3461798}
  {\path{doi:10.1145/3409964.3461798}}.

\bibitem{Coy2022}
Sam Coy, Artur Czumaj, Michael Feldmann, Kristian Hinnenthal, Fabian Kuhn,
  Christian Scheideler, Philipp Schneider, and Martijn Struijs.
\newblock {Near-Shortest Path Routing in Hybrid Communication Networks}.
\newblock In Quentin Bramas, Vincent Gramoli, and Alessia Milani, editors, {\em
  25th International Conference on Principles of Distributed Systems (OPODIS
  2021)}, volume 217 of {\em Leibniz International Proceedings in Informatics
  (LIPIcs)}, pages 11:1--11:23, Dagstuhl, Germany, 2022. Schloss Dagstuhl --
  Leibniz-Zentrum f{\"u}r Informatik.
\newblock URL: \url{https://drops.dagstuhl.de/opus/volltexte/2022/15786}, \href
  {https://doi.org/10.4230/LIPIcs.OPODIS.2021.11}
  {\path{doi:10.4230/LIPIcs.OPODIS.2021.11}}.

\bibitem{CCFHKSSS21}
Sam Coy, Artur Czumaj, Michael Feldmann, Fabian Kuhn, Kristian Hinnenthal,
  Christian Scheideler, Philipp Schneider, and Martijn Struijs.
\newblock Near-shortest path routing in hybrid communication networks.
\newblock In {\em Proc. of the 25th International Conference on Principles of
  Distributed Systems (OPODIS 2021)}, 2021.

\bibitem{Elkin2016}
Michael Elkin and Ofer Neiman.
\newblock On efficient distributed construction of near optimal routing
  schemes.
\newblock In {\em Proceedings of the 2016 ACM Symposium on Principles of
  Distributed Computing}, pages 235--244, 2016.

\bibitem{Feldmann2020}
Michael Feldmann, Kristian Hinnenthal, and Christian Scheideler.
\newblock Fast hybrid network algorithms for shortest paths in sparse graphs.
\newblock In {\em Proc. of the 24th International Conference on Principles of
  Distributed Systems (OPODIS 2020)}, pages 31:1--31:16, 2020.
\newblock \href {https://doi.org/10.4230/LIPIcs.OPODIS.2020.31}
  {\path{doi:10.4230/LIPIcs.OPODIS.2020.31}}.

\bibitem{HalldorssonT19}
Magn{\'{u}}s~M. Halld{\'{o}}rsson and Tigran Tonoyan.
\newblock Plain {SINR} is enough!
\newblock In {\em Proc. of the 38th ACM Symposium on Principles of Distributed
  Computing (PODC 2019)}, pages 127--136, 2019.
\newblock \href {https://doi.org/10.1145/3293611.3331602}
  {\path{doi:10.1145/3293611.3331602}}.

\bibitem{Kuhn2020}
Fabian Kuhn and Philipp Schneider.
\newblock Computing shortest paths and diameter in the hybrid network model.
\newblock In {\em Proceedings of the 39th Symposium on Principles of
  Distributed Computing}, PODC '20, page 109–118, New York, NY, USA, 7 2020.
  Association for Computing Machinery.
\newblock \href {https://doi.org/10.1145/3382734.3405719}
  {\path{doi:10.1145/3382734.3405719}}.

\bibitem{Kuhn2022}
Fabian Kuhn and Philipp Schneider.
\newblock {Routing Schemes and Distance Oracles in the Hybrid Model}.
\newblock In {\em International Symposium on Distributed Computing (DISC)},
  volume 246, pages 28:1--28:22, 2022.

\bibitem{KuhnWZZ03}
Fabian Kuhn, Roger Wattenhofer, Yan Zhang, and Aaron Zollinger.
\newblock Geometric ad-hoc routing: of theory and practice.
\newblock In {\em Proc. of the 22nd ACM Symposium on Principles of Distributed
  Computing (PODC 2003)}, pages 63--72, 2003.
\newblock \href {https://doi.org/10.1145/872035.872044}
  {\path{doi:10.1145/872035.872044}}.

\bibitem{KuhnWZ02}
Fabian Kuhn, Roger Wattenhofer, and Aaron Zollinger.
\newblock Asymptotically optimal geometric mobile ad-hoc routing.
\newblock In {\em Proc. of the 6th International Workshop on Discrete
  Algorithms and Methods for Mobile Computing and Communications (DIAL-M
  2002)}, pages 24--33, 2002.
\newblock \href {https://doi.org/10.1145/570810.570814}
  {\path{doi:10.1145/570810.570814}}.

\bibitem{KuhnWZ03}
Fabian Kuhn, Roger Wattenhofer, and Aaron Zollinger.
\newblock Worst-case optimal and average-case efficient geometric ad-hoc
  routing.
\newblock In {\em Proc. of the 4th ACM Interational Symposium on Mobile Ad Hoc
  Networking and Computing (MobiHoc 2003)}, pages 267--278, 2003.
\newblock \href {https://doi.org/10.1145/778415.778447}
  {\path{doi:10.1145/778415.778447}}.

\bibitem{Lenzen2013a}
Christoph Lenzen and Boaz Patt-Shamir.
\newblock Fast routing table construction using small messages: Extended
  abstract.
\newblock In {\em Proceedings of the Forty-Fifth Annual ACM Symposium on Theory
  of Computing}, STOC '13, page 381–390, New York, NY, USA, 2013. Association
  for Computing Machinery.
\newblock \href {https://doi.org/10.1145/2488608.2488656}
  {\path{doi:10.1145/2488608.2488656}}.

\bibitem{Lenzen2015}
Christoph Lenzen and Boaz Patt-Shamir.
\newblock Fast partial distance estimation and applications.
\newblock In {\em Proceedings of the 2015 ACM Symposium on Principles of
  Distributed Computing}, pages 153--162, 2015.

\bibitem{Lynch1996}
Nancy~A Lynch.
\newblock {\em Distributed algorithms}.
\newblock Elsevier, 1996.

\bibitem{Rozhon2022}
V{\'a}clav Rozho{\v{n}}, Christoph Grunau, Bernhard Haeupler, Goran Zuzic, and
  Jason Li.
\newblock Undirected (1+$\varepsilon$)-shortest paths via minor-aggregates:
  near-optimal deterministic parallel and distributed algorithms.
\newblock In {\em Symposium on Theory of Computing}, pages 478--487, 2022.

\bibitem{Sarma2012}
Atish~Das Sarma, Stephan Holzer, Liah Kor, Amos Korman, Danupon Nanongkai,
  Gopal Pandurangan, David Peleg, and Roger Wattenhofer.
\newblock Distributed verification and hardness of distributed approximation.
\newblock {\em SIAM Journal on Computing}, 41(5):1235--1265, 2012.

\bibitem{Schneider2023}
Philipp Schneider.
\newblock {\em Power and Limitations of Hybrid Communication Networks}.
\newblock PhD thesis, University of Freiburg, 2023.
\newblock \href {https://doi.org/10.6094/UNIFR/232804}
  {\path{doi:10.6094/UNIFR/232804}}.

\end{thebibliography}

\appendix

\section{Intuition for the Main Algorithms}
\label{sec:intuition}

In this section, we present a collection of figures that depict the different stages of the execution of our main algorithms, i.e., the preprocessing algortihm and the routing algorithm. Our goal is to aid the reader by giving intuitive presentations of what actions are performed during each of them. Hence, the section is written to be self contained and understandable without reading the definitions presented in the paper, and we do not worry about computability arguments or analytically well-defined wording here. For a rigorous presentation, we refer the reader to the main part of the paper. 
\subsection{Setting}
We start by presenting our setting and introducing how each of the figures will be drawn. As meaningful examples of grid graphs become too large to present really quickly, we opted for a more abstract representation, drawing a white area for grid spaces and a hatched area for inner holes and the outer hole. As suggested in \autoref{fig:gallery_setting}, we do not assume any structure for our grid graph (except connectedness), i.e., the induced polygon can have an arbitrary shape and contain any amount of arbitrarily shaped holes. 
\begin{figure}[H]
    \centering
    \resizebox{!}{0.42\textwidth}{%
        \begin{tikzpicture}
    \fill [pattern=north east lines, pattern color=gray] (-3, -2.65) -- (-3, 2.65) -- (3,2.65) -- (3,-2.65) -- cycle;
    \draw [black, fill=white] (-2.75, -2.4) -- (-2.75, 2.4) -- (2.75,2.4) -- (2.75,-2.4) -- cycle;

    \fill [draw=black, pattern=north east lines, pattern color=gray] (-2, -1) -- (-2,-0.7) -- (-0.3,-0.7) -- (-0.3,0.7) -- (-2,0.7) -- (-2, 1) -- (0,1) -- (0,-1) -- cycle;

    \fill [draw=black, pattern=north east lines, pattern color=gray] (-1.6,-0.3) -- (-1,-0.3) -- (-1,0.3) -- (-1.6,0.3) -- cycle;

    \fill [draw=black, pattern=north east lines, pattern color=gray] (-1,1.4) -- (-1,2) -- (1,2) -- (1, 1.4) -- cycle;

    \fill [draw=black, pattern=north east lines, pattern color=gray] (-1,-1.4) -- (-1,-2) -- (2,-2) -- (2,-1.7) -- (2.2,-1.7) -- (2.2,-1.3) -- (2,-1.3) -- (2,-.7) -- (2.4,-.7) -- (2.4,-.3) -- (2,-.3) -- (2, 0) -- (1.4,0) -- (1.4,-1.4) -- (.8,-1.4) -- (.8,-.4) -- (.4,-.4) -- (.4,-1.4) -- cycle;
\end{tikzpicture}
    }
    \caption{Abstract representation of the grid graph (white) and the holes (hatched). Both can be arbitrarily shaped.}
    \label{fig:gallery_setting}
\end{figure}
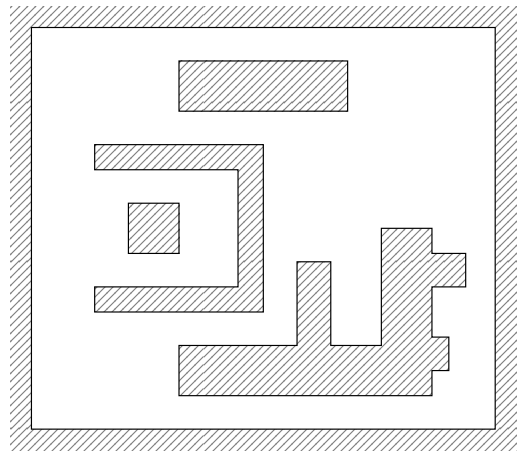

\subsection{Regionalization}
During the Regionalization, we divide our grid graph into regions that do not contain any holes and are path-convex, i.e., for any two nodes in a region, there is a shortest path connecting those nodes wholly within that region. To achieve this, we first mark the leftmost node of each hole (breaking ties by picking the northernmost one) and cut the grid graph vertically through that marked node and horizontally to its right. As each hole is connected to another hole this way and as we can not create cyclic connectedness due to the leftmost node being used, we obtain hole-free regions this way. To satisfy path-convexity, we first split the resulting regions into tunnels, i.e. regions with two exits, by identifying line-shaped junctions and cutting at those aswell. As these tunnels may not be path-convex yet, we further subdivide them by adding several additional cuts to ensure that their length is effectively halved. Any shortest path between nodes of the resulting regions would have to go through at least one of the other regions, intuitively meaning that it would spend more than half of the regions diameter outside the region. Therefore we can find a shorter path inside the region, which yields path-convexity. Each of these steps is presented in a picture of \autoref{fig:gallery_regionalization}
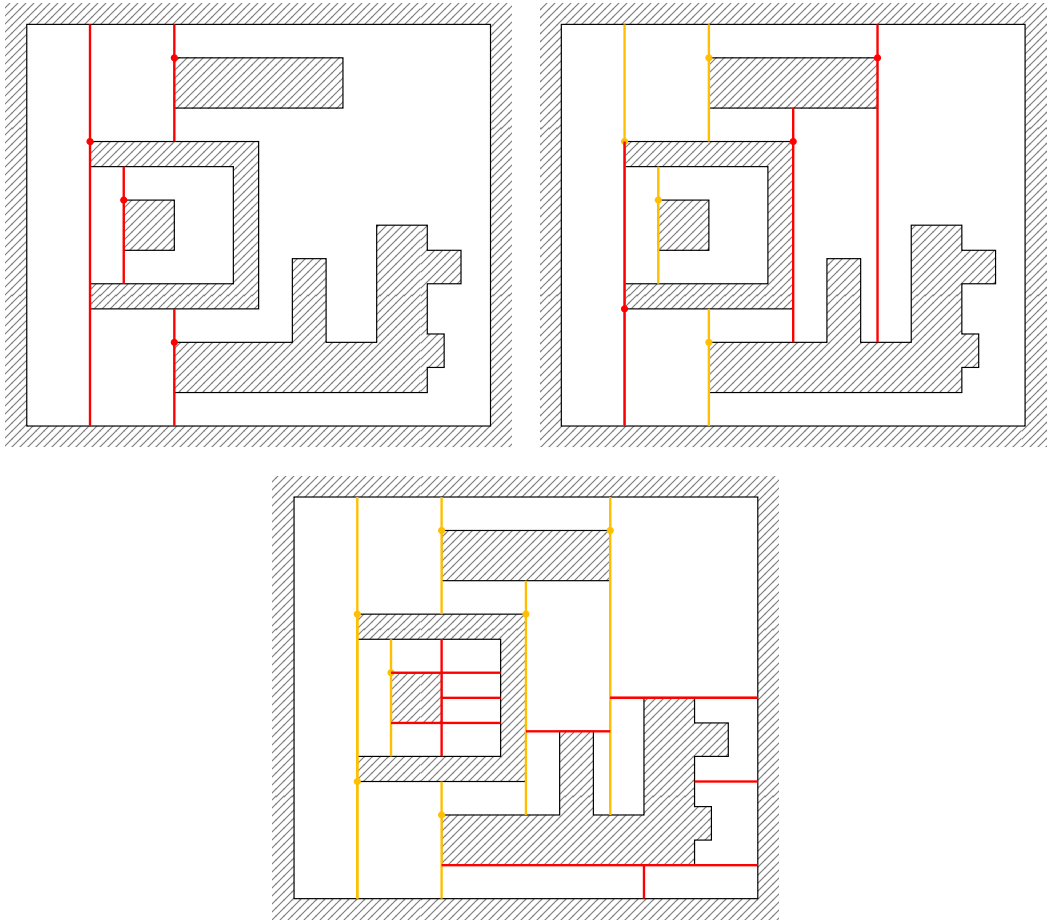
\begin{figure}
    \centering
    \resizebox{!}{0.42\textwidth}{%
        \begin{tikzpicture}
    \fill [pattern=north east lines, pattern color=gray] (-3, -2.65) -- (-3, 2.65) -- (3,2.65) -- (3,-2.65) -- cycle;
    \draw [black, fill=white] (-2.75, -2.4) -- (-2.75, 2.4) -- (2.75,2.4) -- (2.75,-2.4) -- cycle;

    \fill [draw=black, pattern=north east lines, pattern color=gray] (-2, -1) -- (-2,-0.7) -- (-0.3,-0.7) -- (-0.3,0.7) -- (-2,0.7) -- (-2, 1) -- (0,1) -- (0,-1) -- cycle;

    \fill [draw=black, pattern=north east lines, pattern color=gray] (-1.6,-0.3) -- (-1,-0.3) -- (-1,0.3) -- (-1.6,0.3) -- cycle;

    \fill [draw=black, pattern=north east lines, pattern color=gray] (-1,1.4) -- (-1,2) -- (1,2) -- (1, 1.4) -- cycle;

    \fill [draw=black, pattern=north east lines, pattern color=gray] (-1,-1.4) -- (-1,-2) -- (2,-2) -- (2,-1.7) -- (2.2,-1.7) -- (2.2,-1.3) -- (2,-1.3) -- (2,-.7) -- (2.4,-.7) -- (2.4,-.3) -- (2,-.3) -- (2, 0) -- (1.4,0) -- (1.4,-1.4) -- (.8,-1.4) -- (.8,-.4) -- (.4,-.4) -- (.4,-1.4) -- cycle;

    \draw[tenn,fill=tenn] (-2,1) circle (1pt);
    \draw[tenn,fill=tenn] (-1.6,.3) circle (1pt);
    \draw[tenn,fill=tenn] (-1,2) circle (1pt);
    \draw[tenn,fill=tenn] (-1,-1.4) circle (1pt);

    \draw [tenn,thick] (-2,-2.4) -- (-2,2.4); 
    \draw [tenn,thick] (-1.6,-0.7) -- (-1.6,0.7);
    \draw [tenn,thick] (-1,1) -- (-1,2.4); 
    \draw [tenn,thick] (-1,-2.4) -- (-1,-1);
\end{tikzpicture}
    }
    \resizebox{!}{0.42\textwidth}{%
        \begin{tikzpicture}
    \fill [pattern=north east lines, pattern color=gray] (-3, -2.65) -- (-3, 2.65) -- (3,2.65) -- (3,-2.65) -- cycle;
    \draw [black, fill=white] (-2.75, -2.4) -- (-2.75, 2.4) -- (2.75,2.4) -- (2.75,-2.4) -- cycle;

    \fill [draw=black, pattern=north east lines, pattern color=gray] (-2, -1) -- (-2,-0.7) -- (-0.3,-0.7) -- (-0.3,0.7) -- (-2,0.7) -- (-2, 1) -- (0,1) -- (0,-1) -- cycle;

    \fill [draw=black, pattern=north east lines, pattern color=gray] (-1.6,-0.3) -- (-1,-0.3) -- (-1,0.3) -- (-1.6,0.3) -- cycle;

    \fill [draw=black, pattern=north east lines, pattern color=gray] (-1,1.4) -- (-1,2) -- (1,2) -- (1, 1.4) -- cycle;

    \fill [draw=black, pattern=north east lines, pattern color=gray] (-1,-1.4) -- (-1,-2) -- (2,-2) -- (2,-1.7) -- (2.2,-1.7) -- (2.2,-1.3) -- (2,-1.3) -- (2,-.7) -- (2.4,-.7) -- (2.4,-.3) -- (2,-.3) -- (2, 0) -- (1.4,0) -- (1.4,-1.4) -- (.8,-1.4) -- (.8,-.4) -- (.4,-.4) -- (.4,-1.4) -- cycle;

    \draw[starship,fill=starship] (-2,1) circle (1pt);
    \draw[starship,fill=starship] (-1.6,.3) circle (1pt);
    \draw[starship,fill=starship] (-1,2) circle (1pt);
    \draw[starship,fill=starship] (-1,-1.4) circle (1pt);

    \draw [starship,thick] (-2,-2.4) -- (-2,2.4); 
    \draw [starship,thick] (-1.6,-0.7) -- (-1.6,0.7);
    \draw [starship,thick] (-1,1) -- (-1,2.4); 
    \draw [starship,thick] (-1,-2.4) -- (-1,-1);

    \draw[tenn,fill=tenn] (0,1) circle (1pt);
    \draw[tenn,fill=tenn] (-2,-1) circle (1pt);
    \draw[tenn,fill=tenn] (1,2) circle (1pt);

    \draw [tenn,thick] (-2,1) -- (-2,-2.4);
    \draw [tenn,thick] (0,-1.4) -- (0,1.4);
    \draw [tenn,thick] (1,-1.4) -- (1,2.4);
\end{tikzpicture}
    }

    \vspace{1em}
    \resizebox{!}{0.42\textwidth}{%
        \begin{tikzpicture}
    \fill [pattern=north east lines, pattern color=gray] (-3, -2.65) -- (-3, 2.65) -- (3,2.65) -- (3,-2.65) -- cycle;
    \draw [black, fill=white] (-2.75, -2.4) -- (-2.75, 2.4) -- (2.75,2.4) -- (2.75,-2.4) -- cycle;

    \fill [draw=black, pattern=north east lines, pattern color=gray] (-2, -1) -- (-2,-0.7) -- (-0.3,-0.7) -- (-0.3,0.7) -- (-2,0.7) -- (-2, 1) -- (0,1) -- (0,-1) -- cycle;

    \fill [draw=black, pattern=north east lines, pattern color=gray] (-1.6,-0.3) -- (-1,-0.3) -- (-1,0.3) -- (-1.6,0.3) -- cycle;

    \fill [draw=black, pattern=north east lines, pattern color=gray] (-1,1.4) -- (-1,2) -- (1,2) -- (1, 1.4) -- cycle;

    \fill [draw=black, pattern=north east lines, pattern color=gray] (-1,-1.4) -- (-1,-2) -- (2,-2) -- (2,-1.7) -- (2.2,-1.7) -- (2.2,-1.3) -- (2,-1.3) -- (2,-.7) -- (2.4,-.7) -- (2.4,-.3) -- (2,-.3) -- (2, 0) -- (1.4,0) -- (1.4,-1.4) -- (.8,-1.4) -- (.8,-.4) -- (.4,-.4) -- (.4,-1.4) -- cycle;


    \draw[starship,fill=starship] (-2,1) circle (1pt);
    \draw[starship,fill=starship] (-1.6,.3) circle (1pt);
    \draw[starship,fill=starship] (-1,2) circle (1pt);
    \draw[starship,fill=starship] (-1,-1.4) circle (1pt);

    \draw [starship,thick] (-2,-2.4) -- (-2,2.4); 
    \draw [starship,thick] (-1.6,-0.7) -- (-1.6,0.7);
    \draw [starship,thick] (-1,1) -- (-1,2.4); 
    \draw [starship,thick] (-1,-2.4) -- (-1,-1);

    \draw[starship,fill=starship] (0,1) circle (1pt);
    \draw[starship,fill=starship] (-2,-1) circle (1pt);
    \draw[starship,fill=starship] (1,2) circle (1pt);

    \draw [starship,thick] (-2,1) -- (-2,-2.4);
    \draw [starship,thick] (0,-1.4) -- (0,1.4);
    \draw [starship,thick] (1,-1.4) -- (1,2.4);

    \draw [tenn,thick] (-1,-0.7) -- (-1,0.7);
    \draw [tenn,thick] (1.4,-2.4) -- (1.4,-2);

    \draw [tenn,thick] (1,0) -- (2.75,0);
    \draw [tenn,thick] (-1,0) -- (-0.3,0);
    \draw [tenn,thick] (2,-1) -- (2.75,-1);


    \draw [tenn,thick] (0,-.4) -- (1,-.4);

    \draw [tenn,thick] (-1.6,.3) -- (-.3,.3);
    \draw [tenn,thick] (-1.6,-.3) -- (-.3,-.3);
    \draw [tenn,thick] (1,0) -- (2.75,0);
    \draw [tenn,thick] (-1,-2) -- (2.75,-2);
\end{tikzpicture}
    }
    \caption{Regionalization of the grid graph. One after the other, we divide the grid graph into simple regions (left), tunnels (center) and path-convex regions (right).}
    \label{fig:gallery_regionalization}
\end{figure}

\subsection{Landmark Graph}
To enable the nodes to locally make routing decisions, we have them learn the so-called landmark graph, which essentially is a small size skeleton graph that captures important distances of the graph and the different ways to navigate around its holes. To this end, we mark some nodes as \emph{landmarks}, see \autoref{fig:gallery_landmarkgraphnodes}. First, we add the endpoint of each region boundary to the landmark graph. Next, for each region boundary that does not wholly see its adjacent region boundary, we consider its set of closest nodes to that other boundary and add the outermost of these nodes to the landmark graph as well. Finally, for each landmark, we project it to all region boundaries that are orthogonally visible to it, i.e., consider the orthogonal line crossing the original node and add all nodes corresponding to crossings of that line with other region boundaries to the landmark graph.
\begin{figure}
    \centering
    \resizebox{!}{0.42\textwidth}{%
        \begin{tikzpicture}
    \fill [pattern=north east lines, pattern color=gray] (-3, -2.65) -- (-3, 2.65) -- (3,2.65) -- (3,-2.65) -- cycle;
    \draw [black, fill=white] (-2.75, -2.4) -- (-2.75, 2.4) -- (2.75,2.4) -- (2.75,-2.4) -- cycle;

    \fill [draw=black, pattern=north east lines, pattern color=gray] (-2, -1) -- (-2,-0.7) -- (-0.3,-0.7) -- (-0.3,0.7) -- (-2,0.7) -- (-2, 1) -- (0,1) -- (0,-1) -- cycle;

    \fill [draw=black, pattern=north east lines, pattern color=gray] (-1.6,-0.3) -- (-1,-0.3) -- (-1,0.3) -- (-1.6,0.3) -- cycle;

    \fill [draw=black, pattern=north east lines, pattern color=gray] (-1,1.4) -- (-1,2) -- (1,2) -- (1, 1.4) -- cycle;

    \fill [draw=black, pattern=north east lines, pattern color=gray] (-1,-1.4) -- (-1,-2) -- (2,-2) -- (2,-1.7) -- (2.2,-1.7) -- (2.2,-1.3) -- (2,-1.3) -- (2,-.7) -- (2.4,-.7) -- (2.4,-.3) -- (2,-.3) -- (2, 0) -- (1.4,0) -- (1.4,-1.4) -- (.8,-1.4) -- (.8,-.4) -- (.4,-.4) -- (.4,-1.4) -- cycle;



    \draw [starship,thick] (-2,-2.4) -- (-2,2.4); 
    \draw [starship,thick] (-1.6,-0.7) -- (-1.6,0.7);
    \draw [starship,thick] (-1,1) -- (-1,2.4); 
    \draw [starship,thick] (-1,-2.4) -- (-1,-1);


    \draw [starship,thick] (-2,1) -- (-2,-2.4);
    \draw [starship,thick] (0,-1.4) -- (0,1.4);
    \draw [starship,thick] (1,-1.4) -- (1,2.4);

    \draw [starship,thick] (-1,-0.7) -- (-1,0.7);
    \draw [starship,thick] (1.4,-2.4) -- (1.4,-2);

    \draw [starship,thick] (1,0) -- (2.75,0);
    \draw [starship,thick] (-1,0) -- (-0.3,0);
    \draw [starship,thick] (2,-1) -- (2.75,-1);


    \draw [starship,thick] (0,-.4) -- (1,-.4);

    \draw [starship,thick] (-1.6,.3) -- (-.3,.3);
    \draw [starship,thick] (-1.6,-.3) -- (-.3,-.3);
    \draw [starship,thick] (1,0) -- (2.75,0);
    \draw [starship,thick] (-1,-2) -- (2.75,-2);


    \node[draw=deepcerulean,fill=deepcerulean,circle,inner sep=.7pt] () at (-2,2.4) {};
    \node[draw=deepcerulean,fill=deepcerulean,circle,inner sep=.7pt] () at (-2,1) {};
    \node[draw=deepcerulean,fill=deepcerulean,circle,inner sep=.7pt] () at (-2,.7) {};
    \node[draw=deepcerulean,fill=deepcerulean,circle,inner sep=.7pt] () at (-2,-.7) {};
    \node[draw=deepcerulean,fill=deepcerulean,circle,inner sep=.7pt] () at (-2,-1) {};
    \node[draw=deepcerulean,fill=deepcerulean,circle,inner sep=.7pt] () at (-2,-2.4) {};

    \node[draw=deepcerulean,fill=deepcerulean,circle,inner sep=.7pt] () at (-1.6,.7) {};
    \node[draw=deepcerulean,fill=deepcerulean,circle,inner sep=.7pt] () at (-1.6,.3) {};
    \node[draw=deepcerulean,fill=deepcerulean,circle,inner sep=.7pt] () at (-1.6,-.3) {};
    \node[draw=deepcerulean,fill=deepcerulean,circle,inner sep=.7pt] () at (-1.6,-.7) {};
    
    \node[draw=deepcerulean,fill=deepcerulean,circle,inner sep=.7pt] () at (-1,2.4) {};
    \node[draw=deepcerulean,fill=deepcerulean,circle,inner sep=.7pt] () at (-1,2) {};
    \node[draw=deepcerulean,fill=deepcerulean,circle,inner sep=.7pt] () at (-1,1.4) {};
    \node[draw=deepcerulean,fill=deepcerulean,circle,inner sep=.7pt] () at (-1,1) {};
    \node[draw=deepcerulean,fill=deepcerulean,circle,inner sep=.7pt] () at (-1,.7) {};
    \node[draw=deepcerulean,fill=deepcerulean,circle,inner sep=.7pt] () at (-1,.3) {};
    \node[draw=deepcerulean,fill=deepcerulean,circle,inner sep=.7pt] () at (-1,0) {};
    \node[draw=deepcerulean,fill=deepcerulean,circle,inner sep=.7pt] () at (-1,-.3) {};
    \node[draw=deepcerulean,fill=deepcerulean,circle,inner sep=.7pt] () at (-1,-.7) {};
    \node[draw=deepcerulean,fill=deepcerulean,circle,inner sep=.7pt] () at (-1,-1) {};
    \node[draw=deepcerulean,fill=deepcerulean,circle,inner sep=.7pt] () at (-1,-1.4) {};
    \node[draw=deepcerulean,fill=deepcerulean,circle,inner sep=.7pt] () at (-1,-2) {};
    \node[draw=deepcerulean,fill=deepcerulean,circle,inner sep=.7pt] () at (-1,-2.4) {};

    \node[draw=deepcerulean,fill=deepcerulean,circle,inner sep=.7pt] () at (-.3,.3) {};
    \node[draw=deepcerulean,fill=deepcerulean,circle,inner sep=.7pt] () at (-.3,0) {};
    \node[draw=deepcerulean,fill=deepcerulean,circle,inner sep=.7pt] () at (-.3,-.3) {};

    \node[draw=deepcerulean,fill=deepcerulean,circle,inner sep=.7pt] () at (0,1.4) {};
    \node[draw=deepcerulean,fill=deepcerulean,circle,inner sep=.7pt] () at (0,1) {};
    \node[draw=deepcerulean,fill=deepcerulean,circle,inner sep=.7pt] () at (0,-.4) {};
    \node[draw=deepcerulean,fill=deepcerulean,circle,inner sep=.7pt] () at (0,-1) {};
    \node[draw=deepcerulean,fill=deepcerulean,circle,inner sep=.7pt] () at (0,-1.4) {};

    \node[draw=deepcerulean,fill=deepcerulean,circle,inner sep=.7pt] () at (1,2.4) {};
    \node[draw=deepcerulean,fill=deepcerulean,circle,inner sep=.7pt] () at (1,2) {};
    \node[draw=deepcerulean,fill=deepcerulean,circle,inner sep=.7pt] () at (1,1.4) {};
    \node[draw=deepcerulean,fill=deepcerulean,circle,inner sep=.7pt] () at (1,0) {};
    \node[draw=deepcerulean,fill=deepcerulean,circle,inner sep=.7pt] () at (1,-.4) {};
    \node[draw=deepcerulean,fill=deepcerulean,circle,inner sep=.7pt] () at (1,-1.4) {};

    \node[draw=deepcerulean,fill=deepcerulean,circle,inner sep=.7pt] () at (1.4,-2) {};
    \node[draw=deepcerulean,fill=deepcerulean,circle,inner sep=.7pt] () at (1.4,-2.4) {};
    
    \node[draw=deepcerulean,fill=deepcerulean,circle,inner sep=.7pt] () at (2,-1) {};

    \node[draw=deepcerulean,fill=deepcerulean,circle,inner sep=.7pt] () at (2.75,0) {};
    \node[draw=deepcerulean,fill=deepcerulean,circle,inner sep=.7pt] () at (2.75,-1) {};
    \node[draw=deepcerulean,fill=deepcerulean,circle,inner sep=.7pt] () at (2.75,-2) {};

    \node[draw=deepcerulean,fill=deepcerulean,circle,inner sep=.7pt] () at (2.4,0) {};
    \node[draw=deepcerulean,fill=deepcerulean,circle,inner sep=.7pt] () at (2.4,-1) {};

    \node[draw=deepcerulean,fill=deepcerulean,circle,inner sep=.7pt] () at (2.2,-1) {};
    \node[draw=deepcerulean,fill=deepcerulean,circle,inner sep=.7pt] () at (2.2,-2) {};

    \node[draw=deepcerulean,fill=deepcerulean,circle,inner sep=.7pt] () at (-2,2) {};
    \node[draw=deepcerulean,fill=deepcerulean,circle,inner sep=.7pt] () at (-2,1.4) {};
    \node[draw=deepcerulean,fill=deepcerulean,circle,inner sep=.7pt] () at (-2,.3) {};
    \node[draw=deepcerulean,fill=deepcerulean,circle,inner sep=.7pt] () at (-2,-.3) {};
    \node[draw=deepcerulean,fill=deepcerulean,circle,inner sep=.7pt] () at (-2,-1.4) {};
    \node[draw=deepcerulean,fill=deepcerulean,circle,inner sep=.7pt] () at (-2,-2) {};

    \node[draw=deepcerulean,fill=deepcerulean,circle,inner sep=.7pt] () at (0,0) {};

    \node[draw=deepcerulean,fill=deepcerulean,circle,inner sep=.7pt] () at (1,1) {};

    \node[draw=deepcerulean,fill=deepcerulean,circle,inner sep=.7pt] () at (2.4,-2) {};

\end{tikzpicture}
    }
    \caption{The nodeset of the landmark graph consists of the endpoints of the region boundaries, the points induced by obstructions between region boundaries and the points induced by projections.}
    \label{fig:gallery_landmarkgraphnodes}
\end{figure}
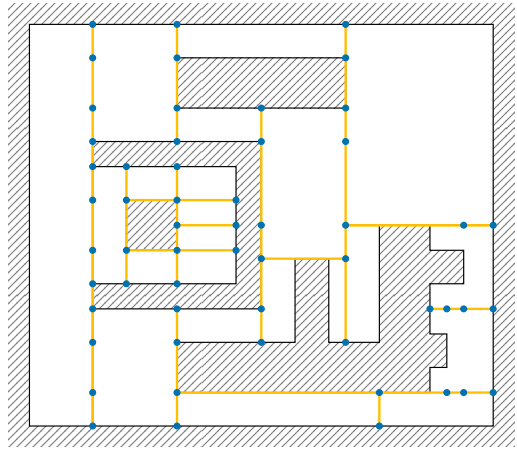
We connect two of these landmark nodes with an edge whenever they are on the same region boundary and no other landmark node is between them, or when they are on different region boundaries adjacent to the same region and they are closest to each other on those region boundaries.
\autoref{fig:gallery_landmarkgraphedges} presents how we connect two nodes of the landmark graphs. 
\begin{figure}
    \centering
    \resizebox{!}{0.42\textwidth}{%
        \input{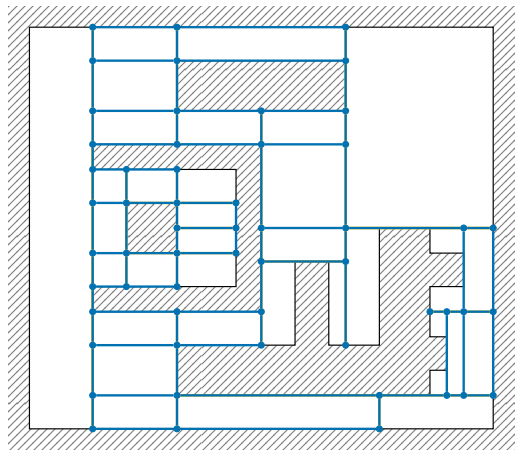}
    }
    \caption{Two nodes of the landmark graph are connected, if they are on the same region boundary without another node of the landmark graph between them, or if they are on adjacent region boudaries and no other landmark on the corresponding other boundary is closer to them.}
    \label{fig:gallery_landmarkgraphedges}
\end{figure}
We will ommit most edges of the landmark graph in the following pictures to keep them easy to read.

\subsection{Routing}
Our routing algorithm distinguishes between two cases. In the first case, depicted in \autoref{fig:gallery_routingeasy}, both source and target node of a routing request are in the same region. As our regions do not have holes, we can employ the routing scheme for hole free grid graphs from \cite{CCFHKSSS21}.
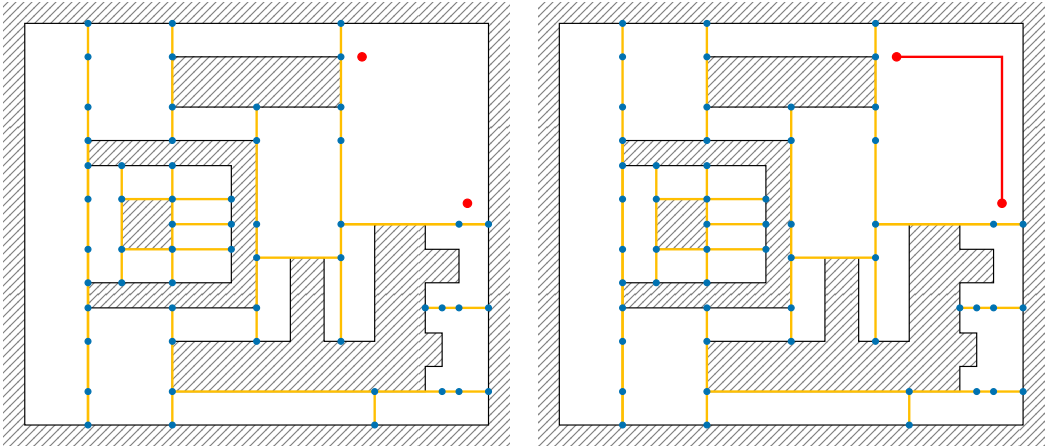
\begin{figure}
    \centering
    \resizebox{!}{0.42\textwidth}{%
            \begin{tikzpicture}
        \fill [pattern=north east lines, pattern color=gray] (-3, -2.65) -- (-3, 2.65) -- (3,2.65) -- (3,-2.65) -- cycle;
        \draw [black, fill=white] (-2.75, -2.4) -- (-2.75, 2.4) -- (2.75,2.4) -- (2.75,-2.4) -- cycle;
    
        \fill [draw=black, pattern=north east lines, pattern color=gray] (-2, -1) -- (-2,-0.7) -- (-0.3,-0.7) -- (-0.3,0.7) -- (-2,0.7) -- (-2, 1) -- (0,1) -- (0,-1) -- cycle;
    
        \fill [draw=black, pattern=north east lines, pattern color=gray] (-1.6,-0.3) -- (-1,-0.3) -- (-1,0.3) -- (-1.6,0.3) -- cycle;
    
        \fill [draw=black, pattern=north east lines, pattern color=gray] (-1,1.4) -- (-1,2) -- (1,2) -- (1, 1.4) -- cycle;
    
        \fill [draw=black, pattern=north east lines, pattern color=gray] (-1,-1.4) -- (-1,-2) -- (2,-2) -- (2,-1.7) -- (2.2,-1.7) -- (2.2,-1.3) -- (2,-1.3) -- (2,-.7) -- (2.4,-.7) -- (2.4,-.3) -- (2,-.3) -- (2, 0) -- (1.4,0) -- (1.4,-1.4) -- (.8,-1.4) -- (.8,-.4) -- (.4,-.4) -- (.4,-1.4) -- cycle;
    
    
    
        \draw [starship,thick] (-2,-2.4) -- (-2,2.4); 
        \draw [starship,thick] (-1.6,-0.7) -- (-1.6,0.7);
        \draw [starship,thick] (-1,1) -- (-1,2.4); 
        \draw [starship,thick] (-1,-2.4) -- (-1,-1);
    
    
        \draw [starship,thick] (-2,1) -- (-2,-2.4);
        \draw [starship,thick] (0,-1.4) -- (0,1.4);
        \draw [starship,thick] (1,-1.4) -- (1,2.4);
    
        \draw [starship,thick] (-1,-0.7) -- (-1,0.7);
        \draw [starship,thick] (1.4,-2.4) -- (1.4,-2);
    
        \draw [starship,thick] (1,0) -- (2.75,0);
        \draw [starship,thick] (-1,0) -- (-0.3,0);
        \draw [starship,thick] (2,-1) -- (2.75,-1);
    
    
        \draw [starship,thick] (0,-.4) -- (1,-.4);
    
        \draw [starship,thick] (-1.6,.3) -- (-.3,.3);
        \draw [starship,thick] (-1.6,-.3) -- (-.3,-.3);
        \draw [starship,thick] (1,0) -- (2.75,0);
        \draw [starship,thick] (-1,-2) -- (2.75,-2);
    
    
        \node[draw=deepcerulean,fill=deepcerulean,circle,inner sep=.7pt] (1-1-1) at (-2,2.4) {};
        \node[draw=deepcerulean,fill=deepcerulean,circle,inner sep=.7pt] (1-1-2) at (-2,1) {};
        \node[draw=deepcerulean,fill=deepcerulean,circle,inner sep=.7pt] (1-1-3) at (-2,.7) {};
        \node[draw=deepcerulean,fill=deepcerulean,circle,inner sep=.7pt] (1-1-4) at (-2,-.7) {};
        \node[draw=deepcerulean,fill=deepcerulean,circle,inner sep=.7pt] (1-1-5) at (-2,-1) {};
        \node[draw=deepcerulean,fill=deepcerulean,circle,inner sep=.7pt] (1-1-6) at (-2,-2.4) {};
    
        \node[draw=deepcerulean,fill=deepcerulean,circle,inner sep=.7pt] (1-2-1) at (-1.6,.7) {};
        \node[draw=deepcerulean,fill=deepcerulean,circle,inner sep=.7pt] (1-2-2) at (-1.6,.3) {};
        \node[draw=deepcerulean,fill=deepcerulean,circle,inner sep=.7pt] (1-2-3) at (-1.6,-.3) {};
        \node[draw=deepcerulean,fill=deepcerulean,circle,inner sep=.7pt] (1-2-4) at (-1.6,-.7) {};
        
        \node[draw=deepcerulean,fill=deepcerulean,circle,inner sep=.7pt] (1-3-1) at (-1,2.4) {};
        \node[draw=deepcerulean,fill=deepcerulean,circle,inner sep=.7pt] (1-3-2) at (-1,2) {};
        \node[draw=deepcerulean,fill=deepcerulean,circle,inner sep=.7pt] (1-3-3) at (-1,1.4) {};
        \node[draw=deepcerulean,fill=deepcerulean,circle,inner sep=.7pt] (1-3-4) at (-1,1) {};
        \node[draw=deepcerulean,fill=deepcerulean,circle,inner sep=.7pt] (1-3-5) at (-1,.7) {};
        \node[draw=deepcerulean,fill=deepcerulean,circle,inner sep=.7pt] (1-3-6) at (-1,.3) {};
        \node[draw=deepcerulean,fill=deepcerulean,circle,inner sep=.7pt] (1-3-7) at (-1,0) {};
        \node[draw=deepcerulean,fill=deepcerulean,circle,inner sep=.7pt] (1-3-8) at (-1,-.3) {};
        \node[draw=deepcerulean,fill=deepcerulean,circle,inner sep=.7pt] (1-3-9) at (-1,-.7) {};
        \node[draw=deepcerulean,fill=deepcerulean,circle,inner sep=.7pt] (1-3-10) at (-1,-1) {};
        \node[draw=deepcerulean,fill=deepcerulean,circle,inner sep=.7pt] (1-3-11) at (-1,-1.4) {};
        \node[draw=deepcerulean,fill=deepcerulean,circle,inner sep=.7pt] (1-3-12) at (-1,-2) {};
        \node[draw=deepcerulean,fill=deepcerulean,circle,inner sep=.7pt] (1-3-13) at (-1,-2.4) {};
    
        \node[draw=deepcerulean,fill=deepcerulean,circle,inner sep=.7pt] (1-4-1) at (-.3,.3) {};
        \node[draw=deepcerulean,fill=deepcerulean,circle,inner sep=.7pt] (1-4-2) at (-.3,0) {};
        \node[draw=deepcerulean,fill=deepcerulean,circle,inner sep=.7pt] (1-4-3) at (-.3,-.3) {};
    
        \node[draw=deepcerulean,fill=deepcerulean,circle,inner sep=.7pt] (1-5-1) at (0,1.4) {};
        \node[draw=deepcerulean,fill=deepcerulean,circle,inner sep=.7pt] (1-5-2) at (0,1) {};
        \node[draw=deepcerulean,fill=deepcerulean,circle,inner sep=.7pt] (1-5-3) at (0,-.4) {};
        \node[draw=deepcerulean,fill=deepcerulean,circle,inner sep=.7pt] (1-5-4) at (0,-1) {};
        \node[draw=deepcerulean,fill=deepcerulean,circle,inner sep=.7pt] (1-5-5) at (0,-1.4) {};
    
        \node[draw=deepcerulean,fill=deepcerulean,circle,inner sep=.7pt] (1-6-1) at (1,2.4) {};
        \node[draw=deepcerulean,fill=deepcerulean,circle,inner sep=.7pt] (1-6-2) at (1,2) {};
        \node[draw=deepcerulean,fill=deepcerulean,circle,inner sep=.7pt] (1-6-3) at (1,1.4) {};
        \node[draw=deepcerulean,fill=deepcerulean,circle,inner sep=.7pt] (1-6-4) at (1,0) {};
        \node[draw=deepcerulean,fill=deepcerulean,circle,inner sep=.7pt] (1-6-5) at (1,-.4) {};
        \node[draw=deepcerulean,fill=deepcerulean,circle,inner sep=.7pt] (1-6-6) at (1,-1.4) {};
    
        \node[draw=deepcerulean,fill=deepcerulean,circle,inner sep=.7pt] (1-7-1) at (1.4,-2) {};
        \node[draw=deepcerulean,fill=deepcerulean,circle,inner sep=.7pt] (1-7-2) at (1.4,-2.4) {};
        
        \node[draw=deepcerulean,fill=deepcerulean,circle,inner sep=.7pt] (1-8-1) at (2,-1) {};
    
        \node[draw=deepcerulean,fill=deepcerulean,circle,inner sep=.7pt] (1-9-1) at (2.75,0) {};
        \node[draw=deepcerulean,fill=deepcerulean,circle,inner sep=.7pt] (1-9-2) at (2.75,-1) {};
        \node[draw=deepcerulean,fill=deepcerulean,circle,inner sep=.7pt] (1-9-3) at (2.75,-2) {};
    
        \node[draw=deepcerulean,fill=deepcerulean,circle,inner sep=.7pt] (2-1-1) at (2.2,-1) {};
        \node[draw=deepcerulean,fill=deepcerulean,circle,inner sep=.7pt] (2-1-2) at (2.2,-2) {};
    
        \node[draw=deepcerulean,fill=deepcerulean,circle,inner sep=.7pt] (2-2-1) at (2.4,0) {};
        \node[draw=deepcerulean,fill=deepcerulean,circle,inner sep=.7pt] (2-2-2) at (2.4,-1) {};
    
        \node[draw=deepcerulean,fill=deepcerulean,circle,inner sep=.7pt] (3-1-1) at (-2,2) {};
        \node[draw=deepcerulean,fill=deepcerulean,circle,inner sep=.7pt] (3-1-2) at (-2,1.4) {};
        \node[draw=deepcerulean,fill=deepcerulean,circle,inner sep=.7pt] (3-1-3) at (-2,.3) {};
        \node[draw=deepcerulean,fill=deepcerulean,circle,inner sep=.7pt] (3-1-4) at (-2,-.3) {};
        \node[draw=deepcerulean,fill=deepcerulean,circle,inner sep=.7pt] (3-1-5) at (-2,-1.4) {};
        \node[draw=deepcerulean,fill=deepcerulean,circle,inner sep=.7pt] (3-1-6) at (-2,-2) {};
    
        \node[draw=deepcerulean,fill=deepcerulean,circle,inner sep=.7pt] (3-2-1) at (0,0) {};
    
        \node[draw=deepcerulean,fill=deepcerulean,circle,inner sep=.7pt] (3-3-1) at (1,1) {};
    
        \node[draw=deepcerulean,fill=deepcerulean,circle,inner sep=.7pt] (3-4-1) at (2.4,-2) {};

            \node[draw=tenn,fill=tenn,circle,inner sep=1pt] at (1.25,2) {};
            \node[draw=tenn,fill=tenn,circle,inner sep=1pt] at (2.5,0.25) {};

    \end{tikzpicture}
    }
    \resizebox{!}{0.42\textwidth}{%
            \begin{tikzpicture}
        \fill [pattern=north east lines, pattern color=gray] (-3, -2.65) -- (-3, 2.65) -- (3,2.65) -- (3,-2.65) -- cycle;
        \draw [black, fill=white] (-2.75, -2.4) -- (-2.75, 2.4) -- (2.75,2.4) -- (2.75,-2.4) -- cycle;
    
        \fill [draw=black, pattern=north east lines, pattern color=gray] (-2, -1) -- (-2,-0.7) -- (-0.3,-0.7) -- (-0.3,0.7) -- (-2,0.7) -- (-2, 1) -- (0,1) -- (0,-1) -- cycle;
    
        \fill [draw=black, pattern=north east lines, pattern color=gray] (-1.6,-0.3) -- (-1,-0.3) -- (-1,0.3) -- (-1.6,0.3) -- cycle;
    
        \fill [draw=black, pattern=north east lines, pattern color=gray] (-1,1.4) -- (-1,2) -- (1,2) -- (1, 1.4) -- cycle;
    
        \fill [draw=black, pattern=north east lines, pattern color=gray] (-1,-1.4) -- (-1,-2) -- (2,-2) -- (2,-1.7) -- (2.2,-1.7) -- (2.2,-1.3) -- (2,-1.3) -- (2,-.7) -- (2.4,-.7) -- (2.4,-.3) -- (2,-.3) -- (2, 0) -- (1.4,0) -- (1.4,-1.4) -- (.8,-1.4) -- (.8,-.4) -- (.4,-.4) -- (.4,-1.4) -- cycle;
    
    
    
        \draw [starship,thick] (-2,-2.4) -- (-2,2.4); 
        \draw [starship,thick] (-1.6,-0.7) -- (-1.6,0.7);
        \draw [starship,thick] (-1,1) -- (-1,2.4); 
        \draw [starship,thick] (-1,-2.4) -- (-1,-1);
    
    
        \draw [starship,thick] (-2,1) -- (-2,-2.4);
        \draw [starship,thick] (0,-1.4) -- (0,1.4);
        \draw [starship,thick] (1,-1.4) -- (1,2.4);
    
        \draw [starship,thick] (-1,-0.7) -- (-1,0.7);
        \draw [starship,thick] (1.4,-2.4) -- (1.4,-2);
    
        \draw [starship,thick] (1,0) -- (2.75,0);
        \draw [starship,thick] (-1,0) -- (-0.3,0);
        \draw [starship,thick] (2,-1) -- (2.75,-1);
    
    
        \draw [starship,thick] (0,-.4) -- (1,-.4);
    
        \draw [starship,thick] (-1.6,.3) -- (-.3,.3);
        \draw [starship,thick] (-1.6,-.3) -- (-.3,-.3);
        \draw [starship,thick] (1,0) -- (2.75,0);
        \draw [starship,thick] (-1,-2) -- (2.75,-2);
    
    
        \node[draw=deepcerulean,fill=deepcerulean,circle,inner sep=.7pt] (1-1-1) at (-2,2.4) {};
        \node[draw=deepcerulean,fill=deepcerulean,circle,inner sep=.7pt] (1-1-2) at (-2,1) {};
        \node[draw=deepcerulean,fill=deepcerulean,circle,inner sep=.7pt] (1-1-3) at (-2,.7) {};
        \node[draw=deepcerulean,fill=deepcerulean,circle,inner sep=.7pt] (1-1-4) at (-2,-.7) {};
        \node[draw=deepcerulean,fill=deepcerulean,circle,inner sep=.7pt] (1-1-5) at (-2,-1) {};
        \node[draw=deepcerulean,fill=deepcerulean,circle,inner sep=.7pt] (1-1-6) at (-2,-2.4) {};
    
        \node[draw=deepcerulean,fill=deepcerulean,circle,inner sep=.7pt] (1-2-1) at (-1.6,.7) {};
        \node[draw=deepcerulean,fill=deepcerulean,circle,inner sep=.7pt] (1-2-2) at (-1.6,.3) {};
        \node[draw=deepcerulean,fill=deepcerulean,circle,inner sep=.7pt] (1-2-3) at (-1.6,-.3) {};
        \node[draw=deepcerulean,fill=deepcerulean,circle,inner sep=.7pt] (1-2-4) at (-1.6,-.7) {};
        
        \node[draw=deepcerulean,fill=deepcerulean,circle,inner sep=.7pt] (1-3-1) at (-1,2.4) {};
        \node[draw=deepcerulean,fill=deepcerulean,circle,inner sep=.7pt] (1-3-2) at (-1,2) {};
        \node[draw=deepcerulean,fill=deepcerulean,circle,inner sep=.7pt] (1-3-3) at (-1,1.4) {};
        \node[draw=deepcerulean,fill=deepcerulean,circle,inner sep=.7pt] (1-3-4) at (-1,1) {};
        \node[draw=deepcerulean,fill=deepcerulean,circle,inner sep=.7pt] (1-3-5) at (-1,.7) {};
        \node[draw=deepcerulean,fill=deepcerulean,circle,inner sep=.7pt] (1-3-6) at (-1,.3) {};
        \node[draw=deepcerulean,fill=deepcerulean,circle,inner sep=.7pt] (1-3-7) at (-1,0) {};
        \node[draw=deepcerulean,fill=deepcerulean,circle,inner sep=.7pt] (1-3-8) at (-1,-.3) {};
        \node[draw=deepcerulean,fill=deepcerulean,circle,inner sep=.7pt] (1-3-9) at (-1,-.7) {};
        \node[draw=deepcerulean,fill=deepcerulean,circle,inner sep=.7pt] (1-3-10) at (-1,-1) {};
        \node[draw=deepcerulean,fill=deepcerulean,circle,inner sep=.7pt] (1-3-11) at (-1,-1.4) {};
        \node[draw=deepcerulean,fill=deepcerulean,circle,inner sep=.7pt] (1-3-12) at (-1,-2) {};
        \node[draw=deepcerulean,fill=deepcerulean,circle,inner sep=.7pt] (1-3-13) at (-1,-2.4) {};
    blue
        \node[draw=deepcerulean,fill=deepcerulean,circle,inner sep=.7pt] (1-4-1) at (-.3,.3) {};
        \node[draw=deepcerulean,fill=deepcerulean,circle,inner sep=.7pt] (1-4-2) at (-.3,0) {};
        \node[draw=deepcerulean,fill=deepcerulean,circle,inner sep=.7pt] (1-4-3) at (-.3,-.3) {};
    
        \node[draw=deepcerulean,fill=deepcerulean,circle,inner sep=.7pt] (1-5-1) at (0,1.4) {};
        \node[draw=deepcerulean,fill=deepcerulean,circle,inner sep=.7pt] (1-5-2) at (0,1) {};
        \node[draw=deepcerulean,fill=deepcerulean,circle,inner sep=.7pt] (1-5-3) at (0,-.4) {};
        \node[draw=deepcerulean,fill=deepcerulean,circle,inner sep=.7pt] (1-5-4) at (0,-1) {};
        \node[draw=deepcerulean,fill=deepcerulean,circle,inner sep=.7pt] (1-5-5) at (0,-1.4) {};
    
        \node[draw=deepcerulean,fill=deepcerulean,circle,inner sep=.7pt] (1-6-1) at (1,2.4) {};
        \node[draw=deepcerulean,fill=deepcerulean,circle,inner sep=.7pt] (1-6-2) at (1,2) {};
        \node[draw=deepcerulean,fill=deepcerulean,circle,inner sep=.7pt] (1-6-3) at (1,1.4) {};
        \node[draw=deepcerulean,fill=deepcerulean,circle,inner sep=.7pt] (1-6-4) at (1,0) {};
        \node[draw=deepcerulean,fill=deepcerulean,circle,inner sep=.7pt] (1-6-5) at (1,-.4) {};
        \node[draw=deepcerulean,fill=deepcerulean,circle,inner sep=.7pt] (1-6-6) at (1,-1.4) {};
    
        \node[draw=deepcerulean,fill=deepcerulean,circle,inner sep=.7pt] (1-7-1) at (1.4,-2) {};
        \node[draw=deepcerulean,fill=deepcerulean,circle,inner sep=.7pt] (1-7-2) at (1.4,-2.4) {};
        
        \node[draw=deepcerulean,fill=deepcerulean,circle,inner sep=.7pt] (1-8-1) at (2,-1) {};
    
        \node[draw=deepcerulean,fill=deepcerulean,circle,inner sep=.7pt] (1-9-1) at (2.75,0) {};
        \node[draw=deepcerulean,fill=deepcerulean,circle,inner sep=.7pt] (1-9-2) at (2.75,-1) {};
        \node[draw=deepcerulean,fill=deepcerulean,circle,inner sep=.7pt] (1-9-3) at (2.75,-2) {};
    
        \node[draw=deepcerulean,fill=deepcerulean,circle,inner sep=.7pt] (2-1-1) at (2.2,-1) {};
        \node[draw=deepcerulean,fill=deepcerulean,circle,inner sep=.7pt] (2-1-2) at (2.2,-2) {};
    
        \node[draw=deepcerulean,fill=deepcerulean,circle,inner sep=.7pt] (2-2-1) at (2.4,0) {};
        \node[draw=deepcerulean,fill=deepcerulean,circle,inner sep=.7pt] (2-2-2) at (2.4,-1) {};
    
        \node[draw=deepcerulean,fill=deepcerulean,circle,inner sep=.7pt] (3-1-1) at (-2,2) {};
        \node[draw=deepcerulean,fill=deepcerulean,circle,inner sep=.7pt] (3-1-2) at (-2,1.4) {};
        \node[draw=deepcerulean,fill=deepcerulean,circle,inner sep=.7pt] (3-1-3) at (-2,.3) {};
        \node[draw=deepcerulean,fill=deepcerulean,circle,inner sep=.7pt] (3-1-4) at (-2,-.3) {};
        \node[draw=deepcerulean,fill=deepcerulean,circle,inner sep=.7pt] (3-1-5) at (-2,-1.4) {};
        \node[draw=deepcerulean,fill=deepcerulean,circle,inner sep=.7pt] (3-1-6) at (-2,-2) {};
    
        \node[draw=deepcerulean,fill=deepcerulean,circle,inner sep=.7pt] (3-2-1) at (0,0) {};
    
        \node[draw=deepcerulean,fill=deepcerulean,circle,inner sep=.7pt] (3-3-1) at (1,1) {};
    
        \node[draw=deepcerulean,fill=deepcerulean,circle,inner sep=.7pt] (3-4-1) at (2.4,-2) {};

            \node[draw=tenn,fill=tenn,circle,inner sep=1pt] at (1.25,2) {};
            \node[draw=tenn,fill=tenn,circle,inner sep=1pt] at (2.5,0.25) {};
            \draw [tenn,thick] (1.25,2) -- (2.5,2) -- (2.5,0.25);
    \end{tikzpicture}
    }
    \caption{When source and target of a routing request are in the same region (left), we employ the routing scheme from \cite{CCFHKSSS21} (right).}
    \label{fig:gallery_routingeasy}
\end{figure}

\autoref{fig:gallery_routingdifficult} presents the more difficult case, where the source and target node of a routing request are in different regions. Here, we have the source node locally augment the landmark graph by adding itself and the target node to it. Then, the source node locally solves the shortest path problem on the landmark graph, to decide where the packet should exit its region. Then, it forwards the packet one step towards the corresponding region boundary. After reevaluating the same computation, the receiving node forwards the packet again until it is passed to the next region, where it is handed over again. This is repeated until the packet arrives in the target node's region, where we can fall back to the easy case presented above.
\begin{figure}
    \centering
    \resizebox{!}{0.42\textwidth}{%
            \begin{tikzpicture}
        \fill [pattern=north east lines, pattern color=gray] (-3, -2.65) -- (-3, 2.65) -- (3,2.65) -- (3,-2.65) -- cycle;
        \draw [black, fill=white] (-2.75, -2.4) -- (-2.75, 2.4) -- (2.75,2.4) -- (2.75,-2.4) -- cycle;
    
        \fill [draw=black, pattern=north east lines, pattern color=gray] (-2, -1) -- (-2,-0.7) -- (-0.3,-0.7) -- (-0.3,0.7) -- (-2,0.7) -- (-2, 1) -- (0,1) -- (0,-1) -- cycle;
    
        \fill [draw=black, pattern=north east lines, pattern color=gray] (-1.6,-0.3) -- (-1,-0.3) -- (-1,0.3) -- (-1.6,0.3) -- cycle;
    
        \fill [draw=black, pattern=north east lines, pattern color=gray] (-1,1.4) -- (-1,2) -- (1,2) -- (1, 1.4) -- cycle;
    
        \fill [draw=black, pattern=north east lines, pattern color=gray] (-1,-1.4) -- (-1,-2) -- (2,-2) -- (2,-1.7) -- (2.2,-1.7) -- (2.2,-1.3) -- (2,-1.3) -- (2,-.7) -- (2.4,-.7) -- (2.4,-.3) -- (2,-.3) -- (2, 0) -- (1.4,0) -- (1.4,-1.4) -- (.8,-1.4) -- (.8,-.4) -- (.4,-.4) -- (.4,-1.4) -- cycle;
    
    
    
        \draw [starship,thick] (-2,-2.4) -- (-2,2.4); 
        \draw [starship,thick] (-1.6,-0.7) -- (-1.6,0.7);
        \draw [starship,thick] (-1,1) -- (-1,2.4); 
        \draw [starship,thick] (-1,-2.4) -- (-1,-1);
    
    
        \draw [starship,thick] (-2,1) -- (-2,-2.4);
        \draw [starship,thick] (0,-1.4) -- (0,1.4);
        \draw [starship,thick] (1,-1.4) -- (1,2.4);
    
        \draw [starship,thick] (-1,-0.7) -- (-1,0.7);
        \draw [starship,thick] (1.4,-2.4) -- (1.4,-2);
    
        \draw [starship,thick] (1,0) -- (2.75,0);
        \draw [starship,thick] (-1,0) -- (-0.3,0);
        \draw [starship,thick] (2,-1) -- (2.75,-1);
    
    
        \draw [starship,thick] (0,-.4) -- (1,-.4);
    
        \draw [starship,thick] (-1.6,.3) -- (-.3,.3);
        \draw [starship,thick] (-1.6,-.3) -- (-.3,-.3);
        \draw [starship,thick] (1,0) -- (2.75,0);
        \draw [starship,thick] (-1,-2) -- (2.75,-2);
    
    
        \node[draw=deepcerulean,fill=deepcerulean,circle,inner sep=.7pt] (1-1-1) at (-2,2.4) {};
        \node[draw=deepcerulean,fill=deepcerulean,circle,inner sep=.7pt] (1-1-2) at (-2,1) {};
        \node[draw=deepcerulean,fill=deepcerulean,circle,inner sep=.7pt] (1-1-3) at (-2,.7) {};
        \node[draw=deepcerulean,fill=deepcerulean,circle,inner sep=.7pt] (1-1-4) at (-2,-.7) {};
        \node[draw=deepcerulean,fill=deepcerulean,circle,inner sep=.7pt] (1-1-5) at (-2,-1) {};
        \node[draw=deepcerulean,fill=deepcerulean,circle,inner sep=.7pt] (1-1-6) at (-2,-2.4) {};
    
        \node[draw=deepcerulean,fill=deepcerulean,circle,inner sep=.7pt] (1-2-1) at (-1.6,.7) {};
        \node[draw=deepcerulean,fill=deepcerulean,circle,inner sep=.7pt] (1-2-2) at (-1.6,.3) {};
        \node[draw=deepcerulean,fill=deepcerulean,circle,inner sep=.7pt] (1-2-3) at (-1.6,-.3) {};
        \node[draw=deepcerulean,fill=deepcerulean,circle,inner sep=.7pt] (1-2-4) at (-1.6,-.7) {};
        
        \node[draw=deepcerulean,fill=deepcerulean,circle,inner sep=.7pt] (1-3-1) at (-1,2.4) {};
        \node[draw=deepcerulean,fill=deepcerulean,circle,inner sep=.7pt] (1-3-2) at (-1,2) {};
        \node[draw=deepcerulean,fill=deepcerulean,circle,inner sep=.7pt] (1-3-3) at (-1,1.4) {};
        \node[draw=deepcerulean,fill=deepcerulean,circle,inner sep=.7pt] (1-3-4) at (-1,1) {};
        \node[draw=deepcerulean,fill=deepcerulean,circle,inner sep=.7pt] (1-3-5) at (-1,.7) {};
        \node[draw=deepcerulean,fill=deepcerulean,circle,inner sep=.7pt] (1-3-6) at (-1,.3) {};
        \node[draw=deepcerulean,fill=deepcerulean,circle,inner sep=.7pt] (1-3-7) at (-1,0) {};
        \node[draw=deepcerulean,fill=deepcerulean,circle,inner sep=.7pt] (1-3-8) at (-1,-.3) {};
        \node[draw=deepcerulean,fill=deepcerulean,circle,inner sep=.7pt] (1-3-9) at (-1,-.7) {};
        \node[draw=deepcerulean,fill=deepcerulean,circle,inner sep=.7pt] (1-3-10) at (-1,-1) {};
        \node[draw=deepcerulean,fill=deepcerulean,circle,inner sep=.7pt] (1-3-11) at (-1,-1.4) {};
        \node[draw=deepcerulean,fill=deepcerulean,circle,inner sep=.7pt] (1-3-12) at (-1,-2) {};
        \node[draw=deepcerulean,fill=deepcerulean,circle,inner sep=.7pt] (1-3-13) at (-1,-2.4) {};
    
        \node[draw=deepcerulean,fill=deepcerulean,circle,inner sep=.7pt] (1-4-1) at (-.3,.3) {};
        \node[draw=deepcerulean,fill=deepcerulean,circle,inner sep=.7pt] (1-4-2) at (-.3,0) {};
        \node[draw=deepcerulean,fill=deepcerulean,circle,inner sep=.7pt] (1-4-3) at (-.3,-.3) {};
    
        \node[draw=deepcerulean,fill=deepcerulean,circle,inner sep=.7pt] (1-5-1) at (0,1.4) {};
        \node[draw=deepcerulean,fill=deepcerulean,circle,inner sep=.7pt] (1-5-2) at (0,1) {};
        \node[draw=deepcerulean,fill=deepcerulean,circle,inner sep=.7pt] (1-5-3) at (0,-.4) {};
        \node[draw=deepcerulean,fill=deepcerulean,circle,inner sep=.7pt] (1-5-4) at (0,-1) {};
        \node[draw=deepcerulean,fill=deepcerulean,circle,inner sep=.7pt] (1-5-5) at (0,-1.4) {};
    
        \node[draw=deepcerulean,fill=deepcerulean,circle,inner sep=.7pt] (1-6-1) at (1,2.4) {};
        \node[draw=deepcerulean,fill=deepcerulean,circle,inner sep=.7pt] (1-6-2) at (1,2) {};
        \node[draw=deepcerulean,fill=deepcerulean,circle,inner sep=.7pt] (1-6-3) at (1,1.4) {};
        \node[draw=deepcerulean,fill=deepcerulean,circle,inner sep=.7pt] (1-6-4) at (1,0) {};
        \node[draw=deepcerulean,fill=deepcerulean,circle,inner sep=.7pt] (1-6-5) at (1,-.4) {};
        \node[draw=deepcerulean,fill=deepcerulean,circle,inner sep=.7pt] (1-6-6) at (1,-1.4) {};
    
        \node[draw=deepcerulean,fill=deepcerulean,circle,inner sep=.7pt] (1-7-1) at (1.4,-2) {};
        \node[draw=deepcerulean,fill=deepcerulean,circle,inner sep=.7pt] (1-7-2) at (1.4,-2.4) {};
        
        \node[draw=deepcerulean,fill=deepcerulean,circle,inner sep=.7pt] (1-8-1) at (2,-1) {};
    
        \node[draw=deepcerulean,fill=deepcerulean,circle,inner sep=.7pt] (1-9-1) at (2.75,0) {};
        \node[draw=deepcerulean,fill=deepcerulean,circle,inner sep=.7pt] (1-9-2) at (2.75,-1) {};
        \node[draw=deepcerulean,fill=deepcerulean,circle,inner sep=.7pt] (1-9-3) at (2.75,-2) {};
    
        \node[draw=deepcerulean,fill=deepcerulean,circle,inner sep=.7pt] (2-1-1) at (2.2,-1) {};
        \node[draw=deepcerulean,fill=deepcerulean,circle,inner sep=.7pt] (2-1-2) at (2.2,-2) {};
    
        \node[draw=deepcerulean,fill=deepcerulean,circle,inner sep=.7pt] (2-2-1) at (2.4,0) {};
        \node[draw=deepcerulean,fill=deepcerulean,circle,inner sep=.7pt] (2-2-2) at (2.4,-1) {};
    
        \node[draw=deepcerulean,fill=deepcerulean,circle,inner sep=.7pt] (3-1-1) at (-2,2) {};
        \node[draw=deepcerulean,fill=deepcerulean,circle,inner sep=.7pt] (3-1-2) at (-2,1.4) {};
        \node[draw=deepcerulean,fill=deepcerulean,circle,inner sep=.7pt] (3-1-3) at (-2,.3) {};
        \node[draw=deepcerulean,fill=deepcerulean,circle,inner sep=.7pt] (3-1-4) at (-2,-.3) {};
        \node[draw=deepcerulean,fill=deepcerulean,circle,inner sep=.7pt] (3-1-5) at (-2,-1.4) {};
        \node[draw=deepcerulean,fill=deepcerulean,circle,inner sep=.7pt] (3-1-6) at (-2,-2) {};
    
        \node[draw=deepcerulean,fill=deepcerulean,circle,inner sep=.7pt] (3-2-1) at (0,0) {};
    
        \node[draw=deepcerulean,fill=deepcerulean,circle,inner sep=.7pt] (3-3-1) at (1,1) {};
    
        \node[draw=deepcerulean,fill=deepcerulean,circle,inner sep=.7pt] (3-4-1) at (2.4,-2) {};

            \node[draw=tenn,fill=tenn,circle,inner sep=1pt] at (-2.25,0.5) {};
            \node[draw=tenn,fill=tenn,circle,inner sep=1pt] at (2.5,0.25) {};
    \end{tikzpicture}
    }
    \resizebox{!}{0.42\textwidth}{%
            \begin{tikzpicture}
        \fill [pattern=north east lines, pattern color=gray] (-3, -2.65) -- (-3, 2.65) -- (3,2.65) -- (3,-2.65) -- cycle;
        \draw [black, fill=white] (-2.75, -2.4) -- (-2.75, 2.4) -- (2.75,2.4) -- (2.75,-2.4) -- cycle;
    
        \fill [draw=black, pattern=north east lines, pattern color=gray] (-2, -1) -- (-2,-0.7) -- (-0.3,-0.7) -- (-0.3,0.7) -- (-2,0.7) -- (-2, 1) -- (0,1) -- (0,-1) -- cycle;
    
        \fill [draw=black, pattern=north east lines, pattern color=gray] (-1.6,-0.3) -- (-1,-0.3) -- (-1,0.3) -- (-1.6,0.3) -- cycle;
    
        \fill [draw=black, pattern=north east lines, pattern color=gray] (-1,1.4) -- (-1,2) -- (1,2) -- (1, 1.4) -- cycle;
    
        \fill [draw=black, pattern=north east lines, pattern color=gray] (-1,-1.4) -- (-1,-2) -- (2,-2) -- (2,-1.7) -- (2.2,-1.7) -- (2.2,-1.3) -- (2,-1.3) -- (2,-.7) -- (2.4,-.7) -- (2.4,-.3) -- (2,-.3) -- (2, 0) -- (1.4,0) -- (1.4,-1.4) -- (.8,-1.4) -- (.8,-.4) -- (.4,-.4) -- (.4,-1.4) -- cycle;
    
    
    
        \draw [starship,thick] (-2,-2.4) -- (-2,2.4); 
        \draw [starship,thick] (-1.6,-0.7) -- (-1.6,0.7);
        \draw [starship,thick] (-1,1) -- (-1,2.4); 
        \draw [starship,thick] (-1,-2.4) -- (-1,-1);
    
    
        \draw [starship,thick] (-2,1) -- (-2,-2.4);
        \draw [starship,thick] (0,-1.4) -- (0,1.4);
        \draw [starship,thick] (1,-1.4) -- (1,2.4);
    
        \draw [starship,thick] (-1,-0.7) -- (-1,0.7);
        \draw [starship,thick] (1.4,-2.4) -- (1.4,-2);
    
        \draw [starship,thick] (1,0) -- (2.75,0);
        \draw [starship,thick] (-1,0) -- (-0.3,0);
        \draw [starship,thick] (2,-1) -- (2.75,-1);
    
    
        \draw [starship,thick] (0,-.4) -- (1,-.4);
    
        \draw [starship,thick] (-1.6,.3) -- (-.3,.3);
        \draw [starship,thick] (-1.6,-.3) -- (-.3,-.3);
        \draw [starship,thick] (1,0) -- (2.75,0);
        \draw [starship,thick] (-1,-2) -- (2.75,-2);
    
    
        \node[draw=deepcerulean,fill=deepcerulean,circle,inner sep=.7pt] (1-1-1) at (-2,2.4) {};
        \node[draw=deepcerulean,fill=deepcerulean,circle,inner sep=.7pt] (1-1-2) at (-2,1) {};
        \node[draw=deepcerulean,fill=deepcerulean,circle,inner sep=.7pt] (1-1-3) at (-2,.7) {};
        \node[draw=deepcerulean,fill=deepcerulean,circle,inner sep=.7pt] (1-1-4) at (-2,-.7) {};
        \node[draw=deepcerulean,fill=deepcerulean,circle,inner sep=.7pt] (1-1-5) at (-2,-1) {};
        \node[draw=deepcerulean,fill=deepcerulean,circle,inner sep=.7pt] (1-1-6) at (-2,-2.4) {};
    
        \node[draw=deepcerulean,fill=deepcerulean,circle,inner sep=.7pt] (1-2-1) at (-1.6,.7) {};
        \node[draw=deepcerulean,fill=deepcerulean,circle,inner sep=.7pt] (1-2-2) at (-1.6,.3) {};
        \node[draw=deepcerulean,fill=deepcerulean,circle,inner sep=.7pt] (1-2-3) at (-1.6,-.3) {};
        \node[draw=deepcerulean,fill=deepcerulean,circle,inner sep=.7pt] (1-2-4) at (-1.6,-.7) {};
        
        \node[draw=deepcerulean,fill=deepcerulean,circle,inner sep=.7pt] (1-3-1) at (-1,2.4) {};
        \node[draw=deepcerulean,fill=deepcerulean,circle,inner sep=.7pt] (1-3-2) at (-1,2) {};
        \node[draw=deepcerulean,fill=deepcerulean,circle,inner sep=.7pt] (1-3-3) at (-1,1.4) {};
        \node[draw=deepcerulean,fill=deepcerulean,circle,inner sep=.7pt] (1-3-4) at (-1,1) {};
        \node[draw=deepcerulean,fill=deepcerulean,circle,inner sep=.7pt] (1-3-5) at (-1,.7) {};
        \node[draw=deepcerulean,fill=deepcerulean,circle,inner sep=.7pt] (1-3-6) at (-1,.3) {};
        \node[draw=deepcerulean,fill=deepcerulean,circle,inner sep=.7pt] (1-3-7) at (-1,0) {};
        \node[draw=deepcerulean,fill=deepcerulean,circle,inner sep=.7pt] (1-3-8) at (-1,-.3) {};
        \node[draw=deepcerulean,fill=deepcerulean,circle,inner sep=.7pt] (1-3-9) at (-1,-.7) {};
        \node[draw=deepcerulean,fill=deepcerulean,circle,inner sep=.7pt] (1-3-10) at (-1,-1) {};
        \node[draw=deepcerulean,fill=deepcerulean,circle,inner sep=.7pt] (1-3-11) at (-1,-1.4) {};
        \node[draw=deepcerulean,fill=deepcerulean,circle,inner sep=.7pt] (1-3-12) at (-1,-2) {};
        \node[draw=deepcerulean,fill=deepcerulean,circle,inner sep=.7pt] (1-3-13) at (-1,-2.4) {};
    
        \node[draw=deepcerulean,fill=deepcerulean,circle,inner sep=.7pt] (1-4-1) at (-.3,.3) {};
        \node[draw=deepcerulean,fill=deepcerulean,circle,inner sep=.7pt] (1-4-2) at (-.3,0) {};
        \node[draw=deepcerulean,fill=deepcerulean,circle,inner sep=.7pt] (1-4-3) at (-.3,-.3) {};
    
        \node[draw=deepcerulean,fill=deepcerulean,circle,inner sep=.7pt] (1-5-1) at (0,1.4) {};
        \node[draw=deepcerulean,fill=deepcerulean,circle,inner sep=.7pt] (1-5-2) at (0,1) {};
        \node[draw=deepcerulean,fill=deepcerulean,circle,inner sep=.7pt] (1-5-3) at (0,-.4) {};
        \node[draw=deepcerulean,fill=deepcerulean,circle,inner sep=.7pt] (1-5-4) at (0,-1) {};
        \node[draw=deepcerulean,fill=deepcerulean,circle,inner sep=.7pt] (1-5-5) at (0,-1.4) {};
    
        \node[draw=deepcerulean,fill=deepcerulean,circle,inner sep=.7pt] (1-6-1) at (1,2.4) {};
        \node[draw=deepcerulean,fill=deepcerulean,circle,inner sep=.7pt] (1-6-2) at (1,2) {};
        \node[draw=deepcerulean,fill=deepcerulean,circle,inner sep=.7pt] (1-6-3) at (1,1.4) {};
        \node[draw=deepcerulean,fill=deepcerulean,circle,inner sep=.7pt] (1-6-4) at (1,0) {};
        \node[draw=deepcerulean,fill=deepcerulean,circle,inner sep=.7pt] (1-6-5) at (1,-.4) {};
        \node[draw=deepcerulean,fill=deepcerulean,circle,inner sep=.7pt] (1-6-6) at (1,-1.4) {};
    
        \node[draw=deepcerulean,fill=deepcerulean,circle,inner sep=.7pt] (1-7-1) at (1.4,-2) {};
        \node[draw=deepcerulean,fill=deepcerulean,circle,inner sep=.7pt] (1-7-2) at (1.4,-2.4) {};
        
        \node[draw=deepcerulean,fill=deepcerulean,circle,inner sep=.7pt] (1-8-1) at (2,-1) {};
    
        \node[draw=deepcerulean,fill=deepcerulean,circle,inner sep=.7pt] (1-9-1) at (2.75,0) {};
        \node[draw=deepcerulean,fill=deepcerulean,circle,inner sep=.7pt] (1-9-2) at (2.75,-1) {};
        \node[draw=deepcerulean,fill=deepcerulean,circle,inner sep=.7pt] (1-9-3) at (2.75,-2) {};
    
        \node[draw=deepcerulean,fill=deepcerulean,circle,inner sep=.7pt] (2-1-1) at (2.2,-1) {};
        \node[draw=deepcerulean,fill=deepcerulean,circle,inner sep=.7pt] (2-1-2) at (2.2,-2) {};
    
        \node[draw=deepcerulean,fill=deepcerulean,circle,inner sep=.7pt] (2-2-1) at (2.4,0) {};
        \node[draw=deepcerulean,fill=deepcerulean,circle,inner sep=.7pt] (2-2-2) at (2.4,-1) {};
    
        \node[draw=deepcerulean,fill=deepcerulean,circle,inner sep=.7pt] (3-1-1) at (-2,2) {};
        \node[draw=deepcerulean,fill=deepcerulean,circle,inner sep=.7pt] (3-1-2) at (-2,1.4) {};
        \node[draw=deepcerulean,fill=deepcerulean,circle,inner sep=.7pt] (3-1-3) at (-2,.3) {};
        \node[draw=deepcerulean,fill=deepcerulean,circle,inner sep=.7pt] (3-1-4) at (-2,-.3) {};
        \node[draw=deepcerulean,fill=deepcerulean,circle,inner sep=.7pt] (3-1-5) at (-2,-1.4) {};
        \node[draw=deepcerulean,fill=deepcerulean,circle,inner sep=.7pt] (3-1-6) at (-2,-2) {};
    
        \node[draw=deepcerulean,fill=deepcerulean,circle,inner sep=.7pt] (3-2-1) at (0,0) {};
    
        \node[draw=deepcerulean,fill=deepcerulean,circle,inner sep=.7pt] (3-3-1) at (1,1) {};
    
        \node[draw=deepcerulean,fill=deepcerulean,circle,inner sep=.7pt] (3-4-1) at (2.4,-2) {};

        \draw[deepcerulean,thick] (1-1-2) to (1-3-4);
        \draw[deepcerulean,thick] (1-3-4) to (1-5-2);
        \draw[deepcerulean,thick] (1-5-2) to (3-3-1);

        \draw[deepcerulean,thick] (-2.25,0.5) to (1-1-2);
        \draw[deepcerulean,thick] (3-3-1) to (2.5,0.25);

        \node[draw=tenn,fill=tenn,circle,inner sep=1pt] at (-2.25,0.5) {};
        \node[draw=tenn,fill=tenn,circle,inner sep=1pt] at (2.5,0.25) {};
    \end{tikzpicture}
    }
    
    \vspace{1em}
    \resizebox{!}{0.42\textwidth}{%
        \begin{tikzpicture}
    \fill [pattern=north east lines, pattern color=gray] (-3, -2.65) -- (-3, 2.65) -- (3,2.65) -- (3,-2.65) -- cycle;
    \draw [black, fill=white] (-2.75, -2.4) -- (-2.75, 2.4) -- (2.75,2.4) -- (2.75,-2.4) -- cycle;

    \fill [draw=black, pattern=north east lines, pattern color=gray] (-2, -1) -- (-2,-0.7) -- (-0.3,-0.7) -- (-0.3,0.7) -- (-2,0.7) -- (-2, 1) -- (0,1) -- (0,-1) -- cycle;

    \fill [draw=black, pattern=north east lines, pattern color=gray] (-1.6,-0.3) -- (-1,-0.3) -- (-1,0.3) -- (-1.6,0.3) -- cycle;

    \fill [draw=black, pattern=north east lines, pattern color=gray] (-1,1.4) -- (-1,2) -- (1,2) -- (1, 1.4) -- cycle;

    \fill [draw=black, pattern=north east lines, pattern color=gray] (-1,-1.4) -- (-1,-2) -- (2,-2) -- (2,-1.7) -- (2.2,-1.7) -- (2.2,-1.3) -- (2,-1.3) -- (2,-.7) -- (2.4,-.7) -- (2.4,-.3) -- (2,-.3) -- (2, 0) -- (1.4,0) -- (1.4,-1.4) -- (.8,-1.4) -- (.8,-.4) -- (.4,-.4) -- (.4,-1.4) -- cycle;



    \draw [starship,thick] (-2,-2.4) -- (-2,2.4); 
    \draw [starship,thick] (-1.6,-0.7) -- (-1.6,0.7);
    \draw [starship,thick] (-1,1) -- (-1,2.4); 
    \draw [starship,thick] (-1,-2.4) -- (-1,-1);


    \draw [starship,thick] (-2,1) -- (-2,-2.4);
    \draw [starship,thick] (0,-1.4) -- (0,1.4);
    \draw [starship,thick] (1,-1.4) -- (1,2.4);

    \draw [starship,thick] (-1,-0.7) -- (-1,0.7);
    \draw [starship,thick] (1.4,-2.4) -- (1.4,-2);

    \draw [starship,thick] (1,0) -- (2.75,0);
    \draw [starship,thick] (-1,0) -- (-0.3,0);
    \draw [starship,thick] (2,-1) -- (2.75,-1);


    \draw [starship,thick] (0,-.4) -- (1,-.4);

    \draw [starship,thick] (-1.6,.3) -- (-.3,.3);
    \draw [starship,thick] (-1.6,-.3) -- (-.3,-.3);
    \draw [starship,thick] (1,0) -- (2.75,0);
    \draw [starship,thick] (-1,-2) -- (2.75,-2);


    \node[draw=deepcerulean,fill=deepcerulean,circle,inner sep=.7pt] (1-1-1) at (-2,2.4) {};
    \node[draw=deepcerulean,fill=deepcerulean,circle,inner sep=.7pt] (1-1-2) at (-2,1) {};
    \node[draw=deepcerulean,fill=deepcerulean,circle,inner sep=.7pt] (1-1-3) at (-2,.7) {};
    \node[draw=deepcerulean,fill=deepcerulean,circle,inner sep=.7pt] (1-1-4) at (-2,-.7) {};
    \node[draw=deepcerulean,fill=deepcerulean,circle,inner sep=.7pt] (1-1-5) at (-2,-1) {};
    \node[draw=deepcerulean,fill=deepcerulean,circle,inner sep=.7pt] (1-1-6) at (-2,-2.4) {};

    \node[draw=deepcerulean,fill=deepcerulean,circle,inner sep=.7pt] (1-2-1) at (-1.6,.7) {};
    \node[draw=deepcerulean,fill=deepcerulean,circle,inner sep=.7pt] (1-2-2) at (-1.6,.3) {};
    \node[draw=deepcerulean,fill=deepcerulean,circle,inner sep=.7pt] (1-2-3) at (-1.6,-.3) {};
    \node[draw=deepcerulean,fill=deepcerulean,circle,inner sep=.7pt] (1-2-4) at (-1.6,-.7) {};
    
    \node[draw=deepcerulean,fill=deepcerulean,circle,inner sep=.7pt] (1-3-1) at (-1,2.4) {};
    \node[draw=deepcerulean,fill=deepcerulean,circle,inner sep=.7pt] (1-3-2) at (-1,2) {};
    \node[draw=deepcerulean,fill=deepcerulean,circle,inner sep=.7pt] (1-3-3) at (-1,1.4) {};
    \node[draw=deepcerulean,fill=deepcerulean,circle,inner sep=.7pt] (1-3-4) at (-1,1) {};
    \node[draw=deepcerulean,fill=deepcerulean,circle,inner sep=.7pt] (1-3-5) at (-1,.7) {};
    \node[draw=deepcerulean,fill=deepcerulean,circle,inner sep=.7pt] (1-3-6) at (-1,.3) {};
    \node[draw=deepcerulean,fill=deepcerulean,circle,inner sep=.7pt] (1-3-7) at (-1,0) {};
    \node[draw=deepcerulean,fill=deepcerulean,circle,inner sep=.7pt] (1-3-8) at (-1,-.3) {};
    \node[draw=deepcerulean,fill=deepcerulean,circle,inner sep=.7pt] (1-3-9) at (-1,-.7) {};
    \node[draw=deepcerulean,fill=deepcerulean,circle,inner sep=.7pt] (1-3-10) at (-1,-1) {};
    \node[draw=deepcerulean,fill=deepcerulean,circle,inner sep=.7pt] (1-3-11) at (-1,-1.4) {};
    \node[draw=deepcerulean,fill=deepcerulean,circle,inner sep=.7pt] (1-3-12) at (-1,-2) {};
    \node[draw=deepcerulean,fill=deepcerulean,circle,inner sep=.7pt] (1-3-13) at (-1,-2.4) {};

    \node[draw=deepcerulean,fill=deepcerulean,circle,inner sep=.7pt] (1-4-1) at (-.3,.3) {};
    \node[draw=deepcerulean,fill=deepcerulean,circle,inner sep=.7pt] (1-4-2) at (-.3,0) {};
    \node[draw=deepcerulean,fill=deepcerulean,circle,inner sep=.7pt] (1-4-3) at (-.3,-.3) {};

    \node[draw=deepcerulean,fill=deepcerulean,circle,inner sep=.7pt] (1-5-1) at (0,1.4) {};
    \node[draw=deepcerulean,fill=deepcerulean,circle,inner sep=.7pt] (1-5-2) at (0,1) {};
    \node[draw=deepcerulean,fill=deepcerulean,circle,inner sep=.7pt] (1-5-3) at (0,-.4) {};
    \node[draw=deepcerulean,fill=deepcerulean,circle,inner sep=.7pt] (1-5-4) at (0,-1) {};
    \node[draw=deepcerulean,fill=deepcerulean,circle,inner sep=.7pt] (1-5-5) at (0,-1.4) {};

    \node[draw=deepcerulean,fill=deepcerulean,circle,inner sep=.7pt] (1-6-1) at (1,2.4) {};
    \node[draw=deepcerulean,fill=deepcerulean,circle,inner sep=.7pt] (1-6-2) at (1,2) {};
    \node[draw=deepcerulean,fill=deepcerulean,circle,inner sep=.7pt] (1-6-3) at (1,1.4) {};
    \node[draw=deepcerulean,fill=deepcerulean,circle,inner sep=.7pt] (1-6-4) at (1,0) {};
    \node[draw=deepcerulean,fill=deepcerulean,circle,inner sep=.7pt] (1-6-5) at (1,-.4) {};
    \node[draw=deepcerulean,fill=deepcerulean,circle,inner sep=.7pt] (1-6-6) at (1,-1.4) {};

    \node[draw=deepcerulean,fill=deepcerulean,circle,inner sep=.7pt] (1-7-1) at (1.4,-2) {};
    \node[draw=deepcerulean,fill=deepcerulean,circle,inner sep=.7pt] (1-7-2) at (1.4,-2.4) {};
    
    \node[draw=deepcerulean,fill=deepcerulean,circle,inner sep=.7pt] (1-8-1) at (2,-1) {};

    \node[draw=deepcerulean,fill=deepcerulean,circle,inner sep=.7pt] (1-9-1) at (2.75,0) {};
    \node[draw=deepcerulean,fill=deepcerulean,circle,inner sep=.7pt] (1-9-2) at (2.75,-1) {};
    \node[draw=deepcerulean,fill=deepcerulean,circle,inner sep=.7pt] (1-9-3) at (2.75,-2) {};

    \node[draw=deepcerulean,fill=deepcerulean,circle,inner sep=.7pt] (2-1-1) at (2.2,-1) {};
    \node[draw=deepcerulean,fill=deepcerulean,circle,inner sep=.7pt] (2-1-2) at (2.2,-2) {};

    \node[draw=deepcerulean,fill=deepcerulean,circle,inner sep=.7pt] (2-2-1) at (2.4,0) {};
    \node[draw=deepcerulean,fill=deepcerulean,circle,inner sep=.7pt] (2-2-2) at (2.4,-1) {};

    \node[draw=deepcerulean,fill=deepcerulean,circle,inner sep=.7pt] (3-1-1) at (-2,2) {};
    \node[draw=deepcerulean,fill=deepcerulean,circle,inner sep=.7pt] (3-1-2) at (-2,1.4) {};
    \node[draw=deepcerulean,fill=deepcerulean,circle,inner sep=.7pt] (3-1-3) at (-2,.3) {};
    \node[draw=deepcerulean,fill=deepcerulean,circle,inner sep=.7pt] (3-1-4) at (-2,-.3) {};
    \node[draw=deepcerulean,fill=deepcerulean,circle,inner sep=.7pt] (3-1-5) at (-2,-1.4) {};
    \node[draw=deepcerulean,fill=deepcerulean,circle,inner sep=.7pt] (3-1-6) at (-2,-2) {};

    \node[draw=deepcerulean,fill=deepcerulean,circle,inner sep=.7pt] (3-2-1) at (0,0) {};

    \node[draw=deepcerulean,fill=deepcerulean,circle,inner sep=.7pt] (3-3-1) at (1,1) {};

    \node[draw=deepcerulean,fill=deepcerulean,circle,inner sep=.7pt] (3-4-1) at (2.4,-2) {};

    \draw[deepcerulean,thick] (1-1-2) to (1-3-4);
    \draw[deepcerulean,thick] (1-3-4) to (1-5-2);
    \draw[deepcerulean,thick] (1-5-2) to (3-3-1);

    \draw[deepcerulean,thick] (-2.125,0.5) to (1-1-2);
    \draw[deepcerulean,thick] (3-3-1) to (2.5,0.25);
    \draw [tenn,thick] (-2.25,0.5) -- (-2.125,0.5);

    \node[draw=tenn,fill=tenn,circle,inner sep=1pt] at (-2.25,0.5) {};
    \node[draw=tenn,fill=tenn,circle,inner sep=1pt] at (2.5,0.25) {};

    \end{tikzpicture}
    }
    \resizebox{!}{0.42\textwidth}{%
        \begin{tikzpicture}
    \fill [pattern=north east lines, pattern color=gray] (-3, -2.65) -- (-3, 2.65) -- (3,2.65) -- (3,-2.65) -- cycle;
    \draw [black, fill=white] (-2.75, -2.4) -- (-2.75, 2.4) -- (2.75,2.4) -- (2.75,-2.4) -- cycle;

    \fill [draw=black, pattern=north east lines, pattern color=gray] (-2, -1) -- (-2,-0.7) -- (-0.3,-0.7) -- (-0.3,0.7) -- (-2,0.7) -- (-2, 1) -- (0,1) -- (0,-1) -- cycle;

    \fill [draw=black, pattern=north east lines, pattern color=gray] (-1.6,-0.3) -- (-1,-0.3) -- (-1,0.3) -- (-1.6,0.3) -- cycle;

    \fill [draw=black, pattern=north east lines, pattern color=gray] (-1,1.4) -- (-1,2) -- (1,2) -- (1, 1.4) -- cycle;

    \fill [draw=black, pattern=north east lines, pattern color=gray] (-1,-1.4) -- (-1,-2) -- (2,-2) -- (2,-1.7) -- (2.2,-1.7) -- (2.2,-1.3) -- (2,-1.3) -- (2,-.7) -- (2.4,-.7) -- (2.4,-.3) -- (2,-.3) -- (2, 0) -- (1.4,0) -- (1.4,-1.4) -- (.8,-1.4) -- (.8,-.4) -- (.4,-.4) -- (.4,-1.4) -- cycle;



    \draw [starship,thick] (-2,-2.4) -- (-2,2.4); 
    \draw [starship,thick] (-1.6,-0.7) -- (-1.6,0.7);
    \draw [starship,thick] (-1,1) -- (-1,2.4); 
    \draw [starship,thick] (-1,-2.4) -- (-1,-1);


    \draw [starship,thick] (-2,1) -- (-2,-2.4);
    \draw [starship,thick] (0,-1.4) -- (0,1.4);
    \draw [starship,thick] (1,-1.4) -- (1,2.4);

    \draw [starship,thick] (-1,-0.7) -- (-1,0.7);
    \draw [starship,thick] (1.4,-2.4) -- (1.4,-2);

    \draw [starship,thick] (1,0) -- (2.75,0);
    \draw [starship,thick] (-1,0) -- (-0.3,0);
    \draw [starship,thick] (2,-1) -- (2.75,-1);


    \draw [starship,thick] (0,-.4) -- (1,-.4);

    \draw [starship,thick] (-1.6,.3) -- (-.3,.3);
    \draw [starship,thick] (-1.6,-.3) -- (-.3,-.3);
    \draw [starship,thick] (1,0) -- (2.75,0);
    \draw [starship,thick] (-1,-2) -- (2.75,-2);


    \node[draw=deepcerulean,fill=deepcerulean,circle,inner sep=.7pt] (1-1-1) at (-2,2.4) {};
    \node[draw=deepcerulean,fill=deepcerulean,circle,inner sep=.7pt] (1-1-2) at (-2,1) {};
    \node[draw=deepcerulean,fill=deepcerulean,circle,inner sep=.7pt] (1-1-3) at (-2,.7) {};
    \node[draw=deepcerulean,fill=deepcerulean,circle,inner sep=.7pt] (1-1-4) at (-2,-.7) {};
    \node[draw=deepcerulean,fill=deepcerulean,circle,inner sep=.7pt] (1-1-5) at (-2,-1) {};
    \node[draw=deepcerulean,fill=deepcerulean,circle,inner sep=.7pt] (1-1-6) at (-2,-2.4) {};

    \node[draw=deepcerulean,fill=deepcerulean,circle,inner sep=.7pt] (1-2-1) at (-1.6,.7) {};
    \node[draw=deepcerulean,fill=deepcerulean,circle,inner sep=.7pt] (1-2-2) at (-1.6,.3) {};
    \node[draw=deepcerulean,fill=deepcerulean,circle,inner sep=.7pt] (1-2-3) at (-1.6,-.3) {};
    \node[draw=deepcerulean,fill=deepcerulean,circle,inner sep=.7pt] (1-2-4) at (-1.6,-.7) {};
    
    \node[draw=deepcerulean,fill=deepcerulean,circle,inner sep=.7pt] (1-3-1) at (-1,2.4) {};
    \node[draw=deepcerulean,fill=deepcerulean,circle,inner sep=.7pt] (1-3-2) at (-1,2) {};
    \node[draw=deepcerulean,fill=deepcerulean,circle,inner sep=.7pt] (1-3-3) at (-1,1.4) {};
    \node[draw=deepcerulean,fill=deepcerulean,circle,inner sep=.7pt] (1-3-4) at (-1,1) {};
    \node[draw=deepcerulean,fill=deepcerulean,circle,inner sep=.7pt] (1-3-5) at (-1,.7) {};
    \node[draw=deepcerulean,fill=deepcerulean,circle,inner sep=.7pt] (1-3-6) at (-1,.3) {};
    \node[draw=deepcerulean,fill=deepcerulean,circle,inner sep=.7pt] (1-3-7) at (-1,0) {};
    \node[draw=deepcerulean,fill=deepcerulean,circle,inner sep=.7pt] (1-3-8) at (-1,-.3) {};
    \node[draw=deepcerulean,fill=deepcerulean,circle,inner sep=.7pt] (1-3-9) at (-1,-.7) {};
    \node[draw=deepcerulean,fill=deepcerulean,circle,inner sep=.7pt] (1-3-10) at (-1,-1) {};
    \node[draw=deepcerulean,fill=deepcerulean,circle,inner sep=.7pt] (1-3-11) at (-1,-1.4) {};
    \node[draw=deepcerulean,fill=deepcerulean,circle,inner sep=.7pt] (1-3-12) at (-1,-2) {};
    \node[draw=deepcerulean,fill=deepcerulean,circle,inner sep=.7pt] (1-3-13) at (-1,-2.4) {};

    \node[draw=deepcerulean,fill=deepcerulean,circle,inner sep=.7pt] (1-4-1) at (-.3,.3) {};
    \node[draw=deepcerulean,fill=deepcerulean,circle,inner sep=.7pt] (1-4-2) at (-.3,0) {};
    \node[draw=deepcerulean,fill=deepcerulean,circle,inner sep=.7pt] (1-4-3) at (-.3,-.3) {};

    \node[draw=deepcerulean,fill=deepcerulean,circle,inner sep=.7pt] (1-5-1) at (0,1.4) {};
    \node[draw=deepcerulean,fill=deepcerulean,circle,inner sep=.7pt] (1-5-2) at (0,1) {};
    \node[draw=deepcerulean,fill=deepcerulean,circle,inner sep=.7pt] (1-5-3) at (0,-.4) {};
    \node[draw=deepcerulean,fill=deepcerulean,circle,inner sep=.7pt] (1-5-4) at (0,-1) {};
    \node[draw=deepcerulean,fill=deepcerulean,circle,inner sep=.7pt] (1-5-5) at (0,-1.4) {};

    \node[draw=deepcerulean,fill=deepcerulean,circle,inner sep=.7pt] (1-6-1) at (1,2.4) {};
    \node[draw=deepcerulean,fill=deepcerulean,circle,inner sep=.7pt] (1-6-2) at (1,2) {};
    \node[draw=deepcerulean,fill=deepcerulean,circle,inner sep=.7pt] (1-6-3) at (1,1.4) {};
    \node[draw=deepcerulean,fill=deepcerulean,circle,inner sep=.7pt] (1-6-4) at (1,0) {};
    \node[draw=deepcerulean,fill=deepcerulean,circle,inner sep=.7pt] (1-6-5) at (1,-.4) {};
    \node[draw=deepcerulean,fill=deepcerulean,circle,inner sep=.7pt] (1-6-6) at (1,-1.4) {};

    \node[draw=deepcerulean,fill=deepcerulean,circle,inner sep=.7pt] (1-7-1) at (1.4,-2) {};
    \node[draw=deepcerulean,fill=deepcerulean,circle,inner sep=.7pt] (1-7-2) at (1.4,-2.4) {};
    
    \node[draw=deepcerulean,fill=deepcerulean,circle,inner sep=.7pt] (1-8-1) at (2,-1) {};

    \node[draw=deepcerulean,fill=deepcerulean,circle,inner sep=.7pt] (1-9-1) at (2.75,0) {};
    \node[draw=deepcerulean,fill=deepcerulean,circle,inner sep=.7pt] (1-9-2) at (2.75,-1) {};
    \node[draw=deepcerulean,fill=deepcerulean,circle,inner sep=.7pt] (1-9-3) at (2.75,-2) {};

    \node[draw=deepcerulean,fill=deepcerulean,circle,inner sep=.7pt] (2-1-1) at (2.2,-1) {};
    \node[draw=deepcerulean,fill=deepcerulean,circle,inner sep=.7pt] (2-1-2) at (2.2,-2) {};

    \node[draw=deepcerulean,fill=deepcerulean,circle,inner sep=.7pt] (2-2-1) at (2.4,0) {};
    \node[draw=deepcerulean,fill=deepcerulean,circle,inner sep=.7pt] (2-2-2) at (2.4,-1) {};

    \node[draw=deepcerulean,fill=deepcerulean,circle,inner sep=.7pt] (3-1-1) at (-2,2) {};
    \node[draw=deepcerulean,fill=deepcerulean,circle,inner sep=.7pt] (3-1-2) at (-2,1.4) {};
    \node[draw=deepcerulean,fill=deepcerulean,circle,inner sep=.7pt] (3-1-3) at (-2,.3) {};
    \node[draw=deepcerulean,fill=deepcerulean,circle,inner sep=.7pt] (3-1-4) at (-2,-.3) {};
    \node[draw=deepcerulean,fill=deepcerulean,circle,inner sep=.7pt] (3-1-5) at (-2,-1.4) {};
    \node[draw=deepcerulean,fill=deepcerulean,circle,inner sep=.7pt] (3-1-6) at (-2,-2) {};

    \node[draw=deepcerulean,fill=deepcerulean,circle,inner sep=.7pt] (3-2-1) at (0,0) {};

    \node[draw=deepcerulean,fill=deepcerulean,circle,inner sep=.7pt] (3-3-1) at (1,1) {};

    \node[draw=deepcerulean,fill=deepcerulean,circle,inner sep=.7pt] (3-4-1) at (2.4,-2) {};

    \draw[tenn,thick] (1-1-2) to (1-3-4);
    \draw[tenn,thick] (1-3-4) to (1-5-2);
    \draw[tenn,thick] (1-5-2) to (3-3-1);

    \draw[tenn,thick] (-2,0.5) to (1-1-2);
    \draw[tenn,thick] (-2.25,0.5) -- (-2,0.5);

    \draw[tenn,thick] (3-3-1) to (2.5,1);
    \draw[tenn,thick] (2.5,1) to (2.5,0.25);

    \node[draw=tenn,fill=tenn,circle,inner sep=1pt] at (-2.25,0.5) {};
    \node[draw=tenn,fill=tenn,circle,inner sep=1pt] at (2.5,0.25) {};

    \end{tikzpicture}
    }
    \caption{When source and target node are in different region (top left), they are added to the landmark graph, enabling each node to locally find a shortest path from the source node to the target node in the landmark graph (top right). After the source node performs this computation, it forwards the packet towards the suggested region boundary (bottom left). This is repeated until the packet the target region, where we it can be forwarded using the algorithm from \cite{CCFHKSSS21} (bottom right).}
    \label{fig:gallery_routingdifficult}
\end{figure}
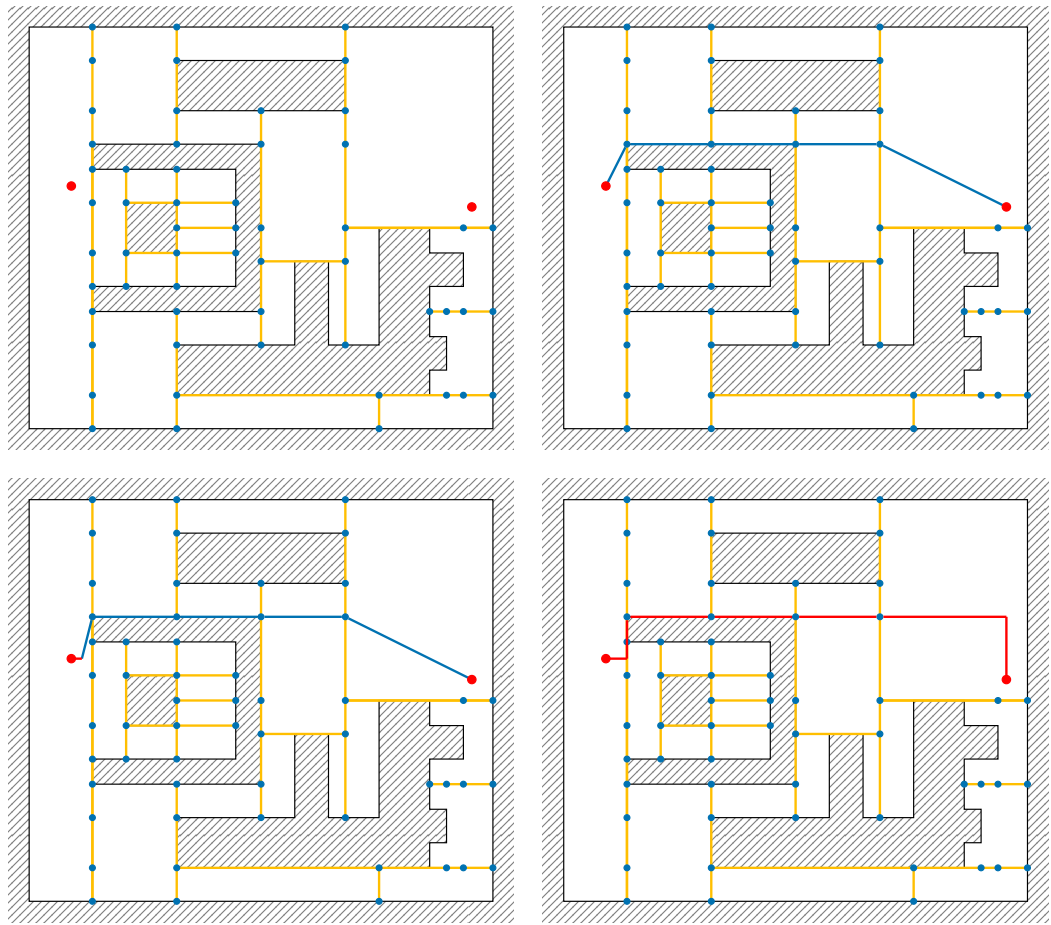

\section{Properties of Grid Graphs}
\label{sec:grid_properties}

In this section we prove some technical properties of grid graphs that we exploit throughout the paper. Here, our grid graphs are denoted with $R=(V_R,E_R)$ to signify that we might deal with regions.
First, we provide some useful definitions and a lemma.

\begin{definition}[Bottleneck]
    \label{def:bottleneck}
    Given a set of paths then any node on all of those paths is a bottleneck with respect to this set.
\end{definition}

\begin{definition}[Monotonous Path]
    \label{def:monotonous_path}
    Given a grid graph $\Gamma = (V_\Gamma, E_\Gamma)$, a path $\Pi \subseteq E_\Gamma$ from a node $u$ to $v$ is monotonous if the sequence of edges from $u$ to $v$ go into at most two cardinal directions: at most one horizontal (east, west) and at most one vertical (north, south).
\end{definition}

\begin{lemma}
    \label{lem:monotonous_bottleneck}
    Let $R$ be a simple grid graph and $u,v \in R$ be two nodes. If there exists a shortest $uv$-path that is not monotonous, then the set of shortest $uv$-paths have a bottleneck $w \neq u,v$ (see Def's \ref{def:monotonous_path}, \ref{def:bottleneck}).
\end{lemma}

\begin{proof}
    Let $\Pi$ be a shortest $uv$-path that is not monotonous. Then $\Pi$ contains a subpath $\tilde\Pi \subseteq \Pi$ that is a ``turn'', i.e., consists of incident edges going in three cardinal directions in an ordered fashion, w.l.o.g., first north, then east, then south. The turn $\tilde\Pi$ must intersect the boundary of $R$ at some node $w \neq u,v$ at the most northern extent of $\tilde\Pi$ with a hole to the south of $w$, because if there where no such node the turn $\tilde\Pi$ could be cut short, which is a contradiction to the presumption that $\Pi$ is shortest. 
    
    Consider a shortest $uv$-path $\Pi' \neq \Pi$ that does not contain $w$. From the viewpoint of travelling from $u$ to $v$, let $s,t$ be the last intersection of $\Pi$ with $\Pi'$ before it reaches $w$ and the earliest after visiting $w$, respectively. Clearly $s,t$ exist as both paths start and end at $u,v$, respectively. Consider the region $R' \subseteq R$ enclosed by the sub paths $\Pi_{st}$ and $\Pi'_{st}$ from $s$ to $t$, excluding the nodes on the paths.
    $R'$ cannot contain a hole since otherwise $R$ would not be simple. Furthermore, $R'$ must be located to the north of $w$ since $w$ has a hole to the south. Since $\Pi'_{st}$ passes above $w$ it must have a ``turn'' (a consecutive sequence of edges that go north, then east, then south). This turn can be shortened by short-cutting through this region $R'$ thus making $\Pi'$ shorter, which is a contradiction that was shortest in the first place.
\end{proof}

\begin{definition}[Vertical and horizontal distance]
    For some path $\Pi$ in a grid graph, let $|\Pi|_x$ and $|\Pi|_y$ be the number of horizontal and vertical edges in $\Pi$, respectively. For two nodes $u,v \in \Gamma$, let $d_x(u,v) = \min_{uv \text{-path } \Pi} |\Pi|_x$ and $d_y(u,v) = \min_{uv \text{-path } \Pi} |\Pi|_y$ be the minimum horizontal or vertical distance between $u$ and $v$, respectively. Note that the triangle inequality holds for $d_x$, $d_y$.
\end{definition}

\begin{lemma}
    \label{lem:length_of_shortest_paths_in_simple_regions}
    Let $R$ be a simple grid graph and $u,v \in R$ be two nodes. Then for any two \emph{shortest} $uv$-paths $\Pi,\Pi'$ we have that $|\Pi|_x = |\Pi'|_x = d_x(u,v)$ and $|\Pi|_y = |\Pi'|_y = d_y(u,v)$.
\end{lemma}

\begin{proof}
    It suffices to prove the claim for each subpath of $\Pi$ and $\Pi'$ between two nodes where both paths intersect with no other intersection node strictly between those two nodes. Therefore, we can assume that the nodes on $\Pi$ and $\Pi'$ intersect only in $u$ and $v$. In particular, this implies that there can not be a bottleneck $w \neq u,v$.\jw{bottleneck wrt the set of $uv$-paths} Therefore the contrapositive of Lemma \ref{lem:monotonous_bottleneck} implies that $\Pi$ and $\Pi'$ are monotonous thus the claim is clearly true.\jw{a little handwavey; should be fine}\jw{rework with the $\Pi_{uw}$ notation?}
\end{proof}

\begin{lemma}\label{lem:grid_cycles}
    Given a grid graph $\Gamma$, any cycle $C$ in $\Gamma$ has even length.
\end{lemma}

\begin{proof}
    By identifying the nodes in a grid graph with coordinates, we can see that all edges connect two nodes where exactly one coordinate differs by exactly $1$. The total sum of changes in both coordinates in a cycle must be $0$ however. Therefore each edge of $C$ increasing a coordinate can be paired with one decreasing the same coordinate and the cycle must have even length. 
\end{proof}

\begin{lemma}\label{lem:neighbor_distances}
Given a grid graph $\Gamma=(V_\Gamma,E_\Gamma)$, let $v, w\in V_\Gamma$, then for each neighbor $v'$ of $v$ either $d(v',w)=d(v,w)+1$ or $d(v',w)=d(v,w)-1$.
\end{lemma}

\begin{proof}
    It suffices to show $d(v,w)=d(v',w)$ never holds. Assume $d(v,w)=d(v',w)$. Consider any shortest $vw$-path and any shortest $v'w$-path.
    We denote the node closest to $v$ that is part of both paths by $x$. By definition $d(v,x)+d(x,w)=d(v',x)+d(x,w)$ and hence $d(v,x)=d(v',x)$. Therefore the edge $\{v,v'\}$, the corresponding shortest $vx$-path and the corresponding shortest $v'x$-path induce a cycle of length $2d(v,x)+1$. This yields a contradiction to \autoref{lem:grid_cycles}.
\end{proof}


\begin{fact}
    \label{fct:compose_monotonous_paths}
    Let $R$ be a simple grid graph, let $u,v$ be two nodes in $R$, and let and let $\Pi,\Pi'$t be wo monotonous $uv$-paths. Let $w$ and $z$ be two nodes on $\Pi$ and $\Pi'$, respectively. Let $\widetilde \Pi$ be a $wz$-path that uses the same two cardinal directions as $\Pi$. Then $\Pi_{uw} \circ \widetilde \Pi_{wz} \circ \Pi'_{zv}$ is monotonous and therefore shortest.
\end{fact}

\begin{lemma}
    \label{lem:compose_paths_with_line}
    Let $R$ be a simple grid graph and $u,v \in R$ and two shortest $uv$-paths $\Pi,\Pi'$.\jw{the last part reads odd to me. put $\Pi,\Pi'$ in front?} Let $w$ and $z$ be two nodes on $\Pi$ and $\Pi'$, respectively with $d_x(w,z) = 0$ or $d_y(w,z) = 0$. Let $\overline{wz}$ be the shortest $wz$-path (i.e., a horizontal or vertical path). Then either $\Pi_{uw} \circ \overline{wz} \circ \Pi'_{zv}$ or $\Pi'_{uz} \circ \overline{zw} \circ \Pi_{wv}$ is a shortest $u,v$-path.\jw{$w=z$ might be a weird edge case regarding the 'either or'}
\end{lemma}

\begin{proof}
    Assume that $\Pi,\Pi'$ only intersect in $u,v$, otherwise redefine $u$ and $v$ as the closest points before and after $w$ where $\Pi,\Pi'$ intersect, then it suffices to prove the claim for the sections $\Pi_{uv},\Pi'_{uv}$. In particular, this means that $\Pi$ and $\Pi'$ do not have any bottleneck node $\neq u,v$\jw{bottleneck wrt the set of $uv$-paths} and therefore must be monotonous (contrapositive of Lemma \ref{lem:monotonous_bottleneck}).
    
    W.l.o.g.\ $d_y(w,z) = 0$, i.e., $w,z$ are on the same horizontal portal. Following $\Pi$ and $\Pi'$ from $u$ to $v$ the paths must use at least two cardinal directions say south and west (w.l.o.g., since one must be horizontal and one vertical and due to symmetry) as otherwise $\Pi = \Pi'$. Say the $w$ is to the west of $z$. Then the path $\Pi_{uw} \circ \overline{wz} \circ \Pi'_{zv}$ is shortest by Fact \ref{fct:compose_monotonous_paths}.
\end{proof}



\begin{lemma}\label{lem:convex_union}
Given a simple grid graph $R$, let $u, v\in R$, let $\mathcal{U}_{uv}$ be the set of nodes that are on a shortest $uv$-path. Then for each $x,y\in \mathcal{U}_{uv}$ every node on a shortest $xy$-path is in $\mathcal U_{uv}$.
\end{lemma}

\begin{proof}
    
    We will prove the lemma by contradiction.
    Assume a shortest $xy$-path $\Pi$ leaves $\mathcal{U}_{vw}$\ps{what's that?}\jw{analogous to $\mathcal{U}_{uv}$. Should I be more specific?} directly after some node $x'$ and joins $\mathcal{U}_{vw}$ again directly before some node $y'$, while no nodes between $x'$ and $y'$ on $\Pi$ are in $\mathcal{U}_{vw}$.
    Let now $\Pi'$ be the shortest $x'y'$-path along the border of $\mathcal{U}_{vw}$.
    Note that the border of $\mathcal{U}_{vw}$ consists of two shortest $vw$-paths. 
    If $x'$ and $y'$ are on the same path $\Pi'$, $\Pi_{xx'}\circ\Pi'_{x'y'}\circ\Pi_{y'y}$ is shorter than $\Pi$, which contradicts $\Pi$ being a shortest path.
    If, however, $x'$ is on the path $\Pi'$, while $y'$ is on the other path $\Pi''$, then $\Pi$ eventually has to get around $v$ or $w$ depending on its direction. 
    \autoref{lem:length_of_shortest_paths_in_simple_regions} yields that $\Pi$ could be shortened.
\end{proof}




\begin{lemma}
    \label{lem:shortest_path_property}
    Given a simple region $R$, let $a,b,v\in R$, $\Pi$ be the shortest $ab$-path closest to $v$, i.e.,
    \[
        \Pi\in\argmin_{\substack{\Pi'\text{ shortest }ab\text{-path}}} \; \min_{u\in \Pi'} d_R(u,v)
    \]
    chosen arbitrarily and $v'$ be the unique (\autoref{lem:neighbor_distances}) closest point to $v$ on $\Pi$, i.e.,
    \[
        v':=\argmin_{u\in \Pi} d_R(u,v)\text{.}
    \]
    Then every shortest $vv'a$-path is a shortest $va$-path and every shortest $vv'b$-path is a shortest $vb$-path.
\end{lemma}
\begin{proof}
    In the following, we refer to the nodes of $\Pi$ with $p_1,...,p_\ell$, where $\ell:=\lvert \Pi\rvert$. 

    We will prove the statement for any shortest $vv'a$-path. The shortest $vv'b$-case is analogous. 

    Let $\mathcal{U}$ denote the set containing all nodes on any shortest $av$-path, i.e. $\mathcal{U}:=\bigcup_{\tilde{\Pi} \text{ shortest $av$-path}} \tilde{\Pi}$. 
    
    We start by proving by contradiction that $\Pi$ leaves $\mathcal{U}$ only once, i.e., there is a unique point $p_i\in \Pi$ s.t. $p_i\in \mathcal{U}$, but $p_j\not\in \mathcal{U}$ for all $j>i$.
    Assume $\Pi$ would leave $\mathcal{U}$ more than once and let $p_i$ be the first leaving node. Then there would be a reentering node $p_j$, $j>i$ s.t. $p_{j-1}\not\in \mathcal{U}$ but $p_j\in \mathcal{U}$. According to \autoref{lem:convex_union}, all shortest $p_ip_j$-paths are inside $\mathcal{U}$. Therefore, $\Pi$ could be shortened by replacing its segment $\Pi_{p_ip_j}$ with such a shortest path. This is a contradiction to $\Pi$ being a shortest path.

    We now call the unique point where $\Pi$ leaves $\mathcal U$ $\overline{v'}$. Our next goal is to show $\overline{v'}=v'$, i.e, to show that $\Pi$ leaves $\mathcal U$ at its closest point to $v$. To this end, we show that $\overline{v'}$ can not be closer to $a$ than $v'$ and can not be further away from $a$ than $v'$. We start with the former.

    Assume $d(\overline{v'},a)<d(v',a)$. According to \autoref{lem:neighbor_distances} and since $\Pi$ is a shortest path, the next node in $\Pi$ after $\overline{v'}$
    must be closer to $v'$ than $\overline{v'}$. As $v'$ is the closest node to $v$ on $\Pi$ it must also be closer to $v$. Hence it is part of a shortest $\overline{v'}v$-path and, since $\overline{v'}\in \mathcal U$, it is part of $\mathcal U$. This is a contradiction to the definition of $\overline{v'}$.
    
    Assume $d(\overline{v'},a)>d(v',a)$. According to \autoref{lem:neighbor_distances} and by definition of $v'$ the node after $v'$ along $\Pi$ increases the distance to $v$. Therefore it can not be part of $\mathcal U$, which contradicts the definition of $\overline{v'}$.

    Hence $d(\overline{v'},a)=d(v',a)$. Since both $\overline{v'}$ and $v'$ are on the shortest $ab$-path $\Pi$, this yields $\overline{v'}=v'$. By definition of $\overline{v'}$ we conclude $v'\in \mathcal{U}$, i.e., $v'$ is contained in the set of shortest paths from $a$ to $v$. This yields the lemma.

\end{proof}

    \begin{lemma}\label{lem:ncc0_to_ncc_in_grids}
        It takes $\bigO(\log n)$ rounds to establish an overlay path graph in a grid graph. After an additional $\bigO(\log n)$ rounds, the overlay path graph can be transformed into an $\lfloor\log n\rfloor$-dimensional overlay butterfly network. Both transformations are deterministic.
    \end{lemma}
    \begin{proof}
        Similarly to the construction of \cite{Coy2022}, we start with the vertical Portals of the grid graph and have each portal's bottommost node add an edge to the left if able and one to the right if able. As the resulting graph may have cycles, we remove horizontal edges again until it is a tree. Specifically, each node that is on the boundary of any inner hole and has no bottom neighbor marks itself as a node on top of the corresponding hole (\autoref{lem:hole_ids} allows us to differentiate holes). Using pointer-jumping on the hole boundaries, we can find the leftmost of these nodes. We remove the left incident horizontal edge for that leftmost node. It remains to show that the resulting structure is indeed a tree. To this end, we can consider cutting an edge as connecting the hole to the next inner hole above it or to the outer hole. Note that the nodes of each inner hole's boundary still remain connected among themselves and the topmost node of each boundary is connected via the vertical portal it is part of to another inner hole or the outer boundary. Applying this argument inductively, each node on any hole boundary is connected to the outer hole boundary. As the same statement clearly holds for nodes not on hole boundaries, we can conclude that the graph remains connected. Therefore the result of this transformation is simple and connected and by the arguments in \cite{Coy2022} it is a tree. The transformation took $\mathcal O(\log n)$ time in total.
        
        To transform this tree into a path, we use the euler tour technique from \ref{lem:euler_tour}. Note that this technique has each node introduce a constant number of virtual nodes in our setting. To get rid of them, we have each node mark its first virtual node and have the nodes perform pointer-jumping to reduce the distance between marked nodes to $\bigO(\log n)$ allowing us to connect them after an in an additional $\bigO(\log n)$ rounds.
        
        Finally, we employ pointer-jumping on the resulting path graph again to establish hypercubic connections. Using the butterfly simulation technique from \cite[Section 2.2]{AGG+18}, we can use the hypercubic connections to have the nodes simulate a $\lfloor\log n\rfloor$-dimensional butterfly network.
    \end{proof}

\section{Subroutines}
\label{app:subroutines}

In this section we will establish a few basic subroutines for the \HYBRID model that we will frequently use for our constructions for setting up the routing scheme.

\subsection{Pointer Jumping, Broadcast and Aggregation}
\label{appsec:pointer_jumping}
We start by explaining an important subroutine in the \HYBRID model known as pointer-jumping, which works as follows (these procedures were extensively used, e.g., in \cite{Feldmann2020}).
The input is a path or cycle graph of length $n$. In step $1$, each node introduces it left neighbor (sends the ID) to its right neighbor and vice versa. Given that in step $i-1$ each node knows the node at distance $2^{i-1}$ to its left and right, they can introduce these to one another using the global network, so that now each node knows the nodes with distance $2^i$ to their left and right. 
We can think of this as creating a shortcut edges via the the global mode of communication. After $O(\log n)$ rounds the resulting structure has diameter and maximum degree of $O(\log n)$. A more formal algorithm is provided in the following


\begin{algorithm}[H]
	\caption{\code{pointer-jumping}\Comment{\textit{Executed by each $v$ on a path or cycle graph $G$}}}
	\label{alg:pointer_jumping}
	\begin{algorithmic}
	    \State $\ell_{v,1}, r_{v,1} \gets$ identifier of left and right neighbor of $v$ in $G$ or $\bot$ if it does not exist
	    \For {$i \gets 2$ to $\lceil\log n \rceil$ rounds}
	        \State $v$ sends $\ell_{v,i-1}$ to $r_{v,i-1}$ and vice versa via the global network
	        \State $\ell_{v,i}, r_{v,i} \gets $ identifiers (or $\bot$) received from nodes $\ell_{v,i-1}, r_{v,i-1}$ this round
	    \EndFor
	\end{algorithmic}
\end{algorithm}

The pointer-jumping technique allows us to reduce the diameter of path graphs and cycle graphs by augmenting it with virtual edges.

\begin{lemma}[Pointer Jumping Structure]
    \label{lem:pointer_jumping_structure}
    Let $G = (V,E)$ be a path graph or a cycle graph of length at most $n$. Algorithm \ref{alg:pointer_jumping} takes $\bigO(\log n)$ rounds. Further, the pointer jumping structure $G' = (V,E')$ with $E' := E \cup \{  \{v,\ell_{v,i}\},\{v,r_{v,i}\} \mid i \in \{1, \dots, \lceil\log n \rceil\} \}$ has degree and diameter $\bigO(\log n)$. 
\end{lemma}

\begin{proof}
    The running time is clear, and since in each loop cycle at most two new connections are established for each node, the claim about the degree follows. Let $u,v \in V$ be at distance $d$ in $G$. Consider the largest integer $i \geq 0$ such that $2^i \leq d$. Then there is an edge in $G'$ from $u$ to some node $w$ with distance at most $d-2^i < 2^{i-1}$ to $u$, i.e., the remaining distance to $v$ in $G'$ at least halves with each such a hop, thus the claim follows from $d \leq n$.
\end{proof}


This structure can be used to broadcast and aggregate messages.

\begin{lemma}
\label{lem:broadcast_and_aggregation}
    On path and cycle graphs $G= (V,E)$ with length at most $n$, an $O(\log n)$ bit message can be broadcast from one node to all others in the structure in $O(\log n)$ rounds. 
    Further, if each node in the structure has a $O(\log n)$ bit value, all nodes in the pointer jumping structure can agree on which value is the maximum or minimum in $O(\log n)$ rounds.
\end{lemma}

\begin{proof}
    We start with the broadcasting protocol. First we set up the pointer doubling structure $G'=(V,E')$ on $G$ (see Lemma \ref{lem:pointer_jumping_structure}). For $i = \lceil n \rceil$ down to $i=0$ do the following. Each node that already knows the message sends it to its neighbors in $G'$ that are at distance $2^i$ in $G$. Let $s\in V$ be the source of the message and consider some $v \in V$. Then there is a path in $G'$ from $s$ to $v$ that uses at most $\lceil \log n \rceil$ edges of strictly decreasing length (in $G$). As the procedure broadcasts the message on all such paths, $v$ will receive the message.
    For computing the minimum (or maximum) value, we use the same procedure as above where each node only forwards the minimum (or maximum) value encountered so far.
\end{proof}



By constructing an Euler tour on a tree, we can replicate the pointer jumping structure for path graphs on trees.

\begin{lemma}[Euler Tour Technique for constant degrees, cf., \cite{Feldmann2020}]\label{lem:euler_tour}
    Let $T = (V,E)$ be a rooted, constant degree tree. We can construct an Euler tour on $T$, to obtain a path graph containing a constant number of virtual nodes for each original node in $\bigO(1)$ rounds.
\end{lemma}

\begin{proof}[Proof sketch]
    Each node of $T$ starts by creating a virtual node for each of its neighbors. After ordering their neighbors locally, each node can decide in one round how its virtual nodes would be connected in a depth first search traversal and establishes these connections. As a result, we obtain a path graph and as each node has a constant degree by our assumption, the path graph has length $\bigO(n)$ and each node is responsible for simulating $\bigO(1)$ virtual nodes.
\end{proof}


The lemma above immediately implies the following corollary.

\begin{corollary}
    \label{cor:broadcast_and_aggregation_trees}
    On trees $T= (V,E)$ with constant degree and at most $n$ nodes, broadcasting an $\bigO(\log n)$ bit message or aggregating the minimum or maximum value in $T$ can be done in $\bigO(\log n)$ rounds. The same is true for simple regions, since we can construct the portal tree (see \autoref{def:portal_graph}) of a simple region in $\bigO(\log n)$ rounds.
\end{corollary}


\subsection{Identifying Portals and Holes}
\label{appsec:identifying_holes}

We can use the above techniques in the \HYBRID model on grid graphs to assign a unique identifier to each hole in the grid graph.\ps{We need a similar thing also for portals, so that these are aware of their portal ID}

\begin{lemma}\label{lem:hole_ids}
    Given a grid graph $\Gamma$ with holes $\mathcal H$. For each $H \in \mathcal H$, a unique ID of $H$ can be broadcast to each node on the boundary of $H$ (Def.\ \ref{def:hole-boundary}).
\end{lemma}

From the definition, each node can locally identify whether it is incident to a hole in $O(1)$ rounds. Next, nodes identify their neighbours on a hole boundary in the following way.

\begin{itemize}
    \item Consider a cycle formed by the orthogonally and diagonally adjacent points to a grid node $v$ (whether or not these points actually correspond to a grid node).
    \item Each maximally connected series of \emph{missing} nodes on this cycle corresponds to a hole $H$ (not necessarily uniquely).
    \item For each maximally connected series of missing nodes $h$, consider the two present nodes on either side and call them $l$ and $r$.
    \item The first node on a shortest path from $v$ to $l$ and the first node on a shortest path from $v$ to $r$ are the neighbors of $v$ with respect to the hole which $h$ corresponds to.
\end{itemize}

This procedure, which can be performed locally by each node in $O(1)$, produces a unique cycle of grid nodes for each hole, which we call a \emph{hole boundary}. Note that a node $v$ may be incident to a hole multiple times, but only constantly many times (this fact follows from the fact that $v$ only has constantly many orthogonal and diagonal neighbours). Therefore the length of the hole boundary is at most $O(n)$. In order to manage multiple hole-incidences, a node can simply split itself into a constant number of virtual nodes (assigning new $O(\log n)$ bit identifiers), with only a constant-factor increase in the messages which the node sends and receives.

\begin{lemma}
\label{lem:internal_holes_uniqely_identified}
    We say that a grid-node $v$ is \emph{east-incident} to a hole $H \in \mathcal H$ if the grid-point immediately to the north of $v$ does not correspond to a grid node and is part of $H$ (see Def.\ \ref{def:hole}).
    
    Given a grid graph $\Gamma$ and a set of hole boundaries $H$, each hole is uniquely identified by the minimum (sorted lexicographically) of the nodes which are east-incident to it.
\end{lemma}
\begin{proof}
    No two holes can share the same identifier by this scheme, as the same node cannot be east-incident to two holes. Holes must trivially have at least one node east-incident to them.
\end{proof}

From \Cref{lem:broadcast_and_aggregation}, we can broadcast the coordinates of each grid node and uniquely identify every hole (by Lemma~\ref{lem:internal_holes_uniqely_identified}) in $O(\log n)$ rounds.\jw{broadcast as in broadcast on the hole boundary?}




    



\subsection{Splitting Portals}
\label{appsec:splitting_portals}

In this section we describe the primitives necessary to split $\Gamma$ along portals to obtain a decomposition of a grid graph $\Gamma$ into regions (see Definition \ref{def:region_decomposition}).
This is used used extensively in Section \ref{sec:partitioning_the_graph}, where we aim to obtain a \convex region decomposition in $\bigO(\log n)$ rounds. In this section we use three routines building on top of each other, that are described in the following.

\begin{figure}[h]
\centering
  \begin{subfigure}{0.23\textwidth}
    \includegraphics[page=1,width=\textwidth]{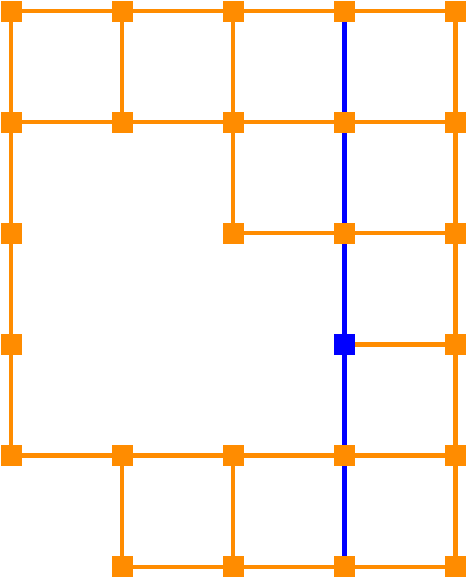}
    \caption{Before split at blue portal $P$ node $v$}
    \label{fig:splitting_portal:a}
  \end{subfigure}
  \hspace{0.1\textwidth}
  \begin{subfigure}{0.23\textwidth}
    \includegraphics[page=2,width=\textwidth]{figures/splitting_portals.pdf}
    \caption{After split at $P$; before split at $v$}
    \label{fig:splitting_portal:b}
  \end{subfigure}
  \hspace{0.1\textwidth}
  \begin{subfigure}{0.23\textwidth}
    \includegraphics[page=3,width=\textwidth]{figures/splitting_portals.pdf}
    \caption{After split at blue node $v$}
    \label{fig:splitting_portal:c}
  \end{subfigure}
  \caption{Illustration of the process by which portals are split at a boundary node on a portal.}
  \label{fig:splitting_portal}
\end{figure}
\jw{iirc we actually use portals to the west of holes. Even though this is independent, mirroring the picture might make sense}

\begin{definition}[Splitting Procedure]
\label{def:splitting_procedure}
\begin{enumerate}
    \item \textit{Splitting $\Gamma$ at a portal $P$}: Assume $P$ is a vertical portal, the other case works symmetrically. Each node on $v \in P$ creates a copy of itself, both of which it simulates; a left copy $v_\ell$ and a right copy $v_r$. We remove left neighbors of right copies and right neighbors of left copies. 
    
    \item \textit{Splitting $\Gamma$ at $P$ and a boundary node $v \in P$ of some hole $H$}: Again, we assume that $P$ is vertical, the other case is analogous. This procedure must specify the hole that $v$ is boundary of in case $v$ is at the boundary of more than one hole. We split $P$ as in the procedure above. Then we create two further copies of the copy of $v$ that is on the side of $H$: $v^\uparrow$ and $v^\downarrow$, which it henceforth both simulates. We define that the copy $v^\uparrow$ does not have a bottom neighbor and $v^\downarrow$ does not have a top neighbor.
    
    \item \textit{Splitting $\Gamma$ at $P$ and at boundary nodes $v_1, \dots, v_\ell \in P$ of holes $\{ H_1, \dots,H_\ell \}$}: We split $P$ as in the first procedure. In general $P$ might have multiple nodes $v_i$ that are boundaries of potentially different holes $H_i$. In that case we repeat the vertical split as described in the second step for each copy of $v_i$ that is to the direction of $H_i$.
\end{enumerate}
\end{definition}

\begin{lemma}
\label{lem:splitting_procedure_overhead}
    The procedures of splitting at portals $P$ and boundary nodes as described above can be coordinated in $O(\log n)$ rounds. The subsequent simulation of node copies has at most constant overhead.
\end{lemma}

\begin{proof}
     In all of our usual applications of this lemma we assume that we have already identified the portal $P$ and the boundary node $v$, and each node is aware of its role. Note that these preconditions are typically either local conditions, or conditions that we can easily establish, using our ability to identify and broadcast on hole boundaries and portals in $\bigO(\log n)$ rounds (when we apply this lemma, we often implicitly use that these tasks are straight forward).
     The subsequent creation of node copies uses only local information, thus can be done $\bigO(1)$ rounds. Note that after a splitting operation, each node has to simulate at most 3 copies, which is why we have only a constant overhead.
\end{proof}

\end{document}